\documentclass[12pt,a4paper]{article}

\usepackage[hidelinks]{hyperref}

\usepackage{afterpage, amsfonts, amsmath, amssymb, amsthm, bm, bbm, caption, color, enumitem, float, graphicx, epstopdf, lipsum, longtable, lscape, mathrsfs, multirow, nccmath, nonfloat, pdflscape, rotating, tocloft}

\usepackage[margin=0.640in]{geometry}

\usepackage[round]{natbib}
\usepackage[doublespacing]{setspace}
\usepackage[table]{xcolor}

\def\be{\begin{equation}}
\def\ee{\end{equation}}
\def\bea{\begin{eqnarray}}
\def\eea{\end{eqnarray}}

\author{}
\title{}

\DeclareMathOperator*{\argmin}{\arg\!\min}
\DeclareMathOperator*{\plim}{p\!\lim}

\newcommand{\normmm}[1]{{\left\vert\kern-0.25ex\left\vert\kern-0.25ex\left\vert #1 \right\vert\kern-0.25ex\right\vert\kern-0.25ex\right\vert}}

\begin{document}
\newcommand\blfootnote[1]{
\begingroup
\renewcommand\thefootnote{}\footnote{#1}
\addtocounter{footnote}{-1}
\endgroup
}

\newtheorem{corollary}{Corollary}
\newtheorem{definition}{Definition}
\newtheorem{lemma}{Lemma}
\newtheorem{proposition}{Proposition}
\newtheorem{remark}{Remark}
\newtheorem{theorem}{Theorem}
\newtheorem{assumption}{Assumption}
\newtheorem{example}{Example}

\numberwithin{corollary}{section}
\numberwithin{definition}{section}
\numberwithin{equation}{section}
\numberwithin{lemma}{section}
\numberwithin{proposition}{section}
\numberwithin{remark}{section}
\numberwithin{theorem}{section}

\allowdisplaybreaks[4]

\begin{titlepage}

\begin{center}
{\large \textbf{On Time-Varying VAR Models: \\Estimation, Testing and Impulse Response Analysis}}\blfootnote{{\it Corresponding author}: Jiti Gao, Department of Econometrics and Business Statistics, Monash University, Caulfield East, Victoria 3145, Australia. Email: \url{Jiti.Gao@monash.edu}. The authors of this paper would like to thank George Athanasopoulos, Rainer Dahlhaus, David Frazier, Oliver Linton, Gael Martin, Peter CB Phillips and Wei Biao Wu for their constructive comments on earlier versions of this paper. Gao acknowledges financial support from the Australian Research Council Discovery Grants Program under Grant Numbers: DP170104421 and DP200102769. Peng also acknowledges the Australian Research Council Discovery Grants Program for its financial support under Grant Number DP210100476.}

\bigskip

{\sc Yayi Yan, Jiti Gao and Bin Peng
\smallskip}

Monash University, Melbourne, Australia

\today

\end{center}

\begin{abstract}

Vector autoregressive (VAR) models are widely used in practical studies, e.g., forecasting, modelling policy transmission mechanism, and measuring connection of economic agents. To better capture the dynamics, this paper introduces a new class of time-varying VAR models in which the coefficients and covariance matrix of the error innovations are allowed to change smoothly over time. Accordingly, we establish a set of theories, including the impulse responses analyses subject to both of the short-run timing and the long-run restrictions, an information criterion to select the optimal lag, and a Wald-type test to determine the constant coefficients. Simulation studies are conducted to evaluate the theoretical findings. Finally, we demonstrate the empirical relevance and usefulness of the proposed methods through an application to the transmission mechanism of U.S. monetary policy.
\bigskip

\noindent{\bf Keywords}: Multivariate Dynamic Time Series; Time-Varying Impulse Response; Testing for Parameter Stability
\medskip

\noindent{\bf JEL Classification}: C14, C32, E52

\end{abstract}
\end{titlepage}

\section{Introduction}\label{Sec1}

Vector autoregressive (VAR) models as well as their extensions are among some of the most popular frameworks for modelling dynamic interactions among multiple variables. These models arise mainly as a response to the ``incredible'' identification conditions embedded in the large-scale macroeconomic models \cite[]{sims1980}. VAR modelling begins with minimal restrictions on the multivariate dynamic models. Gradually armed with identification information, VAR models and their statistical tool-kits like impulse response functions become powerful tools for conducting policy analysis. We refer interested readers to \cite{stock2001vector} for a comprehensive review. Despite the popularity, linear VAR models can always be rejected by data in empirical studies (\citealp{tsay1998testing}). For example, \cite{stock2016dynamic} point out, ``\textit{changes associated with the Great Moderation go beyond reduction in variances to include changes in dynamics and reduction in predictability}''.

To go beyond linear VAR models, various parametric time-varying VAR models have been proposed (e.g., \citealp{tsay1998testing}, \citealp{sims2006}, and the references therein) in order to allow for certain changes in economic relationship. However, model misspecification and parameter instability may undermine the performance of the proposed parametric models.  As pointed out by \cite{hansen2001new}, it may seem unlikely that  structural change could be immediate and might seem more reasonable to allow structural change to take a period of time to take effect. Practically, it may be more reasonable to allow smooth structural changes over a period of time rather than in an abrupt manner. To model the dynamic transit, an important strand of the VAR literature assumes that the coefficients of VAR models evolve in a random way (e.g., \citealp{primiceri2005time, petrova2019quasi}), and the estimation procedure relies on extensive Markov Chain Monte Carlo (MCMC) draws plus the use of a variety of filters, such as Kalman, or related filters. However, by doing so, the asymptotic properties of the estimated model coefficients as well as the corresponding impulse responses are still unclear (\citealp{giraitis2014inference}). 

Meanwhile, there is a separate literature about using nonparametric methods to estimate deterministically unknown time--varying parameters for autoregressive models. Up to this point, it is worth bringing up the terminology ``local stationarity", as a bulk of the literature is carried out using the local stationarity technique, which at least dates back to the seminal work by \cite{dahlhaus1996}. Recent developments focus on univariate autoregressive models (\citealp{dahlhaus2006statistical,zhang2012inference,richter2019cross,shlwz20}). To the best of our knowledge, it has had very little success to extend the local stationarity technique to multivariate settings, such as VAR models with deterministic time--varying parameters. In some specific cases where different locally stationary univariate time series may be approximated by their stationary versions on the same segments, the local stationarity technique may be applicable to such specific multivariate settings. However, univariate time series of general multivariate time series may have quite different patterns and behaviours, such as the three univariate time series plotted in Figure \ref{Fg1} of Section \ref{Sec5}.

To address the aforementioned issues, this paper therefore proposes a class of deterministic time-varying VAR models where both VAR coefficients and covariance matrix of the model's error innovations are allowed to change smoothly over time. We develop a time-varying vector moving average infinity (VMA$(\infty)$) representation for a class of VAR models before we are able to establish uniform consistency results and a joint central limit theory for kernel-based estimators of both VAR coefficients and covariance matrix, which facilitate the inference on time-varying structural impulse responses. This type of impulse responses is of major importance in typical VAR applications \cite[]{ik13,ik20,paul2019time}. In addition, inferences on structural impulse responses subject to both short-run timing and long-run restrictions are considered, so that the proposed model can better capture the simultaneous relations among multiple variables of interest over time. Such modelling strategy is especially useful for analyzing multivariate time series over a long horizon, since it offers a comprehensive treatment on tracking interests which are affected by frequently updated policies, environment, system, etc. In an economy system consisting of inflation, unemployment and interest rates, we discuss time-varying impacts of the interest rate change, which helps stabilize fluctuations in inflation and unemployment in the long-run. Under the proposed framework, it is achieved by investigating the corresponding time-varying impulse response functions.

We now comment on the literature closely related to the parameter stability testing problem we are considering in this paper. Detecting and estimating parametric components in univariate time-varying autoregressive models has been studied by \cite{zhang2012inference}. Recently, \cite{truquet2017parameter} considers parameter stability testing for univariate time-varying autoregressive conditional heteroscedasticity. More recently, \cite{LS2021} discuss testing for time--varying impact effects through a change--point mechanism. We develop a simple test for checking whether some of the time--varying coefficients (if not all) reduce to constant coefficients involved in the VAR models. It can be used to test whether the policy transmission mechanism is changing over time, which is of great importance in the macroeconomic literature \cite[e.g.,][]{primiceri2005time,paul2019time}. To give another example, consider our empirical study on the ongoing debate about the high inflation during 1970-1980, i.e., whether the high inflation is due to bad monetary policy (\citealp{primiceri2005time, sims2006}). Mathematically, the discrepancy comes down to specification testing on the coefficient matrices of the VAR models by the proposed test.

In summary, our contributions are in three-fold. First, we propose a class of time-varying VAR($p$) models, and develop a time-varying VMA$(\infty)$ representation for the VAR($p$) models before we establish an estimation theory for time-varying structural impulse responses, which are often of great interest to describe how the economy reacts over time to structural shocks. The proposed estimation is conducted subject to both short-run timing and long-run restrictions, which have attracted considerable attention in the literature (\citealp{kilian2017structural}). Second, we develop a Wald--type test statistic for detecting time-invariant parameters in time-varying VAR models before we show that the proposed test statistic is asymptotically normally distributed under both the null hypothesis and a sequence of local alternatives. Third, in an empirical study, we investigate the changing dynamics of three key U.S. macroeconomic variables (inflation, unemployment, and interest rate), and uncover a fall in the volatilities of exogenous shocks. We observe that there exists a substantial time-variation in the policy transmission mechanism, and the ``price puzzle'' is limited to periods of bad monetary policy.

The organization of this paper is as follows. Section \ref{Sec2} considers a class of time-varying VAR models, and establish asymptotic properties for the estimated time-varying impulse response functions. Section \ref{Sec3} describes the proposed test and establishes the corresponding asymptotic theory. Section \ref{Sec4} discusses some implementation issues and presents comprehensive simulation studies. Section \ref{Sec5} provides a case study to demonstrate the empirical relevance of the proposed models and estimation theory. Section \ref{Sec6} concludes. The proofs of the main results are given in Appendix A. Preliminary lemmas and their proofs are given in Appendix B.

Before proceeding further, it is convenient to introduce some notations: $\| \cdot \|$ denotes the Euclidean norm of a vector or the Frobenius norm of a matrix; $\otimes$ denotes the Kronecker product; $\bm{I}_a$ stands for an $a\times a$ identity matrix; $\bm{0}_{a\times b}$ stands for an $a\times b$ matrix of zeros, and we write $\bm{0}_a$ for short when $a=b$; for a function $g(w)$, let $g^{(j)}(w)$ be the $j^{th}$ derivative of $g(w)$, where $j\ge 0$ and $g^{(0)}(w) \equiv g(w)$; $K_h(\cdot) =K(\cdot/h)/h$, where $K(\cdot)$ and $h$ stand for a nonparametric kernel function and a bandwidth respectively; let $\tilde{c}_k =\int_{-1}^{1} u^k K(u) du$ and $\tilde{v}_k= \int_{-1}^{1} u^k K^2(u) du$ for integer $k\ge 0$; $\mathrm{vec}(\cdot)$ stacks the elements of an $m\times n$ matrix as an $mn \times 1$ vector; for any $a\times a$ square matrix $\bm{A}$, $\mathrm{vech}\left(\bm{A}\right)$ denotes the $\frac{1}{2}a(a+1)\times 1$ vector obtained from $\mathrm{vec}\left(\bm{A}\right)$ by eliminating all supra--diagonal elements of $\bm{A}$; $\mathrm{tr}\left(\bm{A}\right)$ denotes the trace of $\bm{A}$. Finally, $\to_P$ and $\to_D$ denote convergence in probability and convergence in distribution, respectively.

\section{The Time-Varying VAR($p$) Model}\label{Sec2}

Suppose that we observe $\{\bm{x}_{-p+1},\ldots,\bm{x}_0,\bm{x}_1,\ldots,\bm{x}_T\}$  from the following data generating process:
\begin{equation}\label{Eq2.1}
\bm{x}_t=\bm{a}(\tau_t)+\sum_{j=1}^{p}\bm{A}_{j}(\tau_t) \bm{x}_{t-j}+\bm{\eta}_t  \ \ \mbox{with} \ \ \bm{\eta}_t=\bm{\omega}(\tau_t)\bm{\epsilon}_{t},
\end{equation}
where $\tau_t = t/T$, $\bm{a}(\cdot)$ and $ \bm{A}_{j}(\cdot)$ are respectively a $d\times 1$ vector and $d\times d$ matrices of unknown functional coefficients, and $\bm{\omega}(\tau)$ is a $d\times d$ matrix of unknown functions capturing the heteroskedasticity of the error components.  Allowing $\bm{\omega}(\cdot)$ to vary over time is important theoretically and practically, because a constant covariance matrix implies that the shock to the $i^{th}$ variable of $\bm{x}_{t}$ has a time-invariant effect on the $j^{th}$ variable of $\bm{x}_{t}$, restricting simultaneous interactions among multiple variables to be time-invariant (\citealp{primiceri2005time}). For the time being, we assume $p$ is known, and shall come back to the choice of $p$ in Section \ref{Sec2.4}.

The following conditions are necessary for our development.

\begin{assumption}\label{Ass1}

\item
\begin{enumerate}
\item The roots of $\bm{I}_d-\bm{A}_{1}(\tau)L-\cdots -\bm{A}_{p}(\tau)L^p=\bm{0}_d$ all lie outside the unit circle uniformly in $\tau \in [0,1]$.

\item Each element of $\bm{A}(\tau)=\left[\bm{a}(\tau),\bm{A}_1(\tau),\ldots,\bm{A}_{p}(\tau)\right]$ is second order continuously differentiable on $[0,1]$ and $\bm{A}(\tau)=\bm{A}(0)$ for $\tau<0$.

\item Each element of $\bm{\omega}(\tau)$ is second order continuously differentiable on $[0,1]$. Moreover, $\bm{\Omega}(\tau)=\bm{\omega}(\tau)\bm{\omega}(\tau)^\top$ is positive definite uniformly in $\tau \in [0,1]$ and $\bm{\omega}(\tau)=\bm{\omega}(0)$ for $\tau<0$.
\end{enumerate}
\end{assumption}

\begin{assumption}\label{Ass2}
$\{\bm{\epsilon}_t\}_{t=-\infty}^{\infty}$ is a martingale difference sequence (m.d.s.) adapted to the filtration $\left\{\mathcal{F}_t\right\}$, where $\mathcal{F}_t=\sigma\left(\bm{\epsilon}_t,\bm{\epsilon}_{t-1},\ldots\right)$ is the $\sigma$-field generated by $\left(\bm{\epsilon}_t,\bm{\epsilon}_{t-1},\ldots\right)$, $E[\bm{\epsilon}_t \bm{\epsilon}_t^\top | \mathcal{F}_{t-1}]=\bm{I}_d$ almost surely (a.s.), and $ \max_{t}E \left\|\bm{\epsilon}_t\right\|^\delta < \infty$ for some $\delta > 4$.
\end{assumption}

Assumption \ref{Ass1}.1 ensures the eigenvalues of the companion matrix $\bm{\Phi}(\tau)$ of \eqref{Eq2.3} below all lie inside the unit circle uniformly over $\tau \in [0,1]$.  As a consequence, \eqref{Eq2.1} cannot include any unit-root or explosive components. Similar treatments have also been adopted to investigate univariate locally stationary models in the literature (e.g., Assumption T3 of \citealp{zhang2012inference}). Assumptions \ref{Ass1}.2 and \ref{Ass1}.3 allow the underlying data generating process to evolve over time in a smooth manner. The conditions $\bm{A}(\tau)=\bm{A}(0)$ and $\bm{\omega}(\tau)=\bm{\omega}(0)$ for $\tau<0$ gives
\begin{eqnarray}\label{Eq2.2}
\bm{x}_t=\bm{a}(0)+\sum_{j=1}^{p}\bm{A}_{j}(0) \bm{x}_{t-j}+\bm{\omega}(0)\bm{\epsilon}_{t},
\end{eqnarray}
which basically assumes that $\bm{x}_t$ behaves like a parametric VAR$(p)$ model for $t\le 0$. A similar condition can be found in \cite{vogt2012nonparametric} for a nonparametric time-varying time series model.

Assumption \ref{Ass2} imposes some conditions on the innovation error terms, and are standard in the VAR literature \cite[cf.,][]{lutkepohl2005new}. 

\medskip

With these conditions, the following proposition says \eqref{Eq2.1} admits a time-varying vector moving average infinity (VMA$(\infty)$) representation, which sheds a light on how to recover the time-varying structural impulse responses.

\begin{proposition}\label{Proposition2.1}
Under Assumptions \ref{Ass1} and \ref{Ass2}, there exists a time-varying VMA$(\infty)$ process of the form:
$$
\widetilde{\bm{x}}_t= \bm{\mu}(\tau_t)+\bm{B}_0(\tau_t)\bm{\epsilon}_t+\bm{B}_1(\tau_t)\bm{\epsilon}_{t-1}+\bm{B}_2(\tau_t)\bm{\epsilon}_{t-2}+\cdots
$$
such that $\bm{x}_t$ of \eqref{Eq2.1} satisfies $\max_{t\geq 1} \{E\left\|\bm{x}_t-\widetilde{\bm{x}}_t\right\|^\delta\}^{1/\delta}=O(T^{-1})$, where
\begin{eqnarray}\label{Eq2.3}
&&\bm{\mu}(\tau)=\bm{a}(\tau)+\sum_ {j=1}^{\infty}\bm{\Psi}_j(\tau)\bm{a}(\tau),\quad  \bm{\Psi}_j(\tau)=\bm{J}\bm{\Phi}^j(\tau) \bm{J}^\top \text{ for }j\geq 1,\nonumber \\
&&\bm{J}=\left[\bm{I}_d,\bm{0}_{d\times d(p-1)}\right],\quad \bm{B}_0(\tau)=\bm{\omega}(\tau),\quad  \bm{B}_j(\tau)=\bm{\Psi}_j(\tau)\bm{\omega}(\tau),  \nonumber \\
&&\bm{\Phi}(\tau)=\left[\begin{matrix}
       \bm{A}_{1}(\tau) & \cdots & \bm{A}_{p-1}(\tau) & \bm{A}_{p}(\tau)  \\
       \bm{I}_d & \cdots& \bm{0}_d & \bm{0}_d\\
       \vdots & \ddots&\vdots & \vdots\\
       \bm{0}_d &\cdots &\bm{I}_d & \bm{0}_d\\
    \end{matrix} \right].
\end{eqnarray}
\end{proposition}
From Proposition \ref{Proposition2.1}, we can see that the $d \times 1$ vector of the orthogonalized impulse response function of a unit shock at time $t$ to the $j^{th}$ equation on $\bm{x}_{t+n}$ is given by $\bm{B}_n(\tau_{t+n})\bm{e}_j$, where $\bm{e}_j$ is a $d \times 1$ selection vector with unity as its $j^{th}$ element and zeros elsewhere. Hence, the impulse responses produced by our model are deterministic functions of rescaled time, so that the TV-VAR model captures potential drifts in the transmission mechanism and produces impulse responses which are not history- and shock-dependent.

\subsection{Estimation}\label{Sec2.2}

To estimate $\{\bm{B}_j(\tau): j\geq 0\}$, we need a joint central limit theorem for the estimators of the coefficients and the innovation covariance matrix. That said, we consider the estimation of $\bm{A}(\cdot)$ and $\bm{\Omega}(\cdot)$ using the local linear kernel method\footnote{The local linear kernel method allows us to address the so-called boundary effects of the kernel estimation when establishing the first result of Theorem \ref{Thm2.1} below. Alternatively, one may consider some boundary adjustment approaches as mentioned in \cite{HongLi} and \cite{CHL2012}.}. Intuitively,  when $\tau_t$ is in a small neighbourhood of $\tau$, we can write \eqref{Eq2.1} as

\begin{eqnarray}\label{Eq2.4}
\bm{x}_t = \bm{A}(\tau_t) \bm{z}_{t-1} + \bm{\eta}_t \approx \left[\bm{A}(\tau),h\bm{A}^{(1)}(\tau)\right]\bm{z}_{t-1}^*+\bm{\eta}_t,
\end{eqnarray}
where $\bm{z}_{t-1}=\left[1,\bm{x}_{t-1}^\top,\ldots,\bm{x}_{t-p}^\top\right]^\top$ and $\bm{z}^*_{t-1} = \left[\bm{z}_{t-1}^\top,\frac{\tau_t-\tau}{h}\bm{z}_{t-1}^\top\right]^\top$. The local linear estimators\footnote{It is worth pointing out that the estimation of the covariance matrix using the local linear kernel method (such as the second estimator of \eqref{Eq2.5}) is a non-trivial problem, and even has its own literature. We refer interested readers to \cite{zhang2012inference} for more details.} of $\bm{A}(\tau)$ and $\bm{\Omega}(\tau)$ are then respectively given by
\begin{eqnarray}\label{Eq2.5}
\mathrm{vec} [\bm{\widehat{A}}(\tau)]&=&[\bm{I}_{d^2p+d},\bm{0}_{d^2p+d}]\cdot\left(\sum_{t=1}^{T} \bm{Z}_{t-1}^*\bm{Z}_{t-1}^{*,\top} K_h (\tau_t-\tau) \right)^{-1} \sum_{t=1}^{T}\bm{Z}_{t-1}^*\bm{x}_t K_h (\tau_t-\tau), \nonumber\\
\bm{\widehat{\Omega}}(\tau)&=&\frac{1}{T}\sum_{t=1}^{T}\bm{\widehat{\eta}}_t\bm{\widehat{\eta}}_t^\top \omega_{t}(\tau),
\end{eqnarray}
where $\bm{Z}_{t}^*=\bm{z}_{t}^*\otimes \bm{I}_d$, $\bm{\widehat{\eta}}_t=\bm{x}_t-\bm{\widehat{A}}(\tau_t)\bm{z}_{t-1}$, $\omega_{t}(\tau)=K_h(\tau_t-\tau)\frac{ P_{h,2}(\tau)-\frac{\tau_t-\tau}{h}P_{h,1}(\tau)}{P_{h,0}(\tau)P_{h,2}(\tau)-P_{h,1}^2(\tau)}$ is the local linear weight, and $P_{h,k}(\tau)=\frac{1}{T}\sum_{t=1}^{T}\left(\frac{\tau_t-\tau}{h}\right)^k K_h(\tau_t-\tau)$ for $k=0,1,2$. 

\medskip

We require the following conditions to hold for the kernel function and the bandwidth.

\begin{assumption}\label{Ass3}
Let $K(\cdot)$ be a symmetric and positive kernel function defined on $[-1,1]$ with $\int_{-1}^{1}K(u)\mathrm{d}u = 1$. Moreover, $K(\cdot)$ is Lipschitz continuous on $[-1,1]$. As $(T,h) \to (\infty, 0)$, $Th\to \infty$.
\end{assumption}

With Assumption \ref{Ass3} in hand, we summarize the first theorem of this paper below.

\begin{theorem}\label{Thm2.1}
Let Assumptions \ref{Ass1}-\ref{Ass3} hold. Suppose that $\max_{t\geq1} E [\|\bm{\epsilon}_t \|^4 |\mathcal{F}_{t-1} ] < \infty $ a.s., and $\frac{T^{1-\frac{4}{\delta}}h}{\log T} \to \infty$. Then
\begin{enumerate}
\item[1.] $\sup_{\tau \in [0,1]} \| \bm{\widehat{A}}(\tau)-\bm{A}(\tau) \|=O_P \left(h^2+ (\frac{\log T}{Th} )^{1/2} \right)$.
\end{enumerate}
In addition, suppose that conditional on $\mathcal{F}_{t-1}$, the third and fourth moments of $\bm{\epsilon}_t$ are identical to the corresponding unconditional moments a.s., and $Th^{5} \to \alpha \in [0,\infty)$. Then the following two results also hold for $\forall\tau \in (0,1)$:
\begin{enumerate}
\item[2.] $\sqrt{Th} \widehat{\bm{V}}^{-1/2}(\tau)\left[\begin{matrix}
\mathrm{vec}\left(\bm{\widehat{A}}(\tau)-\bm{A}(\tau)-\frac{1}{2}h^2\tilde{c}_2\bm{A}^{(2)}(\tau)\right) \\
\mathrm{vech}\left(\bm{\widehat{\Omega}}(\tau)-\bm{\Omega}(\tau)-\frac{1}{2}h^2\tilde{c}_2\bm{\Omega}^{(2)}(\tau)\right)  \end{matrix}
\right]\to_D N\left(\bm{0},\bm{I}\right),$
\end{enumerate}
where $\bm{V}(\tau) $ and $\widehat{\bm{V}}(\tau)$ are defined in \eqref{EqA.3} and \eqref{EqA.5} for the sake of presentation.
\end{theorem}

The first result of Theorem \ref{Thm2.1} establishes a uniform convergence rate for $\bm{\widehat{A}}(\tau)$, which further allows us to establish a joint asymptotic distribution in the second result. If $\delta>5$, the usual optimal bandwidth $h_{opt}=O\left(T^{-1/5}\right)$ satisfies the condition $\frac{T^{1-\frac{4}{\delta}}h}{\log T} \to \infty$.

\subsection{On Impulse Responses}\label{Sec2.3}

Having established the joint CLT in Theorem \ref{Thm2.1}, we are now ready to study the impulse responses. As $\bm{\Omega}(\cdot)=\bm{\omega}(\cdot)\bm{\omega}^\top(\cdot)$, we cannot infer the elements of $\bm{\omega}(\cdot)$ unless certain identification restrictions are imposed. In the following, we consider two types of identification conditions: (i) the short-run timing restrictions, and (ii) the long-run restrictions. The economic interpretations of the two types of identification conditions can be found in \cite{kilian2017structural}, and we do not repeat them here for the sake of space.

Under the short-run timing restrictions, $\bm{\omega}(\cdot)$ is a lower-triangular matrix. Thus, $\widehat{\bm{\omega}}(\tau)$ is chosen as the lower triangular matrix from the Cholesky decomposition of $\bm{\widehat{\Omega}}(\tau)$, i.e., $\bm{\widehat{\Omega}}(\tau)=\bm{\widehat{\omega}}(\tau)\bm{\widehat{\omega}}^\top(\tau)$. Alternatively, one can impose the conditions on the long-run impacts of the shocks (i.e., $\bm{B}(\tau)$ defined below).  Specifically, define
\begin{eqnarray}\label{Eq2.6}
&&\bm{B}(\tau):=\sum_{j=0}^{\infty}\bm{B}_j(\tau)=\left(\bm{I}_d-\sum_{i=1}^{p}\bm{A}_i(\tau)\right)^{-1}\bm{\omega}(\tau),\nonumber \\
&&\bm{\Psi}(\tau):=\sum_{j=0}^{\infty}\bm{\Psi}_j(\tau)=\left(\bm{I}_d-\sum_{i=1}^{p}\bm{A}_i(\tau)\right)^{-1},
\end{eqnarray}
where the last equalities of both lines follow in an obvious matter. Thus, the elements of $\bm{B}(\tau)$ may be recovered from $\bm{B}(\tau)\bm{B}^\top(\tau) = \bm{\Psi}(\tau)\bm{\Omega}(\tau)\bm{\Psi}^\top(\tau)$. It is then convenient to assume that  $\bm{B}(\tau)$ is a lower-triangular matrix, so $\widehat{\bm{B}}(\tau)$ is chosen as the lower triangular matrix from the Cholesky decomposition of $\widehat{\bm{\Psi}}(\tau)\widehat{\bm{\Omega}}(\tau)\widehat{\bm{\Psi}}^\top(\tau)$, where $\widehat{\bm{\Psi}}(\tau)$ is defined in the same way as $\bm{\Psi}(\tau)$ in \eqref{Eq2.3} but replacing $\bm{A}_i(\tau)$ with $\widehat{\bm{A}}_i(\tau)$. Under the long-run restrictions, $\widehat{\bm{\omega}}(\tau) = \widehat{\bm{\Psi}}^{-1}(\tau)\widehat{\bm{B}}(\tau)$.

Either way, the estimator of the impulse response function $\bm{B}_j(\tau)$ for each given $j\ge 0$ is given by

\begin{eqnarray}\label{Eq2.7}
\bm{\widehat{B}}_j(\tau)=\bm{\widehat{\Psi}}_j(\tau)\bm{\widehat{\omega}}(\tau),
\end{eqnarray}
where $\bm{\widehat{\Psi}}_j(\tau)=\bm{J} \bm{\widehat{\Phi}}^j(\tau) \bm{J}^\top$, $\bm{J}$ is defined in Proposition \ref{Proposition2.1}, and $ \bm{\widehat{\Phi}}^j(\tau)$ is defined under \eqref{EqA.6}. 

We summarize the asymptotic properties of the estimation of the impulse responses in the following theorem.

\begin{theorem}\label{Thm2.2}
Under the conditions of Theorem \ref{Thm2.1}. For any fixed integer $j\ge 0$
\begin{eqnarray*}
 \sqrt{Th}\left(\mathrm{vec}\left(\bm{\widehat{B}}_j(\tau)-\bm{B}_j(\tau)\right)-\frac{1}{2}h^2\tilde{c}_2\bm{B}_j^{(2)}(\tau)\right)\to_D N\left(0,\bm{\Sigma}_{\bm{B}_j}(\tau)\right),
\end{eqnarray*}
where 
\begin{eqnarray*}
\bm{B}_j^{(2)}(\tau)&=&\bm{C}_{j,1}(\tau)\mathrm{vec}\left(\bm{A}^{(2)}(\tau)\right) +\bm{C}_{j,2}(\tau)\mathrm{vech}\left(\bm{\Omega}^{(2)}(\tau)\right),\\
\bm{\Sigma}_{\bm{B}_j}(\tau) &=& \left[\bm{C}_{j,1}(\tau),\bm{C}_{j,2}(\tau)\right]\bm{V}(\tau)\left[\bm{C}_{j,1}(\tau),\bm{C}_{j,2}(\tau)\right]^\top.
\end{eqnarray*}

Specifically, we have the following expressions.

\begin{enumerate}
\item Under the short-run timing restrictions, we have
\begin{eqnarray*}
\bm{C}_{0,1}(\tau)&=&0,\\
\bm{C}_{j,1}(\tau)&=&\left(\bm{\omega}^\top(\tau)\otimes \bm{I}_d\right) \left(\sum_{m=0}^{j-1} \bm{J}(\bm{\Phi}^\top(\tau))^{j-1-m}\otimes \bm{\Psi}_m(\tau) \right) \left[\bm{0}_{d^2p\times d},\bm{I}_{d^2p}\right],\ j \geq 1,\\
\bm{C}_{j,2}(\tau)&=&\left(\bm{I}_d\otimes\bm{\Psi}_j(\tau)\right) \bm{L}_d^\top\left(\bm{L}_d\bm{N}_1(\tau)\bm{L}_d^\top \right)^{-1},\ j \geq 0,
\end{eqnarray*}
in which $\bm{N}_1(\tau)=(\bm{I}_{d^2}+\bm{K}_{d,d})(\bm{\omega}(\tau)\otimes \bm{I}_d)$, the elimination matrix $\bm{L}_d$ satisfies that $\mathrm{vech}(\bm{F})=\bm{L}_d\mathrm{vec}(\bm{F})$ for any $d\times d$ matrix $\bm{F}$, and the commutation matrix $\bm{K}_{m,n}$ satisfies   $\bm{K}_{m,n}\mathrm{vec}(\bm{G})=\mathrm{vec}(\bm{G}^\top)$ for any $m\times n$ matrix $\bm{G}$.

\item Under the long-run restrictions, we have
\begin{eqnarray*}
\bm{C}_{0,1}(\tau)&=&(\bm{I}_d \otimes \bm{\Psi}_j(\tau))\left(\bm{N}_1^\top(\tau)\bm{N}_1(\tau)+\bm{N}_2^\top(\tau)\bm{N}_2(\tau)\right)^{-1}\bm{N}_2^\top(\tau)\bm{D}_2(\tau),\\
\bm{C}_{j,1}(\tau)&=&\left(\bm{\omega}^\top(\tau)\otimes \bm{I}_d\right) \left(\sum_{m=0}^{j-1} \bm{J}(\bm{\Phi}^\top(\tau))^{j-1-m}\otimes \bm{\Psi}_m(\tau) \right)\left[\bm{0}_{d^2p\times d},\bm{I}_{d^2p}\right]+(\bm{I}_d \otimes \bm{\Psi}_j(\tau))\\
&&\times\left(\bm{N}_1^\top(\tau)\bm{N}_1(\tau)+\bm{N}_2^\top(\tau)\bm{N}_2(\tau)\right)^{-1}\bm{N}_2^\top(\tau)\bm{D}_2(\tau)\left[\bm{0}_{d^2p\times d},\bm{I}_{d^2p}\right],\ j \geq 1,\\
\bm{C}_{j,2}(\tau)&=&(\bm{I}_d \otimes \bm{\Psi}_j(\tau))\left(\bm{N}_1^\top(\tau)\bm{N}_1(\tau)+\bm{N}_2^\top(\tau)\bm{N}_2(\tau)\right)^{-1}\bm{N}_1^\top(\tau)\bm{D}_1,\ j \geq 0,
\end{eqnarray*}
in which $\bm{N}_2(\tau)=\bm{Q}\left(\bm{I}_d\otimes \bm{A}_\tau^{-1}(1)\right)$, $\bm{D}_{2}(\tau)=\bm{Q}\left(\bm{B}^\top(\tau)\otimes\bm{A}_\tau^{-1}(1)\right)\triangledown_{\bm{\alpha}(\tau)}\bm{A}_{\tau}(1)$ with $\bm{A}_{\tau}(1)=\bm{I}_d-\sum_{i=1}^{p}\bm{A}_i(\tau)$ and $\triangledown_{\bm{\alpha}(\tau)}\bm{A}_{\tau}(1)=-\left[\bm{I}_{d^2},...,\bm{I}_{d^2}\right]$ $(d^2\times d^2p)$, the duplication matrix $\bm{D}_1$ satisfies $\mathrm{vec}\left(\bm{\Omega}(\tau)\right)=\bm{D}_1\mathrm{vech}\left(\bm{\Omega}(\tau)\right)$, and $\bm{Q}$ is a $d(d-1)/2 \times d^2$ selection matrix of $0$ and $1$ such that $\bm{Q}\mathrm{vec}\left(\bm{B}(\tau)\right)=0$.
\end{enumerate}
\end{theorem}

For ease of presentation, we provide the definitions of the respective estimators of $ \bm{V}(\tau)$ and $\bm{\Phi}(\tau)$ (i.e., $\widehat{\bm{V}}(\tau)$ and $\widehat{\bm{\Phi}}(\tau)$) in Appendix \ref{ApA.1}. It is easy to see that $\widehat{\bm{\Phi}}(\tau)\to_P \bm{\Phi}(\tau)$, $\widehat{\bm{\omega}}(\tau)\to_P\bm{\omega}(\tau)$, and $\widehat{\bm{V}}(\tau)\to_P \bm{V}(\tau)$ by Theorem \ref{Thm2.1}. As a result, $\widehat{\bm{\Sigma}}_{\bm{B}_j}(\tau)\to_P\bm{\Sigma}_{\bm{B}_j}(\tau)$, where $\widehat{\bm{\Sigma}}_{\bm{B}_j}(\tau)$  has a form identical to $\bm{\Sigma}_{\bm{B}_j}(\tau)$ but replacing $\bm{\Phi}(\tau)$, $\bm{\omega}(\tau)$ and $ \bm{V}(\tau)$ with their estimators, respectively.

\subsection{Selection of the Optimal Lag}\label{Sec2.4}

We next consider the choice of the optimal number of lags, i.e., the estimation of $p$. Specifically, we consider the minimization of the next information criterion.
\begin{eqnarray}\label{Eq2.8}
\widehat{\mathsf{p}} =  \argmin_{1\le \mathsf{p}\le\mathsf{P} } (\log \left\{\text{RSS}(\mathsf{p})\right\}+\mathsf{p}\cdot\chi_T)\equiv \argmin_{1\le \mathsf{p}\le\mathsf{P} }\text{IC}(\mathsf{p}) ,
\end{eqnarray}
where $\text{RSS}(\mathsf{p})=\frac{1}{T}\sum_{t=1}^{T}\widehat{\bm{\eta}}_{\mathsf{p},t}^\top \widehat{\bm{\eta}}_{\mathsf{p},t}$, $\chi_T$ is the penalty term, $\widehat{\bm{\eta}}_{\mathsf{p},t}$ is the value of $\widehat{\bm{\eta}}_{t}$ by letting the lag be $\mathsf{p}$, and $\mathsf{P}$ is a sufficiently large fixed positive integer. The next theorem shows the validity of \eqref{Eq2.8}.

\begin{theorem}\label{Thm2.3}
Let Assumptions \ref{Ass1}-\ref{Ass3} hold. Suppose $\frac{T^{1-\frac{4}{\delta}}h}{\log T} \to \infty$, $\max_{t\geq1} E (\|\bm{\epsilon}_t \|^4 |\mathcal{F}_{t-1} ) < \infty $ a.s., $\chi_T\to 0$, and $c_T^{-2}\chi_T\to \infty$, where $c_T=h^2+\left(\frac{\log T}{Th}\right)^{1/2}$. Then $\Pr\left(\widehat{\mathsf{p}}=p\right)\to 1$.
\end{theorem}

\noindent In view of the conditions on $\chi_T$, a natural choice is
\begin{eqnarray*}
\chi_T = \max\left\{h^4,\frac{\log T}{Th}\right\}\cdot \log(\log(Th)).
\end{eqnarray*}

To this end, we have established the asymptotic properties of the proposed estimators. In the following, we discuss some model specification issues when building \eqref{Eq2.1}.

\section{Testing for Parameter Stability}\label{Sec3}

To close our investigation, we consider testing whether some (if not all) components of the coefficient matrices are time-invariant, which as explained in the introduction may help settle the discrepancy on the ongoing debate about the high inflation in the U.S. during 1970-1980.

For model \eqref{Eq2.1}, we are specifically interested in testing
\begin{equation}\label{Eq3.1}
\mathbb{H}_0: \bm{C}\bm{\beta}(\cdot) = \bm{c} \text{ for some unknown $c\in \mathbb{R}^s$},
\end{equation}
where $\bm{\beta}(\tau) :=  \mathrm{vec}\left(\bm{A}(\tau)\right)$ and $\bm{C}$ is a selection matrix. Practically, the choice of $\bm{C}$ should be theory/application driven, and $\bm{c}$ needs to be estimated. For example, in the context of monetary policy analysis \cite[]{primiceri2005time}, one can let $\bm{C} = \left[\bm{0}_{d^2p\times d},\bm{I}_{d^2p}\right]$ to test whether the policy transmission mechanism is varying over time.  

The test statistic is constructed based on the weighted integrated squared errors:

\begin{equation}\label{Eq3.2}
\widehat{Q}_{\bm{C},\bm{H}}=\int_{0}^{1} \left\{\bm{C}\widehat{\bm{\beta}}(\tau)-\widehat{\bm{c}}\right\}^\top \bm{H}(\tau)\left\{\bm{C}\widehat{\bm{\beta}}(\tau)-\widehat{\bm{c}}\right\} \mathrm{d}\tau,
\end{equation}
where $\widehat{\bm{\beta}}(\cdot) :=\text{vec}[\widehat{\bm A}(\cdot)]$ should be obvious, and $\widehat{\bm{c}}$ is the estimator of  $\bm{c}$ to be constructed below. In \eqref{Eq3.2}, $\bm{H}(\cdot)$ is an $s\times s$ positive definite weighting matrix, and is typically set as the precision matrix associated with $\widehat{\bm{\beta}}(\cdot)$.

Let $\bm{C}$ be a pre-specified matrix and $\bm{c} = \int_{0}^{1}\bm{C}\bm{\beta}(\tau)\mathrm{d}\tau$. Then a natural estimator of $\bm{c}$ is given by
\begin{equation}\label{Eq3.3}
\widehat{\bm{c}} = \int_{0}^{1}\bm{C}\widehat{\bm{\beta}}(\tau)\mathrm{d}\tau,
\end{equation}
of which the  asymptotic properties are summarized in the following lemma.

\begin{lemma}\label{Thm3.1}
Let Assumptions \ref{Ass1}-\ref{Ass3} hold. Suppose further that $\max_{t\geq1} E [\|\bm{\epsilon}_t \|^4 |\mathcal{F}_{t-1} ] < \infty $ a.s., $\frac{T^{1-\frac{4}{\delta}}h}{\log T} \to \infty$, each element of $\bm{A}(\cdot)$ has finite third-order derivative, $Th^2/(\log T)^2 \to \infty $, and $Th^6 \to 0$. As $T \to \infty$, we have
$$
\sqrt{T}\left(\widehat{\bm{c}}-\bm{c}-\frac{1}{2}h^2\tilde{c}_2\int_{0}^{1}\bm{C}\bm{\beta}^{(2)}(\tau)\mathrm{d}\tau\right)\to_D N\left(\bm{0},\int_{0}^{1}\bm{C}\bm{V}_{\bm{\beta}}(\tau)\bm{C}^\top\mathrm{d}\tau\right),
$$
where $\bm{V}_{\bm{\beta}}(\tau) := \bm{\Sigma}^{-1}(\tau)\otimes\bm{\Omega}(\tau)$ and $\bm{\Sigma}(\tau)$ is defined in \eqref{EqA.1}.
\end{lemma}

Lemma \ref{Thm3.1} indicates that under the null, $\widehat{\bm{c}}$ will achieve  a parametric rate (i.e., $\sqrt{T}$). This is important for the finite sample study, since more reliable information can be retrieved from the estimators.
\medskip

Having established Lemma \ref{Thm3.1}, we move on to investigate the test statistic \eqref{Eq3.2}, and examine the local alternatives. After some tedious development, one actually can show that
\begin{eqnarray}\label{Eq3.4}
	\widehat{Q}_{\bm{C}, \bm{H}} =  \frac{1}{Th}\cdot \frac{\widetilde{v}_0}{T}\sum_{t=1}^T \bm{\eta}_t^\top \bm{Z}_{t-1}^\top \bm{H}_{0}(\tau_t) \bm{Z}_{t-1} \bm{\eta}_t  +O_P\left(\frac{1}{T\sqrt{h}} \right)
\end{eqnarray}
where $\bm{Z}_{t} = \bm{z}_t\otimes \bm{I}_d$, $\bm{H}_{0}(\tau)=\bm{\Sigma}_{\bm{Z}}^{-1}(\tau)\bm{C}^\top\bm{H}(\tau)\bm{C}\bm{\Sigma}_{\bm{Z}}^{-1}(\tau)$, $\bm{\Sigma}_{\bm{Z}}(\tau) =  \bm{\Sigma}(\tau) \otimes \bm{I}_d$. Clearly,
\begin{eqnarray}\label{Eq3.41}
\frac{\widetilde{v}_0}{T}\sum_{t=1}^T \bm{\eta}_t^\top \bm{Z}_{t-1}^\top \bm{H}_{0}(\tau_t) \bm{Z}_{t-1} \bm{\eta}_t 
\end{eqnarray}
converges to a fixed value in probability, so the first term on the right hand side of \eqref{Eq3.4} is the leading one. Moreover, $ \frac{\widetilde{v}_0}{T}\sum_{t=1}^T \bm{\eta}_t^\top \bm{Z}_{t-1}^\top \bm{H}_{0}(\tau_t) \bm{Z}_{t-1} \bm{\eta}_t $ can be considered as a weighted average of the second moment of $\bm\eta_t$, it thus does not yield any distribution for us to conduct the hypothesis testing. In this regard, it can be considered a bias term which should be removed. The term which really yields an asymptotic distribution is in fact the second part on the right hand side of (\ref{Eq3.4}). Thus, we present the asymptotic distribution of the test statistic as follows.

\begin{theorem}\label{Thm3.2}
Let the conditions of Lemma \ref{Thm3.1} hold. Under the null hypothesis, if $Th^{5.5}\to 0$ and $E [\|\bm{\epsilon}_t \|^\delta |\mathcal{F}_{t-1} ] < \infty $ a.s., we have
	\begin{eqnarray*}
		T\sqrt{h}\left(\widehat{Q}_{\bm{C}, \widehat{\bm{H}}} -\frac{1}{Th}s\widetilde{v}_0\right)\to_D N(0,4s C_B),
	\end{eqnarray*}
	where $\widehat{\bm{H}}(\tau) =  (\bm{C}\widehat{\bm{V}}_{\bm{\beta}}(\tau)\bm{C}^\top )^{-1}$, $\widehat{\bm{V}}_{\bm{\beta}}(\tau) = \widehat{\bm{\Sigma}}^{-1}(\tau)\otimes\widehat{\bm{\Omega}}(\tau)$, $C_B = \int_{0}^{2}\left(\int_{-1}^{1-v}K(u)K(u+v) \mathrm{d}u\right)^2\mathrm{d}v$, and $\widehat{\bm{\Sigma}}(\tau)$ is defined in \eqref{EqA.5}.
\end{theorem}

Theorem \ref{Thm3.2} states that the test statistic converges to a normal distribution. The term $s\widetilde{v}_0$ is the limit associated with  \eqref{Eq3.41}.  Due to the use of $\widehat{\bm{H}}(\cdot)$, the second moment of $\bm\eta_t$ disappears from the asymptotic distribution automatically.

Here, we would like to emphasize that the proposed test is in fact a one-sided test, which is reflected in the investigation for the local alternative below. An intuitive explanation is that due to the quadratic form of \eqref{Eq3.2}, any departure from the true value will eventually yield a squared term when analysing the asymptotic power. Therefore, the null of \eqref{Eq3.1} will be rejected at the level $\alpha$ if

\begin{eqnarray}\label{Eq3.5}
	\widehat{Q}_{\bm{C}, \widehat{\bm{H}}}^* = \frac{T\sqrt{h}\left(\widehat{Q}_{\bm{C}, \widehat{\bm{H}}} -\frac{1}{Th}s \widetilde{v}_0\right)}{\sqrt{4 s C_B}} > q_{1-\alpha},
\end{eqnarray}
where $q_{1-\alpha}$ stands for the $(1-\alpha)^{th}$ quantile of the standard normal distribution.

\medskip

In what follows, we consider a sequence of local alternatives of the form:
\begin{eqnarray}\label{Eq3.6}
	\mathbb{H}_1: \bm{C}\bm{\beta}(\tau) = \bm{c} + d_T \bm{f}(\tau),
\end{eqnarray}
where $\bm{f}(\tau)$ is a twice continuously differentiable vector of functions, and $d_T \to 0$. The term $d_T \bm{f}(\tau)$ characterizes the departure of the time-varying coefficient $\bm{C}\bm{\beta}(\tau)$ from the constant $\bm{c}$. Following the development of Theorem \ref{Thm3.2}, it is straightforward to obtain the following corollary.

\begin{corollary}\label{Thm3.3}
Let the conditions of Theorem \ref{Thm3.2} hold. Under the $\mathbb{H}_1$ of \eqref{Eq3.6}, if $d_T = T^{-1/2}h^{-1/4}$, then
	
	\begin{eqnarray*}
		T\sqrt{h}\left(\widehat{Q}_{\bm{C}, \widehat{\bm{H}}} -\frac{1}{Th}s\widetilde{v}_0\right)\to_D N\left(\delta_1,4s C_B\right),
	\end{eqnarray*}
	where $\delta_1 = \int_{0}^{1}\bm{f}(\tau)^\top (\bm{C} \bm{V}_{\bm{\beta}}(\tau)\bm{C}^\top )^{-1}\bm{f}(\tau) \mathrm{d}\tau$. 
	
	Moreover, we have $\Pr\left(\widehat{Q}_{\bm{C}, \widehat{\bm{H}}}^*>q_{1-\alpha}\right) \to \Phi\left(q_{\alpha} + \frac{\delta_1}{2\sqrt{ s C_B}} \right)$.

\end{corollary}

Corollary \ref{Thm3.3} shows that the test has a non-trivial power against $\mathbb{H}_1$ when $d_T = T^{-1/2}h^{-1/4}$. If $T^{-1/2}h^{-1/4} = o(d_T)$, the power of the test converges to 1, i.e.,
\begin{eqnarray*}
	\Pr(\widehat{Q}_{\bm{C}, \widehat{\bm{H}}}^*>q_{1-\alpha})\to 1.
\end{eqnarray*}

Before we conclude this section, we would like to add some comments on the assumptions imposed and the main results established in relation to the relevant literature. The construction of the proposed test is similar to those discussed in \cite{zhang2012inference} and then \cite{truquet2017parameter}. Because we have developed and then employed Proposition 2.1 for the time--varying VMA($\infty$) representation, the assumptions, such as requiring  $Th^{5.5} \to 0$, are less restrictive than those assumed in the relevant literature, see, for example, $Th^{3.5} \to 0$ by \cite{truquet2017parameter}. As a consequence, the main techniques employed in our proofs may be of general interest and applicability in dealing with similar problems.
\medskip

In what follows, we will examine the finite sample performance of the asymptotic properties for the proposed estimators and test statistic by simulation studies.

\section{Simulation}\label{Sec4}

In this section, we first provide some details of the numerical implementation in Section \ref{Sec4.1}, and then respectively examine the estimation and hypothesis testing in Sections \ref{Sec4.2} and \ref{Sec4.3}.  

\subsection{Numerical Implementation}\label{Sec4.1}

Throughout the numerical studies, Epanechnikov kernel $K(u)=0.75(1-u^2)I(|u|\leq 1)$ is adopted. 

When selecting the optimal lag by \eqref{Eq2.8}, the bandwidth $\widehat{h}_{cv}$ is always chosen by minimizing the following cross-validation criterion function for each $\mathsf{p}$.
\begin{equation}\label{Eq4.1}
\widehat{h}_{cv}=\arg\min_{h}\sum_{t=1}^{T}\|\bm{x}_t-\widehat{\bm{a}}_{-t}(\tau_t)-\sum_{j=1}^{\mathsf{p}}\widehat{\bm{A}}_{j,-t}(\tau_t) \bm{x}_{t-j}\|^2,
\end{equation}
where $\widehat{\bm{a}}_{-t}(\cdot)$, and $\widehat{\bm{A}}_{j,-t}(\cdot)$ are obtained using \eqref{Eq2.5} but leaving the $t^{th}$ observation out. Once $\widehat{\mathsf{p}}$ and $\widehat{h}_{cv}$ are obtained, the estimation procedure is relatively straightforward. It is pointed out that it might be possible to extend the Bayesian bandwidth selection method proposed for the conventional time--varying regression model by \cite{cgz2019} to the time--varying VAR setting under study for an optimal choice of the bandwidth parameter.

We then comment on the testing procedure. To improve the finite sample performance of the test, we propose a simulation-assisted testing procedure. A similar procedure has also  been adopted by \cite{zhang2012inference} and \cite{truquet2017parameter} to for the same purpose in the context of univariate time-varying models. Alternatively, one may consider the moving block bootstrap as in  \cite{CK2021}. For simplicity, we adopt the former as follows.
\medskip

 {\bf Algorithm --- a simulation-assisted testing procedure}
		
	\begin{itemize}		
		\item[] Step 1: Use the sample $\{\bm{x}_t\}$ to estimate the unrestricted model and the restricted model, and then compute $\widehat{Q}_{\bm{C}, \widehat{\bm{H}}}$ based on \eqref{Eq3.2}.
		
		Step 2: Generate i.i.d. standard multivariate normal random vectors $\{\bm{x}_t^*\}$.
		
		Step 3: Compute the bootstrap statistic $\widetilde{Q}_{\bm{C}, \widehat{\bm{H}}}^b$ in the same way as $\widehat{Q}_{\bm{C}, \widehat{\bm{H}}}$, with $\{\bm{x}_t^*\}$ replacing the original sample $\{\bm{x}_t\}$.
		
		Step 4: Repeat Steps 2-3 $B$ times to obtain $B$ bootstrap test statistics $\{\widetilde{Q}_{\bm{C}, \widehat{\bm{H}}}^b\}_{b=1}^{B}$, as well as its empirical quantile $\widehat{q}_{1-\alpha}$. We reject the null hypothesis \eqref{Eq3.1} at level $\alpha$ if $\widehat{Q}_{\bm{C}, \widehat{\bm{H}}}>\widehat{q}_{1-\alpha}$.
	\end{itemize}

\subsection{Examining the Model Estimation}\label{Sec4.2}

We now examine the finite sample performance of the theoretical findings. The data generating process (DGP) is as follows.
\begin{equation}\label{Eq4.2}
\bm{x}_t=\bm{a}(\tau_t)+\bm{A}_1(\tau_t)\bm{x}_{t-1}+\bm{A}_2(\tau_t)\bm{x}_{t-2}+\bm{\eta}_t \ \ \mbox{with} \ \ \bm{\eta}_t=\bm{\omega}(\tau_t)\bm{\epsilon}_t \ \ \text{for}\ \ t=1,\ldots ,T,
\end{equation}
where $\bm{\epsilon}_t$'s are  i.i.d. draws from $N(\bm{0}_{2\times 1}, \bm{I}_2)$,
\begin{eqnarray*}
\bm{a}(\tau)&=&\left[0.5\sin(2\pi \tau),0.5\cos(2\pi \tau)\right]^\top ,\nonumber \\
\bm{A}_1(\tau)&=&\left[\begin{matrix}
                       0.8\exp{(-0.5+\tau)} & 0.8(\tau-0.5)^3 \\
                       0.8(\tau-0.5)^3 & 0.8+0.3\sin(\pi \tau)
                       \end{matrix} \right], \nonumber \\
\bm{A}_2(\tau)&=&\left[\begin{matrix}
                       -0.2\exp{(-0.5+\tau)} & 0.8(\tau-0.5)^2 \\
                       0.8(\tau-0.5)^2 & -0.4+0.3\cos(\pi \tau)
                       \end{matrix} \right],  \nonumber \\
\bm{\omega}(\tau)&=&\left[\begin{matrix}
                    1.5+0.2\exp{(0.5-\tau)}& 0 \\
                    0.1\exp{(0.5-\tau)}   & 1.5+0.5(\tau-0.5)^2
                       \end{matrix}\right].
\end{eqnarray*}
Let the sample size be $T\in \{200, 400, 800\}$, and conduct 1000 replications for each choice of $T$. 

First, we check whether the coefficient matrices $\bm{A}_1(\tau)$ and $\bm{A}_2(\tau)$ satisfy Assumption \ref{Ass1}.1. Thus, for each generated dataset, we compute the largest eigenvalue of the true companion matrix $\bm{\Phi}(\tau)$ in $\tau \in [0,1]$, which varies from 0.54 to 0.88 over replications indicating the validity of Assumption \ref{Ass1}.1.

Next, we report the percentages of $\widehat{\mathsf{p}} < 2$, $\widehat{\mathsf{p}} = 2$, and $\widehat{\mathsf{p}} > 2$ respectively based on 1000 replications. Table \ref{Tb1} shows that the information criterion \eqref{Eq2.8} performs reasonably well, as the percentages associated with $\widehat{\mathsf{p}}=2$ are sufficiently close to 1 even for $T=200$.

\begin{table}[htbp]
  \centering \small
  \caption{\small The percentages of $\widehat{\mathsf{p}} < 2$, $\widehat{\mathsf{p}} = 2$, and $\widehat{\mathsf{p}} > 2$}\label{Tb1}
    \begin{tabular}{l ccc}
   \hline
    $T$ & $\widehat{\mathsf{p}} < 2$ & $\widehat{\mathsf{p}} = 2$ & $\widehat{\mathsf{p}} > 2$ \\
    \hline
   200   & 0.004 & 0.976 & 0.020 \\
   400   & 0.004 & 0.986 & 0.010 \\
   800   & 0.000 & 1.000 & 0.000 \\
   \hline
    \end{tabular}
\end{table}

Finally, we evaluate the estimates of $\bm{A}(\tau)$, $\bm{\Omega}(\tau)$, as well as the estimates of the impulse responses (say, $\bm{B}_1(\tau)$ and $\bm{B}_5(\tau)$ without loss of generality) based on the short-run timing restrictions. For each parameter of interest, we calculate the root mean square error (RMSE) as follows:
\begin{equation*}
\left\{\frac{1}{1000T} \sum_{n=1}^{1000} \sum_{t=1}^{T}\|\widehat{\bm{\theta}}^{(n)}(\tau_t)-\bm{\theta}(\tau_t) \|^2\right\}^{1/2},
\end{equation*}
where $\bm{\theta}(\cdot)\in\left\{\bm{A}(\cdot),\bm{\Omega}(\cdot),\bm{B}_1(\tau),\bm{B}_5(\tau)\right\}$, and $\widehat{\bm{\theta}}^{(n)}(\tau)$ is the estimate of $\bm{\theta}(\tau)$ for the $n^{th}$ replication. Of interest, we also report the finite sample coverage probabilities of the confidence intervals. The nominal coverage is 95\%. Given $\bm{\theta}(\cdot)$, for each generated dataset, the coverage probability is first calculated for each functional component of $\bm{\theta}(\cdot)$ over the grid points $\{\tau_t,t=1,\ldots,T\}$, and then we further take an average across the elements of $\bm{\theta}(\cdot)$.  After 1000 replications, we present the averaged value of these coverage probabilities  in Table \ref{Tb2}.

As shown in Table \ref{Tb2}, the RMSEs decrease as the sample size increases. The finite sample coverage probabilities are smaller than their nominal level when $T=200$, but are fairly close to 95\% as $T=800$.

\begin{table}[htbp]
  \centering \small
  \caption{\small The RMSEs and the empirical coverage probabilities (in parentheses)}\label{Tb2}
    \begin{tabular}{l cccc}
    \hline
$T$ & $\bm{A}(\tau)$ &$\bm{\Omega}(\tau)$ &$\bm{B}_1(\tau)$ & $\bm{B}_5(\tau)$\\
    \hline
 200   & 0.54 (0.89) & 0.83 (0.87)& 0.46 (0.87)  & 0.31 (0.89) \\
 400   & 0.40 (0.91) & 0.71 (0.91)& 0.30 (0.91)  & 0.34 (0.89) \\
 800   & 0.29 (0.92) & 0.62 (0.93)& 0.29 (0.92)  & 0.30 (0.90)  \\
    \hline
    \end{tabular}
\end{table}

\subsection{Examining the Parameter Stability Testing}\label{Sec4.3}

To evaluate the size and local power of the proposed test statistic, we consider the following DGP:

\begin{equation}\label{Eq4.3}
	\bm{x}_t=\bm{A}_1(\tau_t)\bm{x}_{t-1}+\bm{A}_2(\tau_t)\bm{x}_{t-2}+\bm{\eta}_t,
\end{equation}
where $\bm{A}_2(\cdot)$ and $\bm{\eta}_t$ are generated in the same way as in Section \ref{Sec4.2}, and
$$
	\bm{A}_1(\tau)=\left[\begin{matrix}
		0.4 & -0.1 \\
		-0.1 & 0.4
	\end{matrix} \right]+b\times d_T\times \left[\begin{matrix}
		2\exp(\tau-1)-1  & \exp(\tau-1)-1  \\
		\exp(\tau-1)-1   & 2\exp(\tau-1)-1
	\end{matrix} \right]
$$
in which $d_T = T^{-1/2}h^{-1/4}$ and $b$ is set to be $0$, $2$ or $4$ in order to investigate the size and local power of the proposed test. We use the proposed testing procedure to test whether the coefficient $\bm{A}_1(\cdot)$ is time-varying.

Again, we let $T\in \{200,400,800\}$ and conduct $1000$ replications for each choice of $T$. We use the simulation-assisted testing procedure to get the empirical critical value $\widehat{q}_{1-\alpha}$ after $1000$ bootstrap replications. We consider a sequence of bandwidths to check the robustness of the proposed test with respect to the bandwidth choice:
\begin{equation}\label{h_test}
	h = \alpha_1 T^{-1/5}, \quad \alpha_1= 0.2,0.4,\ldots,1.8.
\end{equation}

Table \ref{table_sim2} reports the rejection rates at the $5\%$ and $10\%$ nominal levels. A few facts emerge from the table. First, our test has reasonable sizes using empirical critical values obtained by the  bootstrap procedure. Second, the size behaviour of our test is not sensitive to the choices of bandwidths. As discussed in \cite{gao2008bandwidth}, the estimation-based optimal bandwidths may also be optimal for testing purposes, so for simplicity one can use the rule-of-thumb in practice. Third, the local power of our test increases rapidly as $b$ increases.
{\small
	\begin{table}[h]
		\caption{Size and power evaluation}\label{table_sim2}
		\vspace{-0.7cm}
		\begin{center}
			\begin{tabular}{c c c ccc c ccc}
				\hline\hline
				& & &\multicolumn{3}{c}{$5\%$}& &\multicolumn{3}{c}{$10\%$}\\
				\cline{4-6} \cline{8-10}
				&\text{Bandwidth}&T &200 &400 &800  &   &200 &400 &800 \\
				\hline
				\multirow{9}{*}{\shortstack{$b=0$\\(size)}}
				&$0.2T^{-1/5}$ & &  0.046 & 0.060  & 0.058&  &  0.081 & 0.113  & 0.117 \\
				&$0.4T^{-1/5}$ & &  0.044 & 0.050  & 0.066&  &  0.095 & 0.093  & 0.104 \\
				&$0.6T^{-1/5}$ & &  0.057 & 0.070  & 0.046&  &  0.110 & 0.121  & 0.102 \\
				&$0.8T^{-1/5}$ & &  0.056 & 0.066  & 0.050&  &  0.113 & 0.131  & 0.105 \\
				&$1.0T^{-1/5}$ & &  0.033 & 0.060  & 0.047&  &  0.066 & 0.104  & 0.100 \\
				&$1.2T^{-1/5}$ & &  0.055 & 0.042  & 0.042&  &  0.115 & 0.093  & 0.086 \\
				&$1.4T^{-1/5}$ & &  0.039 & 0.045  & 0.043&  &  0.085 & 0.103  & 0.078 \\
				&$1.6T^{-1/5}$ & &  0.065 & 0.047  & 0.038&  &  0.114 & 0.120  & 0.088 \\
				&$1.8T^{-1/5}$ & &  0.041 & 0.066  & 0.046&  &  0.094 & 0.117  & 0.088 \\
				\hline
				\multirow{9}{*}{\shortstack{$b=2$\\(local power)}}
				&$0.2T^{-1/5}$ & & 0.095  & 0.103  & 0.143&  & 0.177  & 0.189 & 0.242 \\
				&$0.4T^{-1/5}$ & & 0.145  & 0.155  & 0.162&  & 0.204  & 0.247  & 0.236 \\
				&$0.6T^{-1/5}$ & & 0.148  & 0.129  & 0.170&  & 0.208  & 0.255  & 0.239 \\
				&$0.8T^{-1/5}$ & & 0.124  & 0.168  & 0.168&  & 0.233  & 0.258  & 0.240 \\
				&$1.0T^{-1/5}$ & & 0.144  & 0.176  & 0.145&  & 0.230  & 0.250  & 0.225 \\
				&$1.2T^{-1/5}$ & & 0.123  & 0.189  & 0.158&  & 0.196  & 0.274  & 0.246 \\
				&$1.4T^{-1/5}$ & & 0.160  & 0.173  & 0.183&  & 0.246  & 0.275  & 0.279 \\
				&$1.6T^{-1/5}$ & & 0.142  & 0.139  & 0.189&  & 0.208  & 0.218  & 0.292 \\
				&$1.8T^{-1/5}$ & & 0.172  & 0.170  & 0.194&  & 0.269  & 0.268  & 0.287 \\
				\hline
				\multirow{9}{*}{\shortstack{$b=4$\\(local power)}}
				&$0.2T^{-1/5}$ & & 0.477  & 0.448  & 0.655&  & 0.635  & 0.604  & 0.783 \\
				&$0.4T^{-1/5}$ & & 0.340  & 0.561  & 0.648&  & 0.513  & 0.705  & 0.758 \\
				&$0.6T^{-1/5}$ & & 0.428  & 0.582  & 0.708&  & 0.583  & 0.709  & 0.803 \\
				&$0.8T^{-1/5}$ & & 0.506  & 0.579  & 0.640&  & 0.615  & 0.701  & 0.773 \\
				&$1.0T^{-1/5}$ & & 0.501  & 0.579  & 0.585&  & 0.598  & 0.699  & 0.695 \\
				&$1.2T^{-1/5}$ & & 0.473  & 0.527  & 0.620&  & 0.592  & 0.657  & 0.728 \\
				&$1.4T^{-1/5}$ & & 0.482  & 0.532  & 0.553&  & 0.614  & 0.648  & 0.655 \\
				&$1.6T^{-1/5}$ & & 0.519  & 0.559  & 0.583&  & 0.617  & 0.664  & 0.665 \\
				&$1.8T^{-1/5}$ & & 0.460  & 0.523  & 0.613&  & 0.584  & 0.653  & 0.720 \\
				\hline\hline
			\end{tabular}
		\end{center}
	\end{table}
}

\section{Empirical Study}\label{Sec5}

In this section, we discuss the transmission mechanism of the monetary policy. As well documented, inflation is higher and more volatile in the U.S. during 1970-1980, but substantially decreases in the subsequent period, which is often referred to as the Great Moderation (\citealp{primiceri2005time}). Two explanations (i.e., bad policy and bad luck) have been debated repeatedly in the literature. The first explanation focuses on the changes in the transmission mechanism \cite[e.g.,][]{cogley2005drifts}, while the second regards it as a consequence of changes in the size of exogenous shocks \cite[e.g.,][]{sims2006}. In what follows, we revisit the arguments associated with the Great Moderation using our approach. The estimation is conducted in exactly the same way as in Section \ref{Sec4}, so we no longer repeat the details.

First, we estimate the time-varying VAR$(p)$ model using three commonly adopted macroeconomic variables of the literature (\citealp{primiceri2005time,cogley2010inflation}), which are the inflation rate (measured by 100 times the year-over-year log change in the GDP deflator), the unemployment rate representing the non-policy block, and the interest rate (measured by the average value for the Federal funds rates over the quarter) representing the monetary policy block. Following \cite{primiceri2005time}, the short-run timing restrictions are employed to identify the monetary policy shocks. The interest rate is included as the third variable and being treated as the monetary policy instrument. The identification requirement is that the monetary policy actions affect the inflation and the unemployment with at least one period of lag ({i.e., short-run timing restrictions}). The data are quarterly observations measured at an annual rate from 1954:Q3 to 2015:Q4, which are collected from the Federal Reserve Bank of St. Louis economic database. Figure \ref{Fg1} plots the three variables.

\begin{figure}[H]
\centering
{\includegraphics[width=18cm]{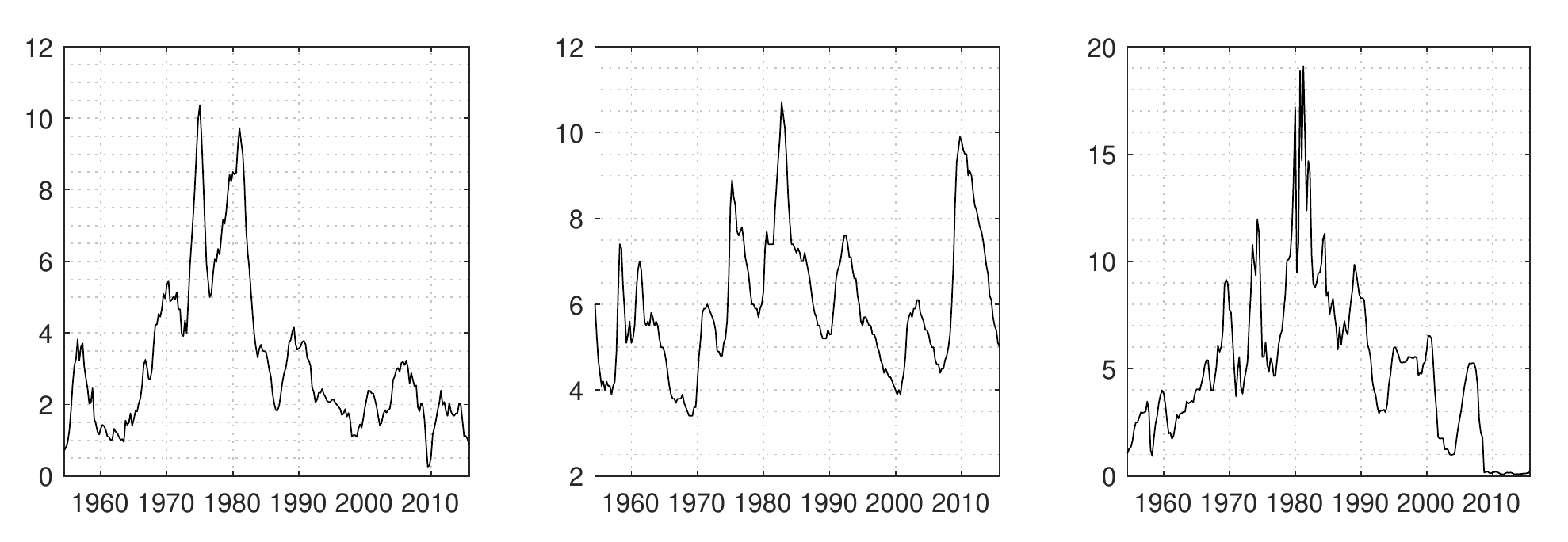}}
\caption{\small Plots of the inflation (left), the unemployment rate (middle) and the interest rate (right)}\label{Fg1}
\end{figure}

For the time-varying VAR$(p)$ model, the optimal lag is $\widehat{\mathsf{p}}=3$ by our approach, while it is often assumed to be known with the values varying from $2$ to $4$ in the literature. Thus, the following analyses focus on the time-varying VAR$(3)$ model (referred to as TV-VAR$(3)$ hereafter). We then conduct a robustness check to see whether the innovation process $\bm{\epsilon}_t$ exhibits serial correlation. We use the Breusch-Godfrey LM test \cite[cf.,][]{breusch1978testing,godfrey1978testing} for testing the first-order autocorrelation of $\bm{\epsilon}_t$'s. The null hypothesis is $H_0: E\left(\bm{\epsilon}_t\bm{\epsilon}_{t-1}^\top\right) = 0$. Based on the estimates $\widehat{\bm{\epsilon}}_t = \widehat{\bm{\omega}}^{-1}(\tau_t)\widehat{\bm{\eta}}_t$, the $p$-value is 0.81 suggesting that the TV-VAR$(3)$ model fits the data quite well.

Before investigating the time-varying effects of monetary policy on real economy, it is reasonable to check whether the VAR coefficients (i.e., the policy transmission mechanism) are time-varying. To this end, we employ the proposed test statistic to examine the constancy of model coefficients, and summarize the results in Table \ref{table_em1}. Specifically, we first set $\bm{C}$ to be an identity matrix for selecting between parametric and time-varying VAR models. From Table \ref{table_em1} (the row labelled ``Constancy''), we conclude that the VAR coefficients are not constant,  implying that there exists significant time-variations in policy transmission mechanism. We then apply the testing procedures to examine whether the selected variables (i.e., intercept, $\bm{x}_{t-1}$, $\bm{x}_{t-2}$, $\bm{x}_{t-3}$) have time-varying contributions. By Table \ref{table_em1},  at the 5\% significance level we conclude that all $\bm{a}(\cdot)$, $\bm{A}_1(\cdot)$, $\bm{A}_2(\cdot)$ and $\bm{A}_3(\cdot)$ are time-varying by examining them individually\footnote{Certainly, one may examine each element of these matrices. However, it will lead to a quite lengthy presentation. In order not to deviate from our main goal, we no longer conduct more testing along this line.}, suggesting that the TV-VAR(3) is more appropriate than a constant VAR model. 

\begin{table}[H]
	\caption{Summary statistics of the test}\label{table_em1}
	\vspace{-0.7cm}
	\begin{center}
		\begin{tabular}{c cc}
			\hline\hline
			&\text{test statistic} & \text{p-value}\\
			\hline
			\text{Constancy}  &  341.27 & 0.00  \\
			$\bm{a}(\cdot)$   &  50.15  & 0.00  \\
			$\bm{A}_1(\cdot)$ &  63.02  & 0.00  \\
			$\bm{A}_2(\cdot)$ &  23.27  & 0.00  \\
			$\bm{A}_3(\cdot)$ &  17.23  & 0.01  \\
			\hline\hline
		\end{tabular}
	\end{center}
\end{table}

We now consider the changes in the size of exogenous shocks. Figure \ref{Fg2} plots the estimated residuals $\widehat{\bm{\eta}}_t$. As can be seen, the magnitudes of estimated residuals change over time. For example, the magnitudes of the estimated residuals in the interest rate component increase from 1970 to 1980, but decrease quickly after 1980. The evidence favours the TV-VAR$(3)$ model which accounts for heteroscedasticity. More precisely, Figure \ref{Fg3} plots the time-varying standard deviation of the estimated residuals $\widehat{\bm{\eta}}_t$ as well as the associated 95\% confidence intervals. We can see that a decline in unconditional volatilities of exogenous shocks after 1980 in every subplot of Figure \ref{Fg3}. Thus, our results support the explanation of ``bad luck'' \cite[e.g.,][]{primiceri2005time,sims2006}.

\begin{figure}[H]
\centering
{\includegraphics[width=18cm]{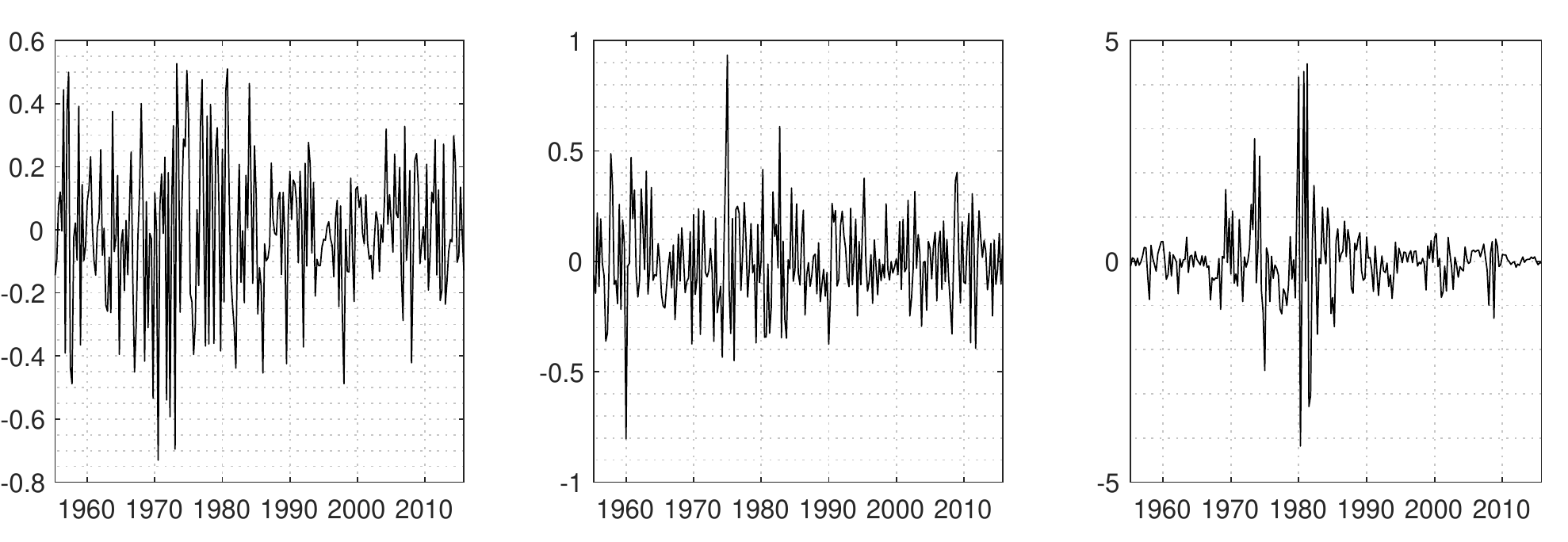}}
\caption{\small Estimation residuals in the inflation equation (left), the unemployment equation (middle) and the interest rate equation (right).}\label{Fg2}
\end{figure}

We then discuss how inflation and unemployment respond to monetary policy shocks. Figure \ref{Fg4} plots the time-varying cumulative impulse responses\footnote{The cumulative impulse responses are constructed in exactly the same way as in \cite{primiceri2005time}, and the standard errors are computed based on Theorem \ref{Thm2.2}.} of inflation and unemployment to a monetary shock subject to the short-run timing restrictions. Clearly, these responses vary over time, indicating a substantial time-variation in the policy transmission mechanism. The result differs from \cite{primiceri2005time} that has found no evidence of time variation in the response of the economy to monetary policy shocks using a Bayesian approach. Interestingly, our results show that the responses of inflation exhibit a price puzzle during 1970-1980, which suggests that inflation increases in response to a monetary tightening although standard macroeconomic theory predicts the opposite. Therefore, our results also support the  explanation of ``bad policy'' \cite[e.g.,][]{cogley2005drifts}.

\begin{figure}[H]
\centering
{\includegraphics[width=18cm]{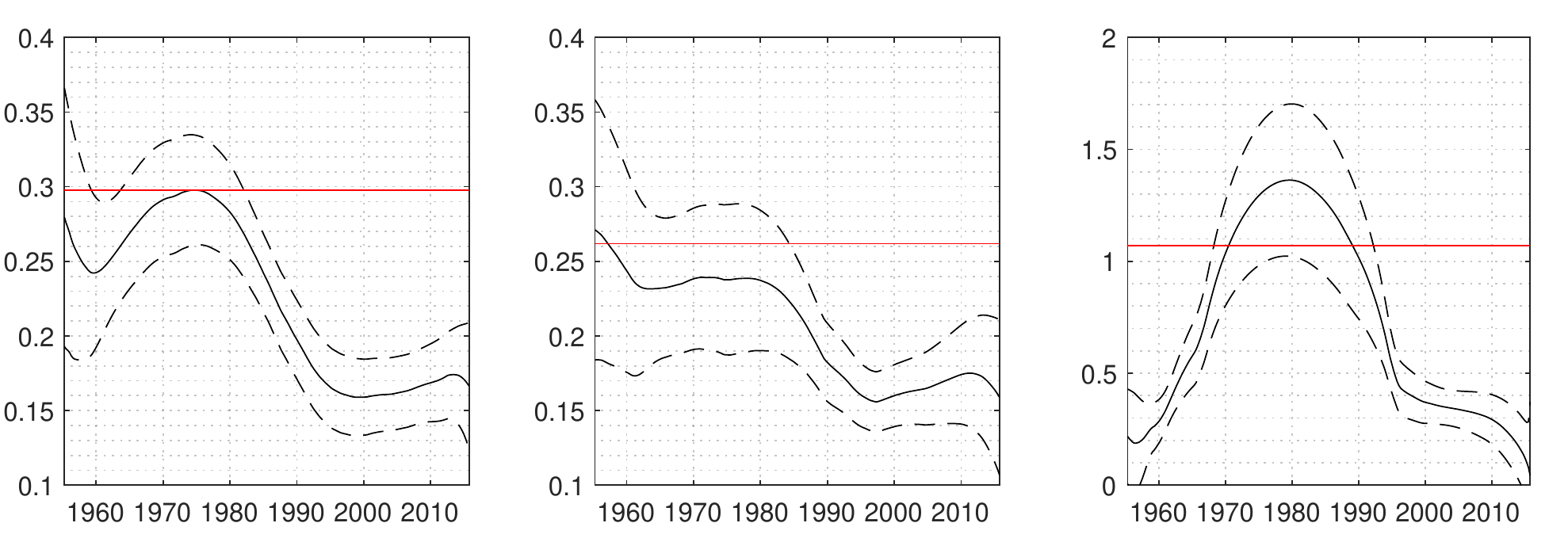}}
\caption{\small The standard deviation of estimation residuals $\widehat{\bm{\eta}}_t$ in the inflation equation (left), the unemployment equation (middle) and the interest rate equation (right) as well as the associated 95\% confidence intervals. The red line denotes the estimated standard deviation using the VAR(3) model with constant parameters.}\label{Fg3}
\end{figure}

\begin{figure}[H]
\centering
{\includegraphics[width=8cm, height=7cm]{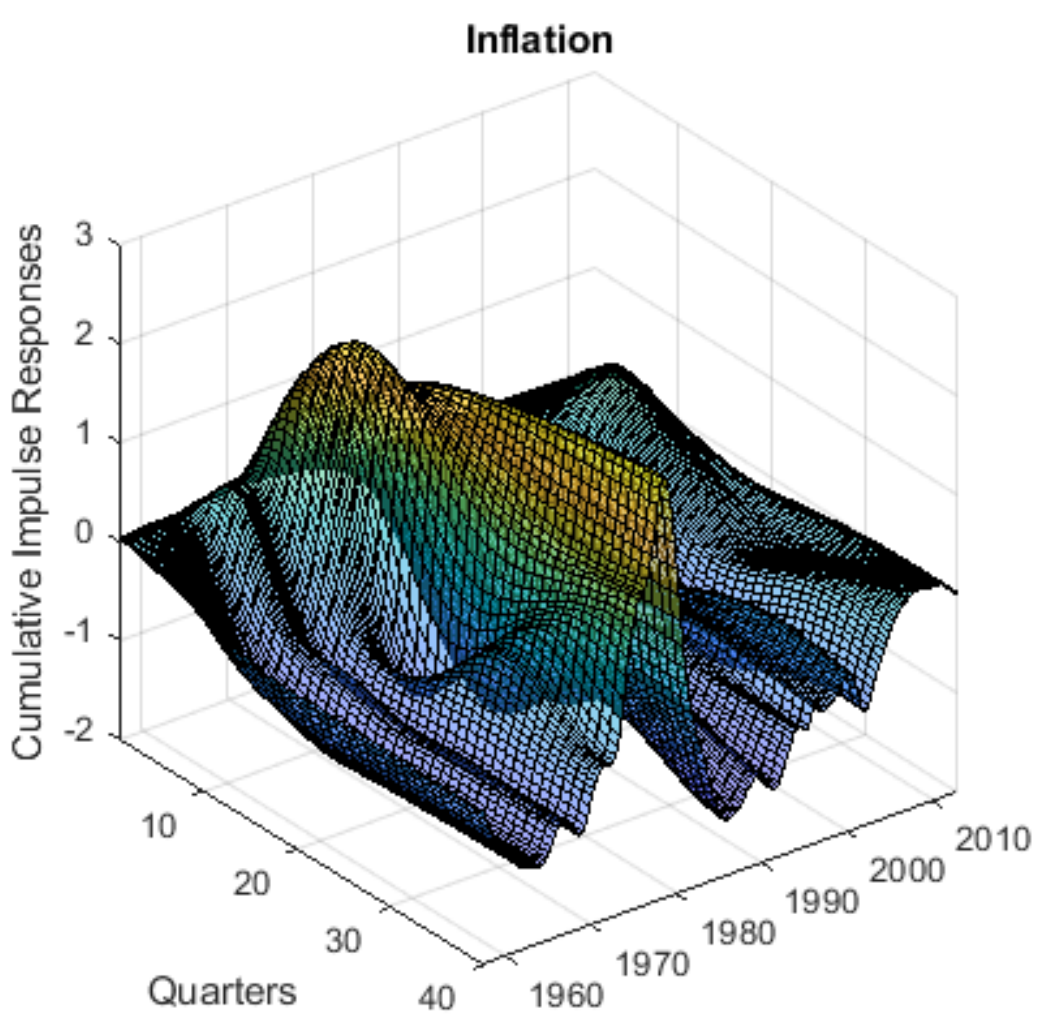}}
{\includegraphics[width=8cm, height=7cm]{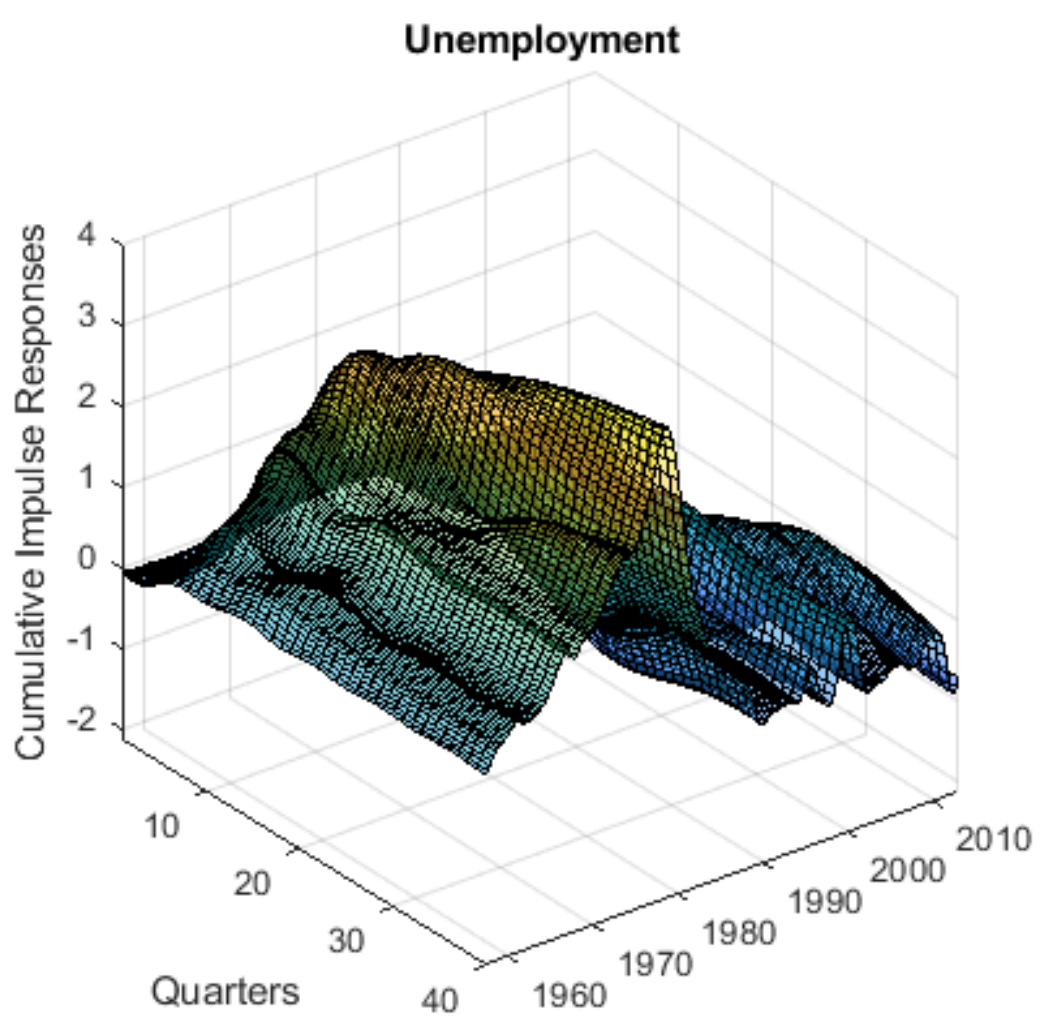}}\\
\caption{\small Time-varying cumulative impulse responses to the monetary policy shocks. Left front axis left: Quarters (Horizon). Right front axis left: Time.}\label{Fg4}
\end{figure}

Additionally, Figure \ref{Fg5} plots the impulse responses of inflation to a monetary policy shock in 1978:Q1 and 1994:Q1. The left panel in Figure \ref{Fg5} exhibits the price puzzle, while the right panel shows a response that corresponds to the mainstream prior: the price level declines soon after a tightening. This result shows that the price puzzle is limited to periods of bad monetary policy, and is consistent with a stream of literature \cite[e.g.,][]{elbourne2006financial,rusnak2013solve}, which suggests that the price puzzle emerges when using data for different monetary policy regimes and VARs estimated within a period of single monetary policy rarely encounter the price puzzle. Hence, the existence of time-varying policy transmission mechanism favours the TV-VAR(3) model.

\begin{figure}[H]
  \centering
  \includegraphics[width=18cm]{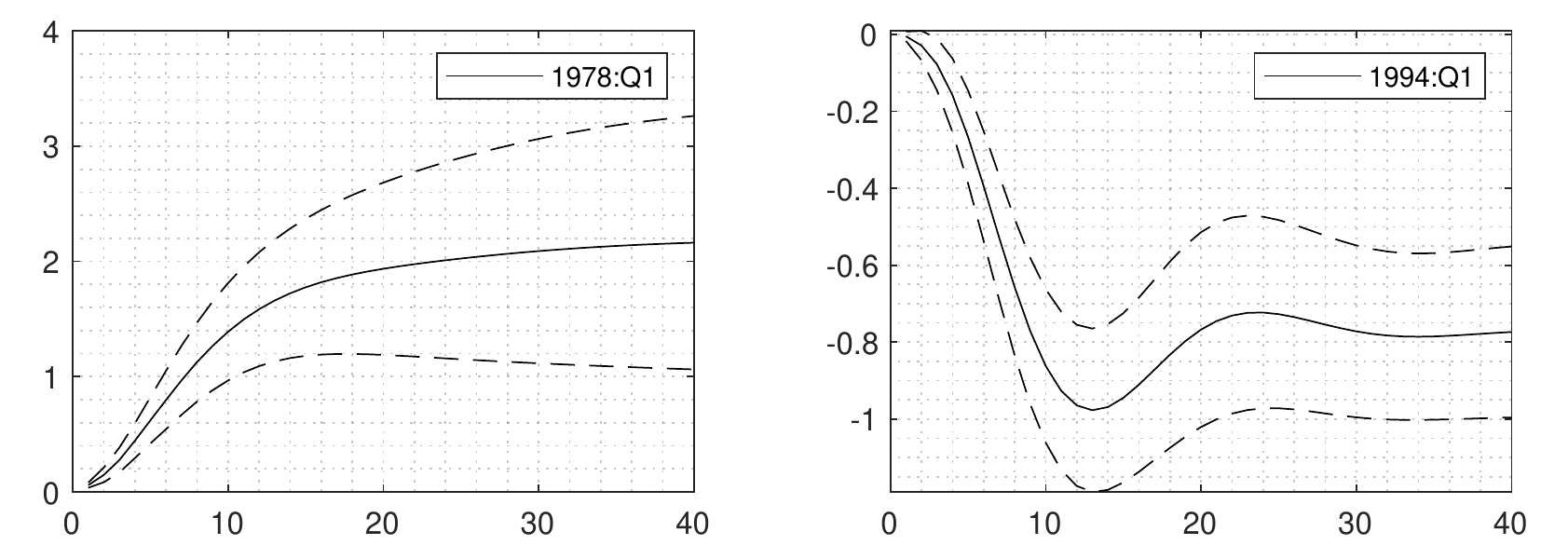}
  \caption{\small Impulse responses of inflation to the monetary policy shocks in 1978:Q1 and 1994:Q1 with one-standard error confidence intervals.}\label{Fg5}
\end{figure}

\section{Conclusions and Discussion}\label{Sec6}

\renewcommand{\theequation}{6.\arabic{equation}}

\setcounter{equation}{0}
In this paper, we have proposed a new class of time-varying VAR models where the VAR coefficients and covariance matrix of the error innovations are allowed to evolve over time. Accordingly, we have established a set of asymptotic results, including the impulse responses analyses subject to both short-run timing and long-run restrictions, an information criterion to select the optimal lag, and a Wald--type test to determine the constant coefficients. Simulation studies are conducted to evaluate the theoretical findings. Finally, we demonstrate the empirical relevance and usefulness of the proposed methods through an application to the transmission mechanism of U.S. monetary policy. We have revealed that there exists a substantial time-variation in the policy transmission mechanism, and the ``price puzzle'' is limited to periods of bad monetary policy.

There are several directions for possible extensions. The first one is about how to consistently estimate the $d$-dimensional components of the VAR$(p)$ process for the case where the dimensionality, $d$, and the number of lags, $p$, may diverge along with the sample size, $T$. The second one is to allow for some time-varying structure in cointegrated dynamic models. We wish to leave such issues for future study. 

\newpage

\section*{Appendix A}

\renewcommand{\theequation}{A.\arabic{equation}}
\renewcommand{\thesection}{A.\arabic{section}}
\renewcommand{\thefigure}{A.\arabic{figure}}
\renewcommand{\thetable}{A.\arabic{table}}
\renewcommand{\thelemma}{A.\arabic{lemma}}
\renewcommand{\theassumption}{A.\arabic{assumption}}
\renewcommand{\thetheorem}{A.\arabic{theorem}}
\renewcommand{\theproposition}{A.\arabic{proposition}}

\setcounter{equation}{0}
\setcounter{lemma}{0}
\setcounter{section}{0}
\setcounter{table}{0}
\setcounter{figure}{0}
\setcounter{assumption}{0}
\setcounter{proposition}{0}

{\small

For the sake of presentation, we first provide some notations and mathematical symbols in Appendix \ref{ApA.1}. The proofs of main results are provided in Appendix \ref{ApA.2}. In what follows, $M$ and $O(1)$ always stand for constants, and may be different at each appearance.

\section{Notations and Mathematical Symbols}\label{ApA.1}

For ease of notation, we define three matrices $\bm{\Sigma}(\tau)$, $\bm{V}(\tau)$ and $\bm{\Phi}(\tau)$ with their estimators respectively. For $\forall \tau\in (0,1)$, let

\begin{eqnarray}\label{EqA.1}
 \bm{\Sigma}(\tau)&=&\left[\begin{matrix}
      1          & \bm{\mu}^\top(\tau)         & \cdots & \bm{\mu}^\top(\tau) \\
      \bm{\mu}(\tau)  & \bm{\Sigma}_{0}(\tau)     & \cdots & \bm{\Sigma}_{p-1}^\top(\tau) \\
      \vdots     & \vdots                 & \ddots & \vdots \\
      \bm{\mu}(\tau)  & \bm{\Sigma}_{p-1}(\tau)   & \cdots & \bm{\Sigma}_{0}(\tau)
    \end{matrix}\right],
\end{eqnarray}
in which $\bm{\mu}(\tau)$ and $\bm{B}_j(\tau)$ are defined in Proposition \ref{Proposition2.1} and $\bm{\Sigma}_m(\tau)=\bm{\mu}(\tau)\bm{\mu}(\tau)^\top+\sum_{j=0}^{\infty}\bm{B}_j(\tau)\bm{B}_{j+m}^\top(\tau)$ for $m=0,\ldots,p-1$. We define the estimator of $\bm{\Sigma}(\tau)$ as

\begin{eqnarray}\label{EqA.2}
\widehat{\bm{\Sigma}}(\tau)=\left(\frac{1}{T}\sum_{t=1}^{T} K_h(\tau_t-\tau)\right)^{-1}\frac{1}{T}\sum_{t=1}^{T}\bm{z}_{t-1}\bm{z}_{t-1}^\top K_h(\tau_t-\tau),
\end{eqnarray}
where $\bm{z}_t$ is defined in \eqref{Eq2.4}.

\medskip

Next, we let
\begin{eqnarray}\label{EqA.3}
\bm{V}(\tau)=\left[\begin{matrix}
                 \bm{V}_{1,1}(\tau) & \bm{V}_{2,1}^\top(\tau) \\
                 \bm{V}_{2,1}(\tau) & \bm{V}_{2,2}(\tau)
               \end{matrix} \right],
\end{eqnarray}
where $\bm{V}_{1,1}(\tau) = \tilde{v}_0\bm{\Sigma}^{-1}(\tau)\otimes \bm{\Omega}(\tau)$,
\begin{eqnarray}\label{EqA.4}
\bm{V}_{2,1}(\tau)&=& \lim_{T\rightarrow \infty}\frac{h}{T}\sum_{t=1}^{T} E\left( \mathrm{vech} (\bm{\eta}_t\bm{\eta}_t^\top )\bm{\eta}_t^\top\bm{Z}_{t-1}^\top \right)K_h (\tau_t-\tau)^2\cdot  (\bm{\Sigma}^{-1}(\tau)\otimes \bm{I}_d ) , \nonumber \\
 \bm{V}_{2,2}(\tau)&=&\lim_{T\rightarrow\infty} \frac{h}{T}\sum_{t=1}^{T}E\left(\mathrm{vech}(\bm{\eta}_t\bm{\eta}_t^\top)\mathrm{vech}(\bm{\eta}_t\bm{\eta}_t^\top)^\top \right) K_h (\tau_t-\tau)^2 \nonumber \\
 &&-\tilde{v}_0\mathrm{vech}\left(\bm{\Omega}(\tau)\right) \mathrm{vech}\left(\bm{\Omega}(\tau)\right)^\top.
\end{eqnarray}

The estimator of $\bm{V}(\tau)$ is then defined as follows:
\begin{eqnarray}\label{EqA.5}
\widehat{\bm{V}}(\tau)&=&\left[\begin{matrix}
\widehat{\bm{V}}_{1,1}(\tau) & \widehat{\bm{V}}_{2,1}^\top(\tau) \\
\widehat{\bm{V}}_{2,1}(\tau) & \widehat{\bm{V}}_{2,2}(\tau)
  \end{matrix}\right],
\end{eqnarray}
where $\bm{\widehat{V}}_{1,1}(\tau)$, $\bm{\widehat{V}}_{2,1}(\tau)$ and $\bm{\widehat{V}}_{2,2}(\tau)$ have the forms identical to their counterparts of \eqref{EqA.4}, but we replace $ \bm{\Sigma}(\tau)$, $ \bm{\eta}_t$ and $\bm{\Omega}(\tau)$ with their estimators presented in \eqref{EqA.2} and \eqref{Eq2.5}.

Finally, recall the following definition:
\begin{eqnarray}\label{EqA.6}
\bm{\Phi}(\tau)=\left[\begin{matrix}
       \bm{A}_{1}(\tau) & \cdots & \bm{A}_{p-1}(\tau) & \bm{A}_{p}(\tau)  \\
       \bm{I}_d & \cdots& \bm{0}_d & \bm{0}_d\\
       \vdots & \ddots&\vdots & \vdots\\
       \bm{0}_d &\cdots &\bm{I}_d & \bm{0}_d\\
    \end{matrix} \right].
\end{eqnarray}

\noindent Replacing $\bm{A}_j(\tau)$'s of \eqref{EqA.6} with their estimators obtained from \eqref{Eq2.5} yields an estimator, $\widehat{\bm{\Phi}}(\tau)$, for $\bm{\Phi}(\tau)$.

\section{Proofs of the Main Results}\label{ApA.2}

\begin{proof}[Proof of Proposition \ref{Proposition2.1}]

\item
Consider the VMA representation of $\bm{x}_t$: \, $\bm{x}_t= \bm{\mu}_{t} +\bm{B}_{0,t}\bm{\epsilon}_t+\bm{B}_{1,t}\bm{\epsilon}_{t-1}+\bm{B}_{2,t}\bm{\epsilon}_{t-2}+\cdots$,
where  $\bm{B}_{0,t}=\bm{\omega}(\tau_t)$, $\bm{B}_{j,t} =\bm{\Psi}_{j,t}\bm{\omega}(\tau_{t-j})$,
$\bm{\Psi}_{j,t}=\bm{J}\prod_{m=0}^{j-1}\bm{\Phi}(\tau_{t-m}) \bm{J}^\top$ for $j\geq 1$, $\bm{\mu}_t=\bm{a}(\tau_t)+\sum_{j=1}^{\infty}\bm{\Psi}_{j,t} \bm{a}(\tau_{t-j})$ and $\tau_{t-j}=\frac{t-j}{T}I(t\geq j)$.

First, we investigate the validity of the VMA representations of $\bm{x}_t$ and $\widetilde{\bm{x}}_t$. Let $\rho_A$ denote the largest eigenvalue of $\bm{\Phi}(\tau)$ uniformly over $\tau\in [0,1]$. Then, $\rho_A<1$ by Assumption 4.1. Similar to the proof of Proposition 2.4 in \cite{dahlhaus2009empirical}, we have $\max_{t\geq 1} \left\|\prod_{m=0}^{j-1}\bm{\Phi}(\tau_{t-m}) \right\| \leq M \rho_A^j$. It follows that $\left\|{\rm E}(\bm{x_t})\right\| \leq \sum_{j=0}^{\infty} \left\|\bm{\Psi}_{j,t} \right\|\cdot \left\|\bm{a}(\tau_{t-j})\right\| \leq M \sum_{j=0}^{\infty} \rho_A^j < \infty$ and
\begin{eqnarray*}
  \left\|{\rm Var}(\bm{x_t})\right\| &=& \left\|\sum_{j=0}^{\infty}\bm{B}_{j,t}\bm{B}_{j,t}^\top\right\| \leq \sum_{j=0}^{\infty}\left\|\bm{B}_{j,t}\right\|^2 \leq M \sum_{j=0}^{\infty}\rho_A^{2j} < \infty.
\end{eqnarray*}
Similarly, we have $\left\|{\rm E}\left(\widetilde{\bm{x}}_t \right)\right\| < \infty$ and $\left\|{\rm Var}\left(\widetilde{\bm{x}}_t \right)\right\| < \infty$.

Then, we need to verify that $\max_{t\geq 1}\{E\left\|\bm{x}_t-\widetilde{\bm{x}}_t \right\|^\delta\}^{1/\delta}=O(T^{-1})$. For any conformable matrices $\{\bm{A}_i\}$ and $\{\bm{B}_i\}$, since
$\prod_{i=1}^{r}\bm{A}_i-\prod_{i=1}^{r}\bm{B}_i=\sum_{j=1}^{r}\left( \prod_{k=1}^{j-1}\bm{A}_k \right)\left(\bm{A}_j-\bm{B}_j\right)\left(\prod_{k=j+1}^{r}\bm{B}_k\right)$, we then obtain 
\begin{eqnarray*}
  &&\left\|\bm{B}_{j,t} -\bm{B}_j(\tau_t)\right\| = \left\|\bm{J}\prod_{m=0}^{j-1}\bm{\Phi}(\tau_{t-m})\bm{J}^\top\bm{\omega}(\tau_{t-j})-\bm{J}\bm{\Phi}^j(\tau_t)\bm{J}^\top\bm{\omega}(\tau_t)\right\|\\
   &&= \left\| \left(\bm{J}\prod_{m=0}^{j-1}\bm{\Phi}(\tau_{t-m})\bm{J}^\top-\bm{J}\bm{\Phi}^j(\tau_t)\bm{J}^\top\right)\bm{\omega}(\tau_t)+\bm{J}\prod_{m=0}^{j-1}\bm{\Phi}(\tau_{t-m})\bm{J}^\top\left( \bm{\omega}(\tau_{t-j})-\bm{\omega}(\tau_t)\right)\right\|\\
   &&\leq M\sum_{i=1}^{j-1}\left\|\bm{\Phi}^i(\tau_t)(\bm{\Phi}(\tau_{t-i})-\bm{\Phi}(\tau_t))\prod_{m=i+1}^{j-1}\bm{\Phi}(\tau_{t-m}) \right\|+M \rho_A^j \frac{j}{T} \leq M\sum_{i=1}^{j-1}\frac{i}{T}\rho_A^{j-1}+M \rho_A^j \frac{j}{T}=O(T^{-1}),
\end{eqnarray*}
which implies for the same $\delta>4$ as in Assumption 2,
\begin{eqnarray*}
\{E\left\|\bm{x}_t-\widetilde{\bm{x}}_t\right\|^\delta\}^{1/\delta} &\leq &\sum_{j=1}^{\infty}\left\|\bm{\Psi}_{j,t}\bm{a}(\tau_{t-j})-\bm{\Psi}_j(\tau_t)\bm{a}(\tau_t)\right\|+ \sum_{j=1}^{\infty}\left\|\bm{B}_{j,t} -\bm{B}_j(\tau_t)\right\| \cdot \{E\left\|\bm{\epsilon}_t\right\|^\delta\}^{1/\delta}\\
   &\leq& M\sum_{j=1}^{\infty} \left(\sum_{i=1}^{j-1}\frac{i}{T}\rho_A^{j-1}+ \rho_A^j \frac{j}{T}\right) =O\left(T^{-1}\right).
\end{eqnarray*}

The proof is now completed.
\end{proof}

\begin{proof}[Proof of Theorem \ref{Thm2.1}]
\item
	
\noindent (1). For notation simplicity, let $\bm{S}_{T,k}(\tau)=\frac{1}{T}\sum_{t=1}^{T}\bm{z}_{t-1} \bm{z}_{t-1}^{\top}\left(\frac{\tau_t-\tau}{h}\right)^k K_h(\tau_t-\tau)$,
	$$
	\bm{S}_{T}(\tau)=\left(\begin{matrix}
		\bm{S}_{T,0}(\tau) & \bm{S}_{T,1}(\tau) \\
		\bm{S}_{T,1}(\tau) & \bm{S}_{T,2}(\tau)
	\end{matrix} \right),
	$$
	and $\bm{M}(\tau_t)=\bm{A}(\tau_t)-\bm{A}(\tau)-\bm{A}^{(1)}(\tau)(\tau_t-\tau)-\frac{1}{2}\bm{A}^{(2)}(\tau)(\tau_t-\tau)^2$.
	
We now begin our investigation. Since
	\begin{eqnarray*}
		\bm{x}_t &=& \left(\bm{A}(\tau)+\bm{A}^{(1)}(\tau)(\tau_t-\tau)+\frac{1}{2}\bm{A}^{(2)}(\tau)(\tau_t-\tau)^2+\bm{M}(\tau_t)\right)\bm{z}_{t-1}+\bm{\eta}_t
		\\
		&=& \bm{Z}_{t-1}^{*,\top}\left[\begin{matrix}
			\mathrm{vec}\left(\bm{A}(\tau)\right) \\
			h\mathrm{vec}\left(\bm{A}^{(1)}(\tau)\right)
		\end{matrix} \right]+\frac{1}{2}h^2\left(\frac{\tau_t-\tau}{h}\right)^2\left(\bm{z}_{t-1}^\top\otimes \bm{I}_d\right)\mathrm{vec}\left(\bm{A}^{(2)}(\tau)\right)\\
		&&+\left(\bm{z}_{t-1}^\top\otimes \bm{I}_d\right)\mathrm{vec}\left(\bm{M}(\tau_t)\right)+\bm{\eta}_t,
	\end{eqnarray*}
	we write
	\begin{eqnarray*}
		&&\mathrm{vec} (\bm{\widehat{A}}(\tau)-\bm{A}(\tau))\\
		&=&[\bm{I}_{d^2p+d},\bm{0}_{d^2p+d}]\cdot\left(\frac{1}{T}\sum_{t=1}^{T} \bm{Z}_{t-1}^*\bm{Z}_{t-1}^{*,\top} K_h(\tau_t-\tau)\right)^{-1}\left(\frac{1}{T}\sum_{t=1}^{T}\bm{Z}_{t-1}^*\bm{x}_t  K_h(\tau_t-\tau)\right)-\mathrm{vec}\left(\bm{A}(\tau)\right)\\
		&=&[\bm{I}_{d^2p+d},\bm{0}_{d^2p+d}]\left(\bm{S}_{T}^{-1}(\tau)\left(\begin{matrix}
			\bm{S}_{T,2}(\tau) \\
			\bm{S}_{T,3}(\tau)
		\end{matrix} \right)\otimes \bm{I}_d\right)\left\{\frac{1}{2}h^2\mathrm{vec}\left(\bm{A}^{(2)}(\tau)\right)\right\}\\
		&& +[\bm{I}_{d^2p+d},\bm{0}_{d^2p+d}]\left(\bm{S}_{T}^{-1}(\tau)\otimes \bm{I}_d\right)\left(\frac{1}{T}\sum_{t=1}^{T}(\bm{z}_{t-1}^*\bm{z}_{t-1}^{\top} \otimes \bm{I}_d )\mathrm{vec}\left(\bm{M}(\tau_t)\right)K_h(\tau_t-\tau)\right)\\
		&& +[\bm{I}_{d^2p+d},\bm{0}_{d^2p+d}]\left(\bm{S}_{T}^{-1}(\tau)\otimes \bm{I}_d\right)\left(\frac{1}{T}\sum_{t=1}^{T}(\bm{z}_{t-1}^*\otimes \bm{I}_d)\bm{\eta}_t K_h(\tau_t-\tau)\right) \\
		&:=& I_{T,1}+I_{T,2}+I_{T,3}.
	\end{eqnarray*}

	By standard arguments for the local linear kernel estimator and the uniform convergence results in Lemmas \ref{LemmaB.3}.2-3, we have $\left\|I_{T,1}+I_{T,2}\right\|=O(h^2)+O_P(h^2\sqrt{\log T/(Th)})$ uniformly over $\tau \in[0,1]$. By Lemma \ref{LemmaB.4}.1, we have $I_{T,3}=O_P \left( (\frac{\log T}{Th} )^{\frac{1}{2}} \right)$ uniformly over $\tau \in [0,1]$. Therefore, the first result follows.
	
	\medskip
	
	\noindent (2). We begin our investigation on the asymptotic normality by writing that for $\forall \tau \in (0,1)$,
	\begin{eqnarray*}
		\widehat{\bm{\Omega}}(\tau) &=&\frac{1}{Th}\sum_{t=1}^{T}\bm{\widehat{\eta}}_t\bm{\widehat{\eta}}_t^\top K\left(\frac{\tau_t-\tau}{h}\right)+O_P\left(\frac{1}{Th}\right) \\
		&=&\frac{1}{Th}\sum_{t=1}^{T}\left(\bm{\eta}_t+\bm{\widehat{\eta}}_t-\bm{\eta}_t\right)\left(\bm{\eta}_t+\bm{\widehat{\eta}}_t-\bm{\eta}_t\right)^\top K\left(\frac{\tau_t-\tau}{h}\right)+O_P\left(\frac{1}{Th}\right)\\
		&=& \frac{1}{Th}\sum_{t=1}^{T}\bm{\eta}_t\bm{\eta}_t^\top K\left(\frac{\tau_t-\tau}{h}\right)+\frac{1}{Th}\sum_{t=1}^{T}(\bm{\widehat{\eta}}_t-\bm{\eta}_t)(\bm{\widehat{\eta}}_t-\bm{\eta}_t)^\top K\left(\frac{\tau_t-\tau}{h}\right) \\
		& &+\frac{1}{Th}\sum_{t=1}^{T}\bm{\eta}_t(\bm{\widehat{\eta}}_t-\bm{\eta}_t)^\top K\left(\frac{\tau_t-\tau}{h}\right)+\frac{1}{Th}\sum_{t=1}^{T}(\bm{\widehat{\eta}}_t-\bm{\eta}_t)\bm{\eta}_t^\top K\left(\frac{\tau_t-\tau}{h}\right)+O_P\left(\frac{1}{Th}\right)\\
		&:=& \bm{I}_{T,4}+\bm{I}_{T,5}+\bm{I}_{T,6}+\bm{I}_{T,7}+O_P\left(\frac{1}{Th}\right).
	\end{eqnarray*}
	
	Let $\rho_T=h^2+\sqrt{\frac{\log T}{Th}}$. By what we have just proved for Theorem 2.1(i), for $\forall \tau \in [0,1]$ we have
	\begin{eqnarray*}
		&&\left\|\frac{1}{Th}\sum_{t=1}^{T}(\bm{\widehat{\eta}}_t-\bm{\eta}_t)(\bm{\widehat{\eta}}_t-\bm{\eta}_t)^\top  K\left(\frac{\tau_t-\tau}{h}\right)\right\|\\
		&\leq& \sup_{\tau_t \in [0,1]}\|\bm{\widehat{A}}(\tau_t)-\bm{A}(\tau_t)\|^2 \cdot \frac{1}{Th}\sum_{t=1}^{T}\left\|\bm{z}_{t-1}\right\|^2 K\left(\frac{\tau_t-\tau}{h}\right)=O_P(\rho_T^2).
	\end{eqnarray*}
	
	By Lemma \ref{LemmaB.4}.2, $\bm{I}_{T,6}$ and $\bm{I}_{T,7}$ are both $o_P((Th)^{-1/2})$. Hence,
	\begin{eqnarray*}
		\sqrt{Th}\left(\frac{1}{Th}\sum_{t=1}^{T}\bm{\widehat{\eta}}_t\bm{\widehat{\eta}}_t^\top K\left(\frac{\tau_t-\tau}{h}\right)-\frac{1}{Th}\sum_{t=1}^{T}\bm{\eta}_t\bm{\eta}_t^\top K\left(\frac{\tau_t-\tau}{h}\right)-o_P(h^4)\right)=o_P(1).
	\end{eqnarray*}
	
	Combined with the convergence results of the sample covariance matrix stated in Lemma \ref{LemmaB.3}, the above development yields that
	\begin{eqnarray*}
		&&\sqrt{Th}\left[\begin{matrix}
			\mathrm{vec}\left(\bm{\widehat{A}}(\tau)-\bm{A}(\tau)-\frac{1}{2}h^2\tilde{c}_2\bm{A}^{(2)}(\tau)\right)+o_P(h^2) \\
			\mathrm{vech}\left(\bm{\widehat{\Omega}}(\tau)-\bm{\Omega}(\tau)-\frac{1}{2}h^2\tilde{c}_2\bm{\Omega}^{(2)}(\tau)\right)+o_P(h^2)
		\end{matrix} \right]\\
		&=&\left[\begin{matrix}
			\left( \bm{\Sigma}^{-1}(\tau)\otimes \bm{I}_d\right)\left(\frac{1}{\sqrt{Th}}\sum_{t=1}^{T}\bm{Z}_{t-1} \bm{\eta}_t
			K\left(\frac{\tau_t-\tau}{h}\right)\right) \\
			\frac{1}{\sqrt{Th}}\sum_{t=1}^{T}\mathrm{vech}\left(\bm{\eta}_t\bm{\eta}_t^\top-\bm{\Omega}(\tau_t) \right)K\left(\frac{\tau_t-\tau}{h}\right)
		\end{matrix} \right]+o_P(1) := \bm{I}_{T,8}+o_P(1).
	\end{eqnarray*}
	
	Below, we focus on $\bm{I}_{T,8}$. First, we show $\text{Var}(\bm{I}_{T,8})\to \bm{V}(\tau)$. Let $\text{Var}(\bm{I}_{T,8})=\left[\begin{matrix}
		\widetilde{\bm{V}}_{1,1}(\tau) & \widetilde{\bm{V}}_{2,1}^\top(\tau) \\
		\widetilde{\bm{V}}_{2,1}(\tau) & \widetilde{\bm{V}}_{2,2}(\tau)
	\end{matrix}\right]$,
	where the definition of each block should be obvious. Moreover, simple algebra shows that $\widetilde{\bm{V}}_{i,j}(\tau) \to  \bm{V}_{i,j}(\tau)$ for $i,j\in\{1,2\}$.
	
	By construction and Assumption 2, $\bm{I}_{T,8}$ is a summation of m.d.s., we thus use Lemma \ref{LemmaB.1} and Cram\'er-Wold device to prove its asymptotic normality. It suffices to show that $\bm{d}^\top \bm{I}_{T,8}\to_D N\left(\bm{0},\bm{d}^\top\bm{V}(\tau)\bm{d}\right)$ for any conformable unit vector $\bm{d}$. Let
	\begin{equation*}
		\bm{Z}_{T,t}(\tau)=\frac{1}{\sqrt{Th}}\bm{d}^\top\left[\begin{matrix}
			\left( \bm{\Sigma}^{-1}(\tau)\otimes \bm{I}_d\right)\left(\bm{Z}_{t-1}\bm{\eta}_t
			K\left(\frac{\tau_t-\tau}{h}\right)\right) \\
			\mathrm{vech}\left(\bm{\eta}_t\bm{\eta}_t^\top-\bm{\Omega}(\tau_t) \right)K\left(\frac{\tau_t-\tau}{h}\right)
		\end{matrix} \right].
	\end{equation*}
	
	By the law of large numbers for martingale differences, we have 
	$$
	\sum_{t=1}^T \bm{Z}_{T,t}^2(\tau) - \sum_{t=1}^T E(\bm{Z}_{T,t}^2(\tau)|\mathcal{F}_{t-1})\to_P 0.
	$$
	Since conditional on $\mathcal{F}_{t-1}$ the third and fourth moments of $\bm{\epsilon}_t$ are identical to the corresponding unconditional moments a.s., by Lemma \ref{LemmaB.3}.1 we can prove that $\sum_{t=1}^T E (\bm{Z}_{T,t}^2(\tau)|\mathcal{F}_{t-1} )\to_P \bm{d}^\top \bm{V}(\tau) \bm{d}$.
	
	Furthermore, for any $\nu > 0$ and $\tau\in(0,1)$, by both Holder's and Markov's inequalities, we have
	\begin{eqnarray*}
	\sum_{t=1}^{T}E\left(\left( \bm{Z}_{T,t}(\tau)\right)^2I\left(| \bm{Z}_{T,t}(\tau)|>\nu\right)  \right)	&\leq&\sum_{t=1}^{T}\left[E|\bm{Z}_{T,t}(\tau)|^{\delta/2}\right]^{4/\delta} \left[\frac{E|\bm{Z}_{T,t}(\tau)|^{\delta/2}|}{\nu^{\delta/2}} \right]^{(\delta-4)/\delta}\\
		&=&O((Th)^{(\delta-4)/4})=o(1).
	\end{eqnarray*}

Thus, the CLT follows. 
	
	\medskip
	
	Finally, we consider $\widehat{\bm V}(\cdot)$. By Lemma \ref{LemmaB.3} and the above proof, we have $\widehat{\bm{V}}_{1,1}(\tau)\to_P \bm{V}_{1,1}(\tau)$. By the uniform convergence results of $\widehat{\bm{A}}(\tau)$, we can replace $\widehat{\bm{\eta}}_t$ with $\bm{\eta}_t$ in the following derivations. Therefore, we have
	\begin{eqnarray*}
		\widehat{\bm{V}}_{2,1}(\tau) &=& \frac{1}{Th}\sum_{t=1}^{T}\mathrm{vech}\left(\bm{\eta}_t\bm{\eta}_t^\top\right)\bm{\eta}_t^\top
		\bm{Z}_{t-1}^\top K^2\left(\frac{\tau_t-\tau}{h}\right)\left( \bm{\Sigma}^{-1}(\tau)\otimes \bm{I}_d\right)+o_P(1)\to_P\bm{V}_{2,1}(\tau),
	\end{eqnarray*}
	and
	\begin{eqnarray*}
		\widehat{\bm{V}}_{2,2}(\tau) &=& \frac{1}{Th}\sum_{t=1}^{T}\mathrm{vech}(\bm{\eta}_t\bm{\eta}_t^\top)\mathrm{vech}(\bm{\eta}_t\bm{\eta}_t^\top)^\top K^2\left(\frac{\tau_t-\tau}{h}\right) -\widetilde{v}_0\mathrm{vech}\left(\bm{\Omega}(\tau)\right)\mathrm{vech}\left(\bm{\Omega}(\tau)\right)^\top+o_P(1)
		\nonumber\\
		& \to_P & \bm{V}_{2,2}(\tau).
	\end{eqnarray*}
	
The proof is now completed.
\end{proof}

\begin{proof}[Proof of Theorem \ref{Thm2.2}]

\item
Let $\bm{A}(\bm{\theta})$ be a real, differentiable, $m\times n$ matrix function of real $p \times 1$ vector $\bm{\theta}$. Define $\triangledown_{\bm{\theta}}\bm{A} = \frac{\partial \mathrm{vec}(\bm{A})}{\partial \bm{\theta}^\top}$, and thus $\mathrm{vec}(\mathrm{d}\bm{A})= \triangledown_{\bm{\theta}}\bm{A} \mathrm{d}\bm{\theta}$.

Let $\bm{\alpha}(\tau)=\mathrm{vec}\left(\bm{A}_1(\tau),...,\bm{A}_p(\tau)\right)$, $\bm{\sigma}(\tau)=\mathrm{vech}\left(\bm{\Omega}(\tau) \right)$ and $\bm{\phi}(\tau)=[\bm{\alpha}^\top(\tau),\bm{\sigma}^\top(\tau)]^\top$. Given the joint distribution of $\mathrm{vec} (\bm{\widehat{A}}(\tau) )$ and $\mathrm{vech}(\bm{\widehat\Omega}(\tau))$ in Theorem \ref{Thm2.1}, Theorem \ref{Thm2.2} can be obtained by the Delta method. By the first-order approximation of $\mathrm{vec}\left(\bm{\widehat{B}}_j(\tau)\right)$ around $\mathrm{vec}\left(\bm{B}_j(\tau)\right)$, we have
$$
\sqrt{Th}\mathrm{vec}\left(\bm{\widehat{B}}_j(\tau)-\bm{B}_j(\tau)\right)\simeq \triangledown_{\bm{\phi}(\tau)}\bm{B}_j(\tau)\sqrt{Th}\left(\widehat{\bm{\phi}}(\tau)-\bm{\phi}(\tau)\right)
$$
and thus
\begin{eqnarray*}
 \sqrt{Th}\left(\mathrm{vec}\left(\bm{\widehat{B}}_j(\tau)-\bm{B}_j(\tau)\right)-\frac{1}{2}h^2\tilde{c}_2\bm{B}_j^{(2)}(\tau)
   +o_P(h^2)\right)\to_D N\left(0,\bm{\Sigma}_{\bm{B}_j}(\tau)\right),
\end{eqnarray*}
where $\bm{B}_j^{(2)}(\tau)$ and $\bm{\Sigma}_{\bm{B}_j}(\tau)$ have been defined in the body of the theorem.

To complete the proof, we first derive an analytic form for the derivative $\triangledown_{\bm{\phi}(\tau)}\bm{B}_j(\tau)$ under each of the identification restrictions. 

We have two sets of restrictions: (a) $d(d+1)/2$ restrictions implied by $\bm{\Omega}(\tau)=\bm{\omega}(\tau)\bm{\omega}^\top(\tau)$ and (b) additional $d(d-1)/2$ structural restrictions based on short-run or long-run restrictions.

Consider type (a) restrictions. We begin by considering $\mathrm{d}\bm{\Omega}(\tau)=\mathrm{d}\bm{\omega}(\tau) \cdot \bm{\omega}^\top(\tau)+\bm{\omega}(\tau) \cdot \mathrm{d}\bm{\omega}^\top(\tau)$. Let $\bm{B}$ and $\bm{C}$ be $n\times q$ and $q\times r$ matrices, respectively. By $\mathrm{vec}(\bm{A}\bm{B}\bm{C})=\bm{C}^\top\otimes\bm{A}\mathrm{vec}(\bm{B})$, $\mathrm{vec}(\bm{A}^\top)=\bm{K}_{m,n}\mathrm{vec}(\bm{A})$ and $\bm{K}_{m,q}(\bm{A}\otimes \bm{C})=(\bm{C}\otimes \bm{A})\bm{K}_{n,r}$, we have $\bm{N}_1(\tau) \mathrm{vec}(\mathrm{d}\bm{\omega}(\tau))=\mathrm{vec}(\mathrm{d}\bm{\Omega}(\tau))$, where $\bm{N}_1(\tau)=(\bm{I}_{d^2}+\bm{K}_{d,d})(\bm{\omega}(\tau)\otimes \bm{I}_{d})$. Let $\bm{D}_1$ be the duplication matrix such that $\mathrm{vec}[\bm{\Omega}(\tau)]=\bm{D}_1\mathrm{vech}[\bm{\Omega}(\tau)]$, which follows that $\bm{N}_1(\tau)\mathrm{vec}(\mathrm{d}\bm{\omega}(\tau))=\bm{D}_1 \mathrm{d}\bm{\sigma}(\tau)$ and
\begin{equation}\label{gradient_a}
  \bm{N}_1(\tau) \triangledown_{\bm{\sigma}(\tau)}\bm{\omega}(\tau) =\bm{D}_1.
\end{equation}

We then illustrate how to combine equation \eqref{gradient_a} with gradient equations from type (b) restrictions in order to compute $\triangledown_{\bm{\phi}(\tau)}\bm{\omega}(\tau)$.

In the case of short-run timing restrictions, because types (a) and (b) restrictions do not involve $\bm{\alpha}$, $\triangledown_{\bm{\phi}(\tau)}\bm{\omega}(\tau)$ has the form $[\bm{0}, \triangledown_{\bm{\sigma}(\tau)}\bm{\omega}(\tau)]$. Let $\bm{L}_d$ be the elimination matrix defined by $\mathrm{vech}[\bm{\omega}(\tau)] = \bm{L}_d \mathrm{vec}[\bm{\omega}(\tau)]$. Because $\bm{\omega}(\tau)$ is lower triangular subject to short-run restrictions, $\bm{L}^\top$ is a duplication matrix such that $\mathrm{vec}[\bm{\omega}(\tau)] = \bm{L}_d^\top \mathrm{vech}[\bm{\omega}(\tau)]$. Write
\begin{eqnarray*}
  \bm{N}_1(\tau) \mathrm{vec}(\mathrm{d}\bm{\omega}(\tau)) &=& \bm{D}_1 \mathrm{d}\bm{\sigma}(\tau), \\
  \bm{L}_d\bm{N}_1(\tau)\bm{L}_d^\top \mathrm{vech}(\mathrm{d}\bm{\omega}(\tau)) &=& \bm{L}_d \bm{D}_1 \mathrm{d}\bm{\sigma}(\tau) =\mathrm{d}\bm{\sigma}(\tau),\\
  \mathrm{vech}(\mathrm{d}\bm{\omega}(\tau)) &=& \left(\bm{L}_d\bm{N}_1(\tau)\bm{L}_d^\top \right)^{-1}\mathrm{d}\bm{\sigma}(\tau),\\
  \mathrm{vec}(\mathrm{d}\bm{\omega}(\tau)) &=& \bm{L}_d^\top \left(\bm{L}_d\bm{N}_1(\tau)\bm{L}_d^\top \right)^{-1}\mathrm{d}\bm{\sigma}(\tau).
\end{eqnarray*}
Hence, $\triangledown_{\bm{\sigma}(\tau)}\bm{\omega}(\tau)= \bm{L}_d^\top \left(\bm{L}_d\bm{N}_1(\tau)\bm{L}_d^\top \right)^{-1}$. Recall that $\triangledown_{\bm{\phi}(\tau)}\bm{B}_j(\tau)=[\triangledown_{\bm{\alpha}(\tau)}\bm{B}_j(\tau),\triangledown_{\bm{\sigma}(\tau)}\bm{B}_j(\tau)]$.
For $\triangledown_{\bm{\alpha}(\tau)}\bm{B}_j(\tau)$,
\begin{eqnarray*}
\triangledown_{\bm{\alpha}(\tau)}\bm{B}_j(\tau) &=& \frac{\partial \mathrm{vec}\left[\bm{\Psi}_j(\tau)\bm{\omega}(\tau)\right]}{\partial \bm{\alpha}^\top(\tau)} =(\bm{\omega}^\top(\tau)\otimes \bm{I}_d) \frac{\partial \mathrm{vec}\left[\bm{\Psi}_j(\tau)\right]}{\partial \bm{\alpha}^\top(\tau)} \\
   &=& (\bm{\omega}^\top(\tau)\otimes \bm{I}_d) \frac{\partial \mathrm{vec}\left[\bm{J}\bm{\Phi}^j(\tau)\bm{J}^\top\right]}{\partial \bm{\alpha}^\top(\tau)} \\
   &=&(\bm{\omega}^\top(\tau)\otimes \bm{I}_d)\left(\sum_{m=0}^{j-1} \bm{J}(\bm{\Phi}^\top(\tau))^{j-1-m}\otimes \bm{\Psi}_m(\tau)\right).
\end{eqnarray*}
For $\triangledown_{\bm{\sigma}(\tau)}\bm{B}_j(\tau)$,
\begin{eqnarray*}
\triangledown_{\bm{\sigma}(\tau)}\bm{B}_j(\tau) &=& \frac{\partial \mathrm{vec}\left[\bm{\Psi}_j(\tau)\bm{\omega}(\tau)\right]}{\partial \bm{\sigma}^\top(\tau)} = (\bm{I}_d\otimes \bm{\Psi}_j(\tau))\frac{\partial \mathrm{vec}\left[\bm{\omega}(\tau)\right]}{\partial \bm{\sigma}^\top(\tau)} \\
   &=&(\bm{I}_d\otimes \bm{\Psi}_j(\tau))\bm{L}_d^\top \left(\bm{L}_d\bm{N}_1(\tau)\bm{L}_d^\top \right)^{-1}.
\end{eqnarray*}

In the case of long-run restrictions, type (b) restrictions involve $\bm{\alpha}(\tau)$, so that $\triangledown_{\bm{\phi}(\tau)}\bm{\omega}(\tau)$ has the form $[\triangledown_{\bm{\alpha}(\tau)}\bm{\omega}(\tau), \triangledown_{\bm{\sigma}(\tau)}\bm{\omega}(\tau)]$. First, equation \eqref{gradient_a} must be extended in the form $\bm{N}_1\triangledown_{\bm{\phi}(\tau)}\bm{\omega}(\tau)=[\bm{0},\bm{D}_1]$. Second, long-run restrictions can be expressed as $\bm{Q}\mathrm{vec}\left[\bm{A}_\tau^{-1}(1)\bm{\omega}\right]=0$, where $\bm{Q}$ is a $d(d-1)/2 \times d^2$ matrix of $0$ and $1$, and $\bm{A}_\tau(1)=\bm{I}_d-\sum_{i=1}^{p}\bm{A}_i(\tau)$. By $\mathrm{d}\bm{A}^{-1}=-\bm{A}^{-1}\cdot\mathrm{d}\bm{A} \cdot \bm{A}^{-1}$, we have
\begin{eqnarray*}
  &&\bm{Q}\mathrm{vec}\left[\bm{A}_\tau^{-1}(1)\bm{\omega}\right] = 0 \\
  &&\bm{Q}\mathrm{vec}\left[\mathrm{d}(\bm{A}_\tau^{-1}(1))\bm{\omega}+\bm{A}_\tau^{-1}(1)\mathrm{d}\bm{\omega}\right] = 0 \\
  &&\bm{Q}\mathrm{vec}\left[-\bm{A}_\tau^{-1}(1)\mathrm{d}(\bm{A}_\tau(1))\bm{A}_\tau^{-1}(1)\bm{\omega}+\bm{A}_\tau^{-1}(1)\mathrm{d}\bm{\omega}\right] = 0 \\
  &&\bm{Q}\left[\bm{I}_d\otimes \bm{A}_\tau^{-1}(1)\right]\mathrm{vec}\left[\mathrm{d}\bm{\omega}\right] = \bm{Q}[\bm{B}^\top(\tau)\otimes\bm{A}_\tau^{-1}(1)]\mathrm{vec}[\mathrm{d}\bm{A}_\tau(1)]\\
  && \bm{N}_2(\tau)  \triangledown_{\bm{\phi}(\tau)}\bm{\omega}(\tau) = [\bm{D}_{2}(\tau), \bm{0}],
\end{eqnarray*}
where $\bm{N}_2(\tau)=\bm{Q}\left[\bm{I}_d\otimes \bm{A}_\tau^{-1}(1)\right]$ and $\bm{D}_{2}(\tau)=\bm{Q}[\bm{B}^\top(\tau)\otimes\bm{A}_\tau^{-1}(1)]\triangledown_{\bm{\alpha}(\tau)}\bm{A}_{\tau}(1)$ with $\triangledown_{\bm{\alpha}(\tau)}\bm{A}_{\tau}(1)=-[\bm{I}_{d^2},...,\bm{I}_{d^2}]$ $(d^2\times d^2p)$. Hence,
\begin{eqnarray*}
  \triangledown_{\bm{\phi}(\tau)}\bm{\omega}(\tau) &=& [\triangledown_{\bm{\alpha}(\tau)}\bm{\omega}(\tau), \triangledown_{\bm{\sigma}(\tau)}\bm{\omega}(\tau)] \\
   &=& \left[\left(\bm{N}_1^\top(\tau),\bm{N}_2^\top(\tau)\right)\left(\begin{matrix}
                                                                       \bm{N}_1(\tau) \\
                                                                       \bm{N}_2(\tau)
                                                                     \end{matrix} \right) \right]^{-1}\left[\bm{N}_1^\top(\tau),\bm{N}_2^\top(\tau)\right]\left[\begin{matrix}
                                                                                    \bm{0} & \bm{D}_1 \\
                                                                                    \bm{D}_2(\tau) & \bm{0}
                                                                             \end{matrix}\right]\\
   &=&\left(\bm{N}_1^\top(\tau)\bm{N}_1(\tau)+\bm{N}_2^\top(\tau)\bm{N}_2(\tau)\right)^{-1}\left[\bm{N}_2^\top(\tau)\bm{D}_2(\tau), \bm{N}_1^\top(\tau)\bm{D}_1 \right].
\end{eqnarray*}
For $\triangledown_{\bm{\alpha}(\tau)}\bm{B}_j(\tau)$,
\begin{eqnarray*}
\triangledown_{\bm{\alpha}(\tau)}\bm{B}_j(\tau) &=& \frac{\partial \mathrm{vec}\left[\bm{\Psi}_j(\tau)\bm{\omega}(\tau)\right]}{\partial \bm{\alpha}^\top(\tau)} \\
   &=& (\bm{\omega}^\top(\tau)\otimes \bm{I}_d) \frac{\partial \mathrm{vec}\left[\bm{\Psi}_j(\tau)\right]}{\partial \bm{\alpha}^\top(\tau)} + (\bm{I}_d \otimes \bm{\Psi}_j(\tau)) \frac{\partial \mathrm{vec}\left[\bm{\omega}(\tau)\right]}{\partial \bm{\alpha}^\top(\tau)}\\
   &=&(\bm{\omega}^\top(\tau)\otimes \bm{I}_d)\left(\sum_{m=0}^{j-1} \bm{J}(\bm{\Phi}^\top(\tau))^{j-1-m}\otimes \bm{\Psi}_m(\tau)\right)\\
   &&+(\bm{I}_d \otimes \bm{\Psi}_j(\tau))\left(\bm{N}_1^\top(\tau)\bm{N}_1(\tau)+\bm{N}_2^\top(\tau)\bm{N}_2(\tau)\right)^{-1}\bm{N}_2^\top(\tau)\bm{D}_2(\tau).
\end{eqnarray*}

For $\triangledown_{\bm{\sigma}(\tau)}\bm{B}_j(\tau)$,
\begin{eqnarray*}
\triangledown_{\bm{\sigma}(\tau)}\bm{B}_j(\tau) &=& \frac{\partial \mathrm{vec}\left[\bm{\Psi}_j(\tau)\bm{\omega}(\tau)\right]}{\partial \bm{\sigma}^\top(\tau)} =(\bm{I}_d\otimes \bm{\Psi}_j(\tau))\frac{\partial \mathrm{vec}\left[\bm{\omega}(\tau)\right]}{\partial \bm{\sigma}^\top(\tau)} \\
   &=&(\bm{I}_d\otimes \bm{\Psi}_j(\tau))\bm{L}_d^\top \left(\bm{N}_1^\top(\tau)\bm{N}_1(\tau)+\bm{N}_2^\top(\tau)\bm{N}_2(\tau)\right)^{-1}\bm{N}_1^\top(\tau)\bm{D}_1.
\end{eqnarray*}

The proof is now completed.
\end{proof}

\begin{proof}[Proof of Theorem 2.3]

\item
We need to prove that $\lim_{T\to \infty}\Pr\left(\text{IC}(\mathsf{p}) < \text{IC}(p)\right)=0$ for all $\mathsf{p}\neq p$ and $\mathsf{p}\leq \mathsf{P}$.

Note that
\be
\text{IC}(\mathsf{p})-\text{IC}(p)=\log[\text{RSS}(\mathsf{p})/\text{RSS}(p)]+(\mathsf{p}-p)\chi_T.
\nonumber
\ee

(i) For $\mathsf{p} < p$, Lemma \ref{LemmaB.5} implies that $\text{RSS}(\mathsf{p})/\text{RSS}(p) > 1 + \nu$ for some $\nu > 0 $ with large probability for all large $T$. Thus, $\log[\text{RSS}(\mathsf{p})/\text{RSS}(p)] \geq \nu/2$ for large $T$. Because $\chi_T\to 0$, we have $\text{IC}(\mathsf{p})-\text{IC}(p)\geq \nu/2-(p-\mathsf{p})\chi_T \geq \nu/3$ for large $T$ with large probability. Thus $\Pr\left(\text{IC}(\mathsf{p}) < \text{IC}(p)\right) \to 0$ for $\mathsf{p} < p$.

(ii) We then consider $\mathsf{p} > p$.  Lemma \ref{LemmaB.5} implies that $\log[\text{RSS}(\mathsf{p})/\text{RSS}(p)]=1+O_P(\rho_T^2)$ with $\rho_T= h^2 + \sqrt{\frac{\log(T)}{Th}}$. Hence, $\log[\text{RSS}(\mathsf{p})/\text{RSS}(p)]=O_P(\rho_T^2)$. Because $(\mathsf{p}-p) \chi_T \geq \chi_T $, which converges to zero at a slower rate than $\rho_T^2$, it follows that
\be
\Pr\left(\text{IC}(\mathsf{p}) < \text{IC}(p)\right)\leq \Pr\left(O_P(\rho_T^2)+ \chi_T < 0\right) \to 0.
\nonumber
\ee

The proof is now completed.
\end{proof}

\begin{proof}[Proof of Lemma \ref{Thm3.1}]
\item
By part(1) of Theorem \ref{Thm2.1}, we have $\sup_{\tau\in[0,1]}\|\widehat{\bm{\beta}}(\tau)-\bm{\beta}(\tau)\|=O_P(h^2+\sqrt{\log T/(Th)})$. Hence, by $Th^6\to 0$, we have
\be
\sqrt{T}\left(\widehat{\bm{c}} - \bm{c}\right)= \sqrt{T}\int_{h}^{1-h}\bm{C}\left(\widehat{\bm{\beta}}(\tau)-\bm{\beta}(\tau) \right)\mathrm{d}\tau + o_P(1).
\nonumber
\ee

In addition, by the proof of part(1) of Theorem \ref{Thm2.1}, the uniform convergence results of Lemmas \ref{LemmaB.3}-\ref{LemmaB.4} and the condition that $\bm{\beta}(\tau)$ has third-order derivative, we have
\begin{eqnarray*}
\widehat{\bm{\beta}}(\tau)-\bm{\beta}(\tau) &=& \frac{1}{2}h^2\widetilde{c}_2 \bm{\beta}^{(2)}(\tau) + \bm{\Sigma}_{\bm{Z}}^{-1}(\tau)\left(\frac{1}{Th}\sum_{t=1}^{T}\bm{Z}_{t-1} \bm{\eta}_t
K\left(\frac{\tau_t-\tau}{h}\right)\right)\\
&& + O_P(h^2\sqrt{\log T/(Th)}) + O_P(h^3)+O_P(\log T/(Th))
\end{eqnarray*}
uniformly over $\tau \in [h,1-h]$.

As $Th^6\to0$ and $Th^2/(\log T)^2 \to \infty$, we have
\begin{eqnarray*}
&&\sqrt{T}\left(\widehat{\bm{c}} - \bm{c} -\frac{1}{2}h^2\widetilde{c}_2\int_{0}^{1}\bm{C}\bm{\beta}^{(2)}(\tau)\mathrm{d}\tau\right) \\
&=& \frac{1}{\sqrt{T}}\sum_{t=1}^{T}\bm{C}\left(\int_{-1}^{1}\bm{\Sigma}_{\bm{Z}}^{-1}(\tau)K\left(\frac{\tau_t-\tau}{h}\right)\mathrm{d}\tau\right)\bm{Z}_{t-1} \bm{\eta_t} +o_P(1)\\
&=&\frac{1}{\sqrt{T}}\sum_{t=1}^{T}\bm{\Sigma}_{\bm{Z}}^{-1}(\tau_t)\bm{Z}_{t-1} \bm{\eta_t} +o_P(1)\\
&\to_D& N\left(\bm{0},\int_{0}^{1}\bm{C}\left( \bm{\Sigma}^{-1}(\tau)\otimes \bm{\Omega}(\tau)\right) \bm{C}^\top\mathrm{d}\tau\right)
\end{eqnarray*}
by Lemma B.1 under the conditions of this lemma.

The proof is now completed.
\end{proof}

\begin{proof}[Proof of Theorem \ref{Thm3.2}]
\item

First, we introduce some additional notation to facilitate the development. Let $A_Q=\widetilde{v}_0\cdot \mathrm{tr}\left\{\int_{0}^{1} \bm{\Sigma}_Q(\tau) \mathrm{d}\tau\right\}$ and $B_Q= 4 C_B\cdot \mathrm{tr}\left\{\int_{0}^{1}\bm{\Sigma}_Q(\tau)^2 \mathrm{d}\tau\right\}$, where $\bm{\Sigma}_Q(\tau)=\bm{H}(\tau)^{1/2}\bm{C}\bm{V}_{\bm{\beta}}(\tau)\bm{C}^\top \bm{H}(\tau)^{1/2} $. Recall $\rho_T=h^2+\sqrt{\frac{\log T}{Th}}$. By Theorem \ref{Thm2.1}.1, we have
\begin{eqnarray}\label{Rbeta}
	\sup_{\tau\in [0,1]} \left\|\widehat{\bm{\beta}}(\tau)-\bm{\beta}(\tau) \right\|=O_P\left(\rho_T\right).
\end{eqnarray}
Then we can conclude that as $Th^{11/2} = o(1)$,
\begin{eqnarray}
	\int_{\mathcal{B}_T}\left[\bm{C}\widehat{\bm{\beta}}(\tau)-\bm{c}\right]^\top\bm{H}(\tau)
	\left[\bm{C}\widehat{\bm{\beta}}(\tau)-\bm{c}\right]\mathrm{d}\tau = O_P\left(\log T / T +h^5\right) =o_P\left(T^{-1}h^{-1/2}\right),
\end{eqnarray}
where $\mathcal{B}_T = [0,h]\cup [1-h,1]$.

In addition, by  Lemma \ref{LemmaB.3} and Lemma \ref{LemmaB.4}.1, we have
\begin{eqnarray}\label{Rrate}
	&&\sup_{[h,1-h]}\left\|\bm{S}_T(\tau)-\bm{\Sigma}_{\bm{Z}}(\tau) \otimes \bm{\Lambda}_1\right\| =  O_P\left(h^2 + \sqrt{\frac{\log T}{Th}}\right),\nonumber \\
	&&\sup_{[0,1]}\left\| \bm{R}_T(\tau) \right\|= O_P\left(\sqrt{\frac{\log T}{Th}}\right),
\end{eqnarray}
where $\bm{\Sigma}_{\bm{Z}}(\tau) = \bm{\Sigma}(\tau)\otimes \bm{I}_d$, $\bm{\Lambda}_1 = \mathrm{diag}(\widetilde{c}_0,\widetilde{c}_2)$ and $\bm{R}_T(\tau)=\frac{1}{T}\sum_{t=1}^{T}\bm{Z}_{t-1}\bm{\eta}_tK_h(\tau_t-\tau)$. Hence, under the null hypothesis, we have

\begin{eqnarray}\label{Rrate2}
	\sup_{[h,1-h]}\left\|\bm{C}\widehat{\bm{\beta}}(\tau)- \bm{c}- \bm{C}\bm{\Sigma}_{\bm{Z}}^{-1}(\tau) \bm{R}_T(\tau)\right\|=O_P(\rho_T^2).
\end{eqnarray}

By \eqref{Rrate} and $Th^{11/2} \to 0$, we have
\begin{eqnarray*}
	&&\int_0^1\left[\bm{C}\widehat{\bm{\beta}}(\tau)-\bm{c}\right]^\top\bm{H}(\tau)
	\left[\bm{C}\widehat{\bm{\beta}}(\tau)-\bm{c}\right]\mathrm{d}\tau \\
	&=& \int_{h}^{1-h}\bm{R}_T^\top(\tau) \bm{H}_0 (\tau)\bm{R}_T(\tau)\mathrm{d}\tau + O_P\left(\log T / T +h^5\right) + O_P\left(\rho_T^2 \sqrt{\frac{\log T}{Th}}\right) \\
	&=& \int_{0}^{1}\bm{R}_T^\top(\tau) \bm{H}_{0}(\tau)\bm{R}_T(\tau)\mathrm{d}\tau +o_P(T^{-1}h^{-1/2}),
\end{eqnarray*}
where $\bm{H}_{0}(\tau)=\bm{\Sigma}_{\bm{Z}}^{-1}(\tau)\bm{C}^\top\bm{H}(\tau)\bm{C}\bm{\Sigma}_{\bm{Z}}^{-1}(\tau)$.

Consider $\int_{0}^{1}\bm{R}_T^\top(\tau)\bm{H}_{0}(\tau)\bm{R}_T(\tau) \mathrm{d}\tau$, and write
\begin{eqnarray*}
	&&\int_{0}^{1}\bm{R}_T^\top(\tau)\bm{H}_{0}(\tau)\bm{R}_T(\tau) \mathrm{d}\tau =  \frac{1}{T^2h^2}\sum_{t=1}^{T}\bm{\eta}_t^\top\bm{Z}_{t-1}^\top \left\{\int_{0}^{1}\bm{H}_0(\tau)K^2\left(\frac{\tau-\tau_t}{h}\right) \mathrm{d}\tau\right\} \bm{Z}_{t-1}\bm{\eta}_t\\
	&& + \frac{1}{T^2h^2}\sum_{t=1}^{T}\sum_{s=1,\neq t}^{T}\bm{\eta}_t^\top\bm{Z}_{t-1}^\top \left\{\int_{0}^{1}\bm{H}_0(\tau)K\left(\frac{\tau-\tau_t}{h}\right)K\left(\frac{\tau-\tau_s}{h}\right) \mathrm{d}\tau \right\}\bm{Z}_{s-1} \bm{\eta}_s\\
	&:=&I_{T,1}+I_{T,2},
\end{eqnarray*}
where the definitions of $I_{T,1}$ and $I_{T,2}$ should be obvious.

For $I_{T,1}$, simple algebra shows that
\begin{eqnarray*}
	&&\bm{\eta}_t^\top\bm{Z}_{t-1}^\top\bm{\Sigma}_{\bm{Z}}^{-1}(\tau)\bm{C}^\top\bm{H}(\tau)\bm{C}\bm{\Sigma}_{\bm{Z}}^{-1}(\tau)\bm{Z}_{t-1}\bm{\eta}_t \\
	&=& \mathrm{tr}\left\{\left[\left(\bm{\Sigma}^{-1}(\tau)\bm{z}_{t-1}\bm{z}_{t-1}^\top\bm{\Sigma}^{-1}(\tau)\right)\otimes  \bm{\eta}_t\bm{\eta}_t^\top \right]\cdot\bm{C}^\top\bm{H}(\tau)\bm{C} \right\}.
\end{eqnarray*}

Then we have
\begin{eqnarray*}
	I_{T,1} &=& \widetilde{v}_0\frac{1}{T^2h}\sum_{t=1}^{T}\bm{\eta}_t^\top\bm{Z}_{t-1}^\top\left[\bm{H}_0(\tau_t)+O(h)\right]\bm{Z}_{t-1}\bm{\eta}_t\\
	&=& \widetilde{v}_0\frac{1}{T^2h}\sum_{t=1}^{T}\mathrm{tr}\left\{\left[\left(\bm{\Sigma}^{-1}(\tau_t)\bm{z}_{t-1}\bm{z}_{t-1}^\top\bm{\Sigma}^{-1}(\tau_t)\right)\otimes  \bm{\eta}_t\bm{\eta}_t^\top \right]\cdot \bm{C}^\top\bm{H}(\tau_t)\bm{C}\right\}+O_P(T^{-1})\\
	&=& \widetilde{v}_0\frac{1}{T^2h}\sum_{t=1}^{T}\mathrm{tr}\left\{\left[\bm{\Sigma}^{-1}(\tau_t)\otimes  \bm{\Omega}(\tau_t)\right]\cdot\bm{C}^\top\bm{H}(\tau_t)\bm{C}\right\}+O_P(T^{-1}+T^{-3/2}h^{-1})\\
	&\to_P&(Th)^{-1}A_Q.
\end{eqnarray*}

Consider $I_{T,2}$. Let $w_{s,t}=\frac{1}{T\sqrt{h}}\int_{-1}^{1}K\left(u\right)K\left(u+\frac{t-s}{Th}\right)\mathrm{d}u$. Since
$$
\int_{0}^{1}\bm{H}_0(\tau)K\left(\frac{\tau-\tau_t}{h}\right)K\left(\frac{\tau-\tau_s}{h}\right)\mathrm{d}\tau= h\int_{-1}^{1}\bm{H}_0(\tau_t+uh)K\left(u\right)K\left(u+\frac{t-s}{Th}\right)du,
$$
we have
$$
T\sqrt{h}I_{T,2}= 2\sum_{t=2}^{T}\sum_{s=1}^{t-1}\bm{\eta}_t^\top\bm{Z}_{t-1}^\top\bm{H}_0(\tau_t)\bm{Z}_{s-1}\bm{\eta}_s w_{s,t}(1+o(1))= 2 \widetilde{U}+o_P(1),
$$
where the definition of $\widetilde{U}$ is obvious. By Lemma \ref{L5}, we have $$\widetilde{U} \to_D N\left(0,\sigma_{\widetilde{U}}^2\right),$$ where $ \sigma_{\widetilde{U}}^2 = \lim_{T \to \infty}\sum_{t=2}^{T}\mathrm{tr}\left\{E\left(\bm{H}_0^\top(\tau_t)\bm{Z}_{t-1}\bm{\eta}_t\bm{\eta}_t^\top\bm{Z}_{t-1}^\top\bm{H}_0(\tau_t)\right)E\left( \sum_{s=1}^{t-1}\bm{Z}_{s-1}\bm{\eta}_s\bm{\eta}_s^\top\bm{Z}_{s-1}^\top\right)w_{s,t}^2\right\}$.

We then show that $\sigma_{\widetilde{U}}^2 = C_B \mathrm{tr}\left\{\int_{0}^{1}\bm{\Sigma}_Q(\tau)^2\mathrm{d}\tau\right\}$. Let $\bm{V}_1 (\tau)=\bm{H}_0(\tau)\bm{V}_2(\tau)\bm{H}_0(\tau)$ and $\bm{V}_2(\tau) = \bm{\Sigma}(\tau) \otimes \bm{\Omega}(\tau)$. Write

\begin{eqnarray*}
	&&\sum_{t=2}^{T}E\left(\bm{H}_0^\top(\tau_t)\bm{Z}_{t-1}\bm{\eta}_t\bm{\eta}_t^\top\bm{Z}_{t-1}^\top\bm{H}_0(\tau_t)\right)E\left( \sum_{s=1}^{t-1}\bm{Z}_{s-1}\bm{\eta}_s\bm{\eta}_s^\top\bm{Z}_{s-1}^\top\right)w_{s,t}^2\\
	&=&\sum_{t=2}^{T}\sum_{s=1}^{t-1}\bm{V}_1(\tau_t)\bm{V}_2(\tau_s)w_{s,t}^2 =  \frac{1}{T^2h}\sum_{t=2}^{T}\sum_{s=1}^{t-1}\bm{V}_1(\tau_t)\bm{V}_2(\tau_s) \left[\int_{-1}^{1}K\left(u\right)K\left(u+\frac{t-s}{Th}\right)\mathrm{d}u\right]^2\\
	&=&\frac{1}{T^2h}\sum_{s=1}^{T-1}\sum_{j=1}^{T-s}\bm{V}_1(\tau_s+j/T)\bm{V}_2(\tau_s)\left[\int_{-1}^{1}K\left(u\right)K\left(u+\frac{t-s}{Th}\right)\mathrm{d}u\right]^2\\
	&=&\frac{1}{T^2h}\sum_{s=1}^{T-1}\sum_{j=1}^{T-s}\bm{V}_1(\tau_s+j/T)\bm{V}_2(\tau_s)\left[\int_{-1}^{1}K\left(u\right)K\left(u+\frac{j}{Th}\right)\mathrm{d}u\right]^2\\
	&=&\frac{1}{T^2h}\sum_{s=1}^{T-1}\sum_{j=1}^{T-s}\bm{V}_1(\tau_s)\bm{V}_2(\tau_s)\left[\int_{-1}^{1}K\left(u\right)K\left(u+\frac{j}{Th}\right)\mathrm{d}u\right]^2\\
	&&+\frac{1}{T^2h}\sum_{s=1}^{T-1}\sum_{j=1}^{T-s}O(j/T)\bm{V}_2(\tau_s)\left[\int_{-1}^{1}K\left(u\right)K\left(u+\frac{j}{Th}\right)\mathrm{d}u\right]^2 := I_{T,3} + I_{T,4},
\end{eqnarray*}
where the definitions of $I_{T,3} $ and $I_{T,4}$ are obvious.

It is easy to verify $\mathrm{tr}\left\{I_{T,3}\right\} \to C_B \mathrm{tr}\left\{\int_{0}^{1}\bm{\Sigma}_Q(\tau)^2\mathrm{d}\tau\right\}$. For $I_{T,4}$, we have
\begin{eqnarray*}
	\left\|I_{T,4}\right\|&\leq&M \frac{1}{Th}\sum_{j=1}^{T}j/T\left[\int_{-1}^{1}K\left(u\right)K\left(u+\frac{j}{Th}\right)\mathrm{d}u\right]^2\\
	&\simeq&M \int_{0}^{2}vh\left[\int_{-1}^{1}K\left(u\right)K\left(u+v\right)\mathrm{d}u\right]^2\mathrm{d}v =o(1).
\end{eqnarray*}

Combining the above results, we have proved
\begin{eqnarray*}
	T\sqrt{h}\left[\int_{0}^{1}\bm{R}_T^\top(\tau) \bm{H}_{0}(\tau)\bm{R}_T(\tau)\mathrm{d}\tau-(Th)^{-1}A_Q\right]\to_D N\left(0,B_Q\right).
\end{eqnarray*}

\medskip

Note that
\begin{eqnarray*}
	&&\int_0^1\left[\bm{C}\widehat{\bm{\beta}}(\tau)-\widehat{\bm{c}}\right]^\top\bm{H}(\tau)
	\left[\bm{C}\widehat{\bm{\beta}}(\tau)-\widehat{\bm{c}}\right]\mathrm{d}\tau-\int_0^1\left[\bm{C}\widehat{\bm{\beta}}(\tau)-\bm{c}\right]^\top\bm{H}(\tau)
	\left[\bm{C}\widehat{\bm{\beta}}(\tau)-\bm{c}\right]\mathrm{d}\tau\\
	&=&\int_{0}^{1}\left( \widehat{\bm{c}}-\bm{c}\right)^\top \bm{H}(\tau)\left(\widehat{\bm{c}}-\bm{c}\right) \mathrm{d}\tau - 2\int_{0}^{1}\left( \widehat{\bm{c}}-\bm{c}\right)^\top \bm{H}(\tau)\left(\bm{C}\widehat{\bm{\beta}}(\tau)-\bm{c}\right)\mathrm{d}\tau  \\
	&:=&I_{T,5} - 2 I_{T,6},
\end{eqnarray*}
where the definitions of $I_{T,5}$ and $I_{T,6}$ are obvious.

Since $\widehat{\bm{c}} = \bm{c} + O_P\left(T^{-1/2}\right)$, we have $I_{T,5} = O_P(T^{-1})$. For $I_{T,6}$, by \eqref{Rbeta} and \eqref{Rrate2}, we have

\begin{eqnarray*}
	I_{T,6} &=& \left(\widehat{\bm{c}}-\bm{c}\right)^\top\int_{h}^{1-h}\bm{H}(\tau)\bm{C}\bm{\Sigma}_{\bm{Z}}^{-1}(\tau)\bm{R}_T(\tau)\mathrm{d}\tau +o_P\left(T^{-1}h^{-1/2}\right)\\
	&=&  O_P(T^{-1})+o_P\left(T^{-1}h^{-1/2}\right)=o_P\left(T^{-1}h^{-1/2}\right)
\end{eqnarray*}
provided that

\begin{eqnarray*}
	&& \int_{h}^{1-h}\bm{H}(\tau)\bm{C}\bm{\Sigma}_{\bm{Z}}^{-1}(\tau)\bm{R}_T(\tau)\mathrm{d}\tau \\
	&=& \frac{1}{T}\sum_{t=1}^{T}\int_{-1}^{1}\bm{H}(\tau_t+uh)\bm{C}\bm{\Sigma}_{\bm{Z}}^{-1}(\tau_t+uh) K(u)\mathrm{d}u \bm{Z}_{t-1}\bm{\eta}_t = O_P(T^{-1/2}).
\end{eqnarray*}

We thenconclude that $T\sqrt{h}\left[\widehat{Q}_{\bm{C},\bm{H}}-(Th)^{-1}A_Q\right]\to_D N\left(0,B_Q\right)$.

Observe that
\begin{eqnarray*}
	&&\int_0^1\left[\bm{C}\widehat{\bm{\beta}}(\tau)-\widehat{\bm{c}}\right]^\top\widehat{\bm{H}}(\tau)
	\left[\bm{C}\widehat{\bm{\beta}}(\tau)-\widehat{\bm{c}}\right]\mathrm{d}\tau -\int_0^1\left[\bm{C}\widehat{\bm{\beta}}(\tau)-\bm{c}\right]^\top\bm{H}(\tau)
	\left[\bm{C}\widehat{\bm{\beta}}(\tau)-\bm{c}\right]\mathrm{d}\tau\\
	&=& \int_0^1\left[\bm{C}\widehat{\bm{\beta}}(\tau)-\widehat{\bm{c}}\right]^\top\widehat{\bm{H}}(\tau)  \left[\bm{C}\widehat{\bm{\beta}}(\tau)-\widehat{\bm{c}}\right]\mathrm{d}\tau -\int_0^1\left[\bm{C}\widehat{\bm{\beta}}(\tau)-\widehat{\bm{c}}\right]^\top\bm{H}(\tau) \left[\bm{C}\widehat{\bm{\beta}}(\tau)-\widehat{\bm{c}}\right]\mathrm{d}\tau \\
	&&+\int_0^1\left[\bm{C}\widehat{\bm{\beta}}(\tau)-\widehat{\bm{c}}\right]^\top\bm{H}(\tau) \left[\bm{C}\widehat{\bm{\beta}}(\tau)-\widehat{\bm{c}}\right]\mathrm{d}\tau -\int_0^1\left[\bm{C}\widehat{\bm{\beta}}(\tau)-\bm{c}\right]^\top\bm{H}(\tau)
	\left[\bm{C}\widehat{\bm{\beta}}(\tau)-\bm{c}\right]\mathrm{d}\tau.
\end{eqnarray*}

Then we just need to focus on
$$ \int_0^1\left[\bm{C}\widehat{\bm{\beta}}(\tau)-\widehat{\bm{c}}\right]^\top\widehat{\bm{H}}(\tau)  \left[\bm{C}\widehat{\bm{\beta}}(\tau)-\widehat{\bm{c}}\right]\mathrm{d}\tau -\int_0^1\left[\bm{C}\widehat{\bm{\beta}}(\tau)-\widehat{\bm{c}}\right]^\top\bm{H}(\tau) \left[\bm{C}\widehat{\bm{\beta}}(\tau)-\widehat{\bm{c}}\right]\mathrm{d}\tau := I_{T,7}
$$

Hence, it suffices to show $T\sqrt{h}I_{T,7} = o_P(1)$. Using Lemma \ref{LemmaB.4}.3, it is easy to know that

\begin{eqnarray*}
	|I_{T,7}|&\leq &\sup_{\tau \in [0,1]}\left\|\widehat{\bm{H}}(\tau)- \bm{H}(\tau)\right\|\times \widehat{Q}_{\bm{C},\bm{I}_s}\\
	& = &O_P\left(h + \sqrt{\frac{\log T}{Th}}\right)O_P\left((Th)^{-1} + 1/(T\sqrt{h})\right) = o_P(1/(T\sqrt{h})).
\end{eqnarray*}

The proof is now completed.
\end{proof}

\begin{proof}[Proof of Corollary \ref{Thm3.3}]
\item
	
Under the local alternative \eqref{Eq3.6}, we have $\bm{C}\bm{\beta}(\tau) = \bm{c} + d_T \bm{f}(\tau)$ and thus
\begin{eqnarray*}
	&&\widehat{Q}_{\bm{C},\bm{H}} - \int_{0}^{1}\bm{R}_T^\top(\tau) \bm{H}_{0}(\tau)\bm{R}_T(\tau) \mathrm{d}\tau \\
	&=& d_T^2 \int_{0}^{1}\bm{f}(\tau)^\top \bm{H}(\tau)\bm{f}(\tau)\mathrm{d}\tau + 2d_T \int_{0}^{1}\bm{f}(\tau)^\top \bm{H}(\tau)\left(\bm{C}\widehat{\bm{\beta}}(\tau)-\bm{C}\bm{\beta}(\tau)\right)\mathrm{d}\tau\\
	&& + \left[\int_{0}^{1}\left(\bm{C}\widehat{\bm{\beta}}(\tau)-\bm{C}\bm{\beta}(\tau)\right)^\top \bm{H}(\tau)\left(\bm{C}\widehat{\bm{\beta}}(\tau)-\bm{C}\bm{\beta}(\tau)\right)\mathrm{d}\tau-\int_{0}^{1}\bm{R}_T^\top(\tau) \bm{H}_{0}(\tau)\bm{R}_T(\tau) \mathrm{d}\tau\right]\\
	&=& d_T^2 \int_{0}^{1}\bm{f}(\tau)^\top \bm{H}(\tau)\bm{f}(\tau)\mathrm{d}\tau + I_{T,1} + I_{T,2}.
\end{eqnarray*}

Since $\bm{C}\widehat{\bm{\beta}}(\tau)-\bm{C}\bm{\beta}(\tau) = O_P\left(d_T \rho_T+\sqrt{\frac{\log T}{Th}}\rho_T\right) + \bm{C}\bm{\Sigma}_{\bm{Z}}^{-1}(\tau)\bm{R}_T(\tau)$ uniformly over $\tau\in[h,1-h]$ and
$$
\int_{0}^{1} \bm{f}(\tau)^\top \bm{H}(\tau) \bm{C}\bm{\Sigma}_{\bm{Z}}^{-1}(\tau)\bm{R}_T(\tau) \mathrm{d}\tau =O_P(T^{-1/2}),
$$
we have $I_{T,1} = O_P\left(d_T (d_T \rho_T + \sqrt{\frac{\log T}{Th}}\rho_T + T^{-1/2} )\right)=o_P(T^{-1}h^{-1/2}) $.

For $I_{T,2}$, since $\sup_{\tau \in [0,1]}\left\|\bm{R}_T(\tau)\right\| = O_P\left(\sqrt{\frac{\log T}{Th}} \right)$, we have
$$
I_{T,2} = O_P\left(d_T^2\rho_T^2 + d_T\rho_T\sqrt{\frac{\log T}{Th}} \right) = o_P(T^{-1}h^{-1/2}).
$$

As $T\sqrt{h}\left(\int_{0}^{1}\bm{R}_T^\top(\tau) \bm{H}_{0}(\tau)\bm{R}_T(\tau)- (Th)^{-1}A_Q\right) \to N(0,B_Q)$,
we have
$$
T\sqrt{h}\left(\widehat{Q}_{\bm{C},\bm{H}}- (Th)^{-1}A_Q\right) \to N(\delta_1,B_Q).
$$
In addition, similar to the proofs of Theorem \ref{Thm3.2}, we have $T\sqrt{h}\left(\widehat{Q}_{\bm{C},\bm{H}}-\widehat{Q}_{\bm{C},\widehat{\bm{H}}}\right)=o_P(1)$. The proof is now completed.
\end{proof}

}
\newpage

\section*{Appendix B}

\renewcommand{\theequation}{B.\arabic{equation}}
\renewcommand{\thesection}{B.\arabic{section}}
\renewcommand{\thefigure}{B.\arabic{figure}}
\renewcommand{\thetable}{B.\arabic{table}}
\renewcommand{\thelemma}{B.\arabic{lemma}}

\setcounter{equation}{0}
\setcounter{lemma}{0}
\setcounter{section}{0}
\setcounter{table}{0}
\setcounter{figure}{0}

{\small

In this appendix, we present the results omitted from the main text of this paper. Specifically, Appendix \ref{AppB.1} includes the preliminary lemmas. Appendix \ref{AppB.2} provides the proofs of the preliminary lemmas.

\section{Preliminary Lemmas}\label{AppB.1}

\begin{lemma}\label{LemmaB.1}
Suppose $\{Z_t,\mathcal{F}_{t}\}$ is a martingale difference sequence, $S_T=\sum_{t=1}^{T}Z_t$, $U_T=\sum_{t=1}^{T}Z_t^2$ and $s_T^2=E(U_T^2)=E(S_T^2)$. If $s_T^{-2}U_T^2 \to_P 1$ and $\sum_{t=1}^{T}E[Z_{T,t}^2I(|Z_{T,t}|>\nu)] \to 0$ for any $\nu > 0$ with $Z_{T,t}=s_T^{-1}Z_t$, then as $T \to \infty$, $ s_T^{-1}S_T \to_D N(0,1).$
\end{lemma}

Lemma \ref{LemmaB.1} can be found in \cite{hall2014martingale}.

\medskip

\begin{lemma}\label{LemmaB.2}
Let $\{Z_t,\mathcal{F}_{t}\}$ be a martingale difference sequence. Suppose that $\left| Z_t\right|\leq M$ for a constant $M$, $t=1,\ldots,T$. Let $V_T=\sum_{t=1}^{T}\text{\normalfont Var}\left(Z_t|\mathcal{F}_{t-1}\right)\leq V$ for some $V>0$. Then for any given $\nu> 0$,
\begin{equation*}
  \Pr\left(\left|\sum_{t=1}^{T}Z_t\right| > \nu\right)\leq \exp\left\{-\frac{\nu^2}{2(V+M\nu)}\right\}.
\end{equation*}
\end{lemma}

Lemma \ref{LemmaB.2} is the Proposition 2.1 of \cite{freedman1975tail}.

\begin{assumption}\label{AssB.1}
$\max_t \|\bm{\mu}_t\| < \infty$, $\max_t \sum_{j=1}^{\infty}j\|\bm{B}_{j,t}\| < \infty$, $\limsup_{T\to\infty} \sum_{t=1}^{T-1}\|\bm{\mu}_{t+1}-\bm{\mu}_{t}\| < \infty$ and $\limsup_{T\to\infty} \sum_{t=1}^{T-1}\sum_{j=1}^{\infty}j\|\bm{B}_{j,t+1}-\bm{B}_{j,t}\| < \infty$. Define the stochastic process of the form $\bm{h}_t = \bm{\mu}_t + \sum_{j=0}^{\infty}\bm{B}_{j,t}\bm{\epsilon}_t$ for $t=1,...,T$.
\end{assumption}

\begin{lemma}\label{LemmaB.6}
	Let Assumptions \ref{Ass2} and \ref{AssB.1} hold and $\max_{t \geq 1} E\left(\left\|\bm{\epsilon}_t\right\|^\delta | \mathcal{F}_{t-1}\right)<\infty$ a.s.. In addition, let $\left\{\bm{W}_{T,t}(\cdot)\right\}_{t=1}^T$ be a sequence of $q \times d$ matrices of deterministic functions, in which $q$ is fixed, each functional component is Lipschitz continuous and defined on a compact set $[a,b]$. Moreover, suppose that
	\begin{enumerate}
		\item $\sup_{\tau\in[a,b]}\sum_{t=1}^{T} \|\bm{W}_{T,t}(\tau) \|=O(1)$;
		\item $\sup_{\tau\in[a,b]}\sum_{t=1}^{T-1} \|\bm{W}_{T,t+1}(\tau)-\bm{W}_{T,t}(\tau) \| = O(d_T)$, where $d_T=\sup_{\tau\in[a,b],t\ge 1} \|\bm{W}_{T,t}(\tau) \|$.
	\end{enumerate}
	Then as $T\to \infty$,
	\begin{enumerate}
		\item $ \sup_{\tau\in[a,b]} \|\sum_{t=1}^{T}\bm{W}_{T,t}(\tau)\left(\bm{h}_t-E(\bm{h}_t)\right) \|=O_P(\sqrt{d_T\log T})$ provided $T^{\frac{2}{\delta}} d_T \log T \rightarrow0$;
		
		\item $\sup_{\tau\in[a,b]} \|\sum_{t=1}^{T}\bm{W}_{T,t}(\tau)\left(\bm{h}_t\bm{h}_{t+p}^\top-E(\bm{h}_t \bm{h}_{t+p}^\top)\right)\|=O_P (\sqrt{d_T\log T} )$ for any fixed integer $p\geq0$ provided $T^{\frac{4}{\delta}} d_T \log T \rightarrow0$, where $\delta$ is the same as in Assumption \ref{Ass2}.
	\end{enumerate}
\end{lemma}

Lemma \ref{LemmaB.6} is the Theorem 2.1 in \cite{yan2020class}.

\begin{lemma}\label{LemmaB.7}
	Let Assumptions \ref{Ass2} and \ref{AssB.1} hold. In addition, let $\left\{\bm{W}_{T,t}(\cdot)\right\}_{t=1}^T$ be a sequence of $q \times d$ matrices of deterministic functions, in which $q$ is fixed, each functional component is Lipschitz continuous and defined on a compact set $[a,b]$. Moreover, suppose that
	\begin{enumerate}
		\item $\sup_{\tau\in[a,b]}\sum_{t=1}^{T} \|\bm{W}_{T,t}(\tau) \|=O(1)$;
		\item $\sup_{\tau\in[a,b]}\sum_{t=1}^{T-1} \|\bm{W}_{T,t+1}(\tau)-\bm{W}_{T,t}(\tau) \| = O(d_T)$, where $d_T=\sup_{\tau\in[a,b],t\ge 1} \|\bm{W}_{T,t}(\tau) \|$.
	\end{enumerate}
	Then as $T\to \infty$, for any $\tau \in [a,b]$
	\begin{enumerate}
		\item $ \|\sum_{t=1}^{T}\bm{W}_{T,t}(\tau)\left(\bm{h}_t-E(\bm{h}_t)\right) \|=O_P(\sqrt{d_T})$;
		
		\item $ \|\sum_{t=1}^{T}\bm{W}_{T,t}(\tau)\left(\bm{h}_t\bm{h}_{t+p}^\top-E(\bm{h}_t \bm{h}_{t+p}^\top)\right)\|=O_P (\sqrt{d_T} )$ for any fixed integer $p\geq0$.
	\end{enumerate}
\end{lemma}

The proof of Lemma \ref{LemmaB.7} is similar with that of Theorem 2.1 in \cite{yan2020class}.

\begin{lemma}\label{LemmaB.3}
Suppose Assumptions 1-3 hold. Let $\bm{W}(\cdot)$ be a twice-differentiable functional matrix in $R^{m\times d}$. As $T\to \infty$,
\begin{enumerate}
    \item for $\tau\in(0,1)$,
    \begin{eqnarray*}
    &&  \frac{1}{T}\sum_{t=1}^{T}\bm{W}(\tau_t)\bm{x}_t\left(\frac{\tau_t-\tau}{h}\right)^kK_h(\tau_t-\tau)-\widetilde{c}_k\bm{W}(\tau)\bm{\mu}(\tau) = O_P(h^2+1/(\sqrt{Th})),\\
    &&    \frac{1}{T}\sum_{t=1}^{T}\bm{W}(\tau_t)\bm{x}_t\bm{x}_{t+p}^\top\left(\frac{\tau_t-\tau}{h}\right)^k K_h(\tau_t-\tau)-\widetilde{c}_k\bm{W}(\tau)\bm{\Sigma}_{p}(\tau)=O_P(h^2+1/(\sqrt{Th})),
    \end{eqnarray*}
        where $\bm{\Sigma}_{p}(\tau)=\bm{\mu}(\tau)\bm{\mu}^\top(\tau)+\sum_{j=0}^{\infty}\bm{B}_j(\tau) \bm{B}_{j+p}^\top(\tau)$ for fixed integers $k$ and $p\geq0$;
    \item given $\frac{T^{1-\frac{2}{\delta}}h}{\log T} \rightarrow \infty$ and $\rho_T = h^2 + \sqrt{\frac{\log(T)}{Th}}$,
\begin{equation*}
 \sup_{\tau\in[h, 1-h]} \left\|\frac{1}{T}\sum_{t=1}^{T}\bm{W}(\tau_t)\bm{x}_t \left(\frac{\tau_t-\tau}{h}\right)^k K_h(\tau_t-\tau)-\widetilde{c}_k\bm{W}(\tau)\bm{\mu}(\tau)\right\|=O_{P}\left(\rho_T^2\right);
\end{equation*}

\item given $\frac{T^{1-\frac{4}{\delta}}h}{\log T} \rightarrow \infty$ and $\max_{t\ge 1} E[\|\bm{\epsilon}_t\|^4 |\mathcal{F}_{t-1} ] < \infty $ a.s.,
\begin{equation*}
 \sup_{\tau\in [h,1-h]} \left\|\frac{1}{T}\sum_{t=1}^{T}\bm{W}(\tau_t)\bm{x}_t \bm{x}_{t+p}^\top \left(\frac{\tau_t-\tau}{h}\right)^k K_h(\tau_t-\tau)-\widetilde{c}_k\bm{W}(\tau)\bm{\Sigma}_{p}(\tau) \right\|=O_{P}\left(\rho_T^2\right).
\end{equation*}

\item $\frac{1}{T}\sum_{t=1}^{T}\bm{W}(\tau_t)(\bm{x}_t-E(\bm{x}_t)) = O_P(1/\sqrt{T})$ and $\frac{1}{T}\sum_{t=1}^{T}\bm{W}(\tau_t)(\bm{x}_t\bm{x}_{t+p}^\top-E(\bm{x}_t\bm{x}_{t+p}^\top))=O_P(1/\sqrt{T})$.
\end{enumerate}
\end{lemma}

\begin{lemma}\label{LemmaB.4}
Let Assumptions 1-3 hold. Suppose $\frac{T^{1-\frac{4}{\delta}}h}{\log T} \rightarrow \infty$ and $\max_{t\ge 1} E[\left\|\bm{\epsilon}_t \right\|^4 |\mathcal{F}_{t-1}]< \infty $ a.s. As $T\to \infty$,
  \begin{enumerate}
    \item $\sup_{\tau \in [0,1]}\left\|\frac{1}{Th}\sum_{t=1}^{T} \bm{Z}_{t-1}\bm{\eta}_t K\left(\frac{\tau_t-\tau}{h} \right)\right\|=O_P\left( \left(\frac{\log T}{Th} \right)^{\frac{1}{2}} \right)$;
    \item $\frac{1}{\sqrt{Th}}\sum_{t=1}^{T}\bm{\eta}_t \left(\bm{\eta}_t-\bm{\widehat{\eta}}_t\right)^\top K\left(\frac{\tau_t-\tau}{h} \right)=o_P(1)$ for $\forall \tau \in [0,1]$;
    \item $\sup_{\tau \in[h,1-h]}\left\|\widehat{\bm{V}}_{\bm{\beta}}(\tau)-\bm{V}_{\bm{\beta}}(\tau) \right\| =O_P(h^2+\left(\frac{\log T}{Th}\right)^{\frac{1}{2}})$.
  \end{enumerate}
\end{lemma}

\begin{lemma}\label{LemmaB.5}
Let Assumptions 1-3 hold. Suppose $\frac{T^{1-\frac{4}{\delta}}h}{\log T} \rightarrow \infty$ and $\max_{t\ge 1} E[\left\|\bm{\epsilon}_t \right\|^4 |\mathcal{F}_{t-1}]< \infty $ a.s. As $T\to \infty$,
  \begin{enumerate}
    \item if $\mathsf{p}\geq p$, then $\text{RSS}(\mathsf{p})=\frac{1}{T}\sum_{t=1}^{T}E\left(\bm{\eta}_t^\top\bm{\eta}_t\right) + O_P\left(\rho_T^2\right)$ with $\rho_T= h^2 + \sqrt{\frac{\log(T)}{Th}}$;
    \item if $\mathsf{p} < p$, then $\text{RSS}(\mathsf{p})=\frac{1}{T}\sum_{t=1}^{T}E\left(\bm{\eta}_t^\top\bm{\eta}_t\right) + c + o_P\left(1\right)$ with some constant $c >0$.
  \end{enumerate}
\end{lemma}

Define $w_{s,t} = \frac{1}{T\sqrt{h}}\int_{-1}^{1}K\left(u\right)K\left(u+\frac{t-s}{Th}\right)\mathrm{d}u$. Let $a_t = \sum_{s=1}^{t-1} w_{s,t}^2$, $b_s=\sum_{t=s+1}^{T}w_{s,t}^2$ and $\sigma^2 = \sum_{t=2}^{T}a_t$.

\begin{lemma}\label{L1}
	Suppose Assumption 3 hold. Then, we have
	\begin{enumerate}
		\item $\sigma^2 \to \int_{0}^{2}\left[\int_{-1}^{1-v}K\left(u\right)K\left(u+v\right)\mathrm{d}u\right]^2\mathrm{d}v$;
		\item $\max_{2\leq t\leq T} a_t = O(1/T)$;
		\item for any fixed $J \in \mathbb{N}$, $\sum_{s=1}^{T-J}w_{s,s+J}^2 = O(1/(Th))$;
		\item $T\sum_{s=1}^{T-1}b_s^2 = O(1)$;
		\item $\sum_{k=1}^{T-1}\sum_{t=1}^{k-1}\left[\sum_{j=k+1}^{T}w_{k,j}w_{t,j}\right]^2 = O(1/T)$;
	\end{enumerate}
\end{lemma}

For a random vector $\bm{z}$, we write $ \bm{z} \in \mathcal{L}^q$, $q > 0$, if $\normmm{\bm{z}}_q = \left[ E\left(\left\|\bm{z}\right\|^q\right)\right]^{1/q} < \infty$, and we denote $\normmm{\cdot} = \normmm{\cdot}_2$.

\begin{lemma}\label{L2}
	Let $\bm{\xi}_1$,...,$\bm{\xi}_T$ be a $d$-dimensional martingale difference for which $\bm{\xi}_t \in \mathcal{L}^p$, $p >1$. Let $p^* = \min(2,p)$. Then
	$$
	\normmm{\sum_{t=1}^{T}\bm{\xi}_t}_p^{p^*} \leq M \sum_{t=1}^{T}\normmm{\bm{\xi}_t}_p^{p^*}.
	$$
\end{lemma}

Let $\bm{h}_{t-1}^* := \sum_{s=1}^{t-1}w_{s,t}\bm{y}_{s}$, where $\bm{y}_{t} \in \mathcal{L}^{\delta}$, $t=1,2,...,T$ are martingale differences subject to the filtration $\mathcal{F}_t$ and $\delta > 4$.

\begin{lemma}\label{L3}
	Suppose Assumption \ref{Ass3} hold. Assume $\bm{w}_t^* \in \mathcal{L}^{\delta/2}$ for some $\delta > 4$ is $\mathcal{F}_t$-measurable. Then, as $T \to \infty$,
	\begin{equation*}
		E\left|\sum_{t=2}^{T}\mathrm{tr}\left[\left(\bm{w}_t^* -E(\bm{w}_t^*)\right)\bm{h}_{t-1}^*\bm{h}_{t-1}^{*,\top}\right]\right| \to 0.
	\end{equation*}
\end{lemma}

\begin{lemma}\label{L4}
	Suppose Assumptions \ref{Ass1}-\ref{Ass3} hold and $\bm{y}_t = \bm{Z}_{t-1} \bm{\eta}_t$. Then, as $T \to \infty$,
	\begin{equation*}
		\sum_{t=2}^{T}\mathrm{tr}\left[\bm{H}_t\left(\bm{h}_{t-1}^*\bm{h}_{t-1}^{*,\top}-E(\bm{h}_{t-1}^*\bm{h}_{t-1}^{*,\top})\right)\right] \to_P 0,
	\end{equation*}
	where $\bm{H}_t$ is a $d^2 \times d^2$ deterministic weighting function satisfying $\left\|\bm{H}_t\right\| < \infty$.
\end{lemma}

Let $Q_T = \sum_{t=2}^{T}\bm{y}_{t}^{\top}\bm{H}_t\bm{h}_{t-1}^*$, where $\bm{H}_t$ is a $d^2\times d^2$ functional weighting matrix.
\begin{lemma}\label{L5}
	Suppose Assumptions \ref{Ass1}-\ref{Ass3} hold and $\bm{y}_t = \bm{Z}_{t-1} \bm{\eta}_t$. Then as $T \to \infty$,
	\begin{equation*}
		Q_T \to_D N\left(0,\sigma_Q^2\right).
	\end{equation*}
	where $\sigma_Q^2 = \lim_{T\to \infty} \sum_{t=2}^{T}\mathrm{tr}\left[E\left(\bm{H}_t^\top\bm{y}_t\bm{y}_{t}^{\top}\bm{H}_t\right)E\left(\bm{h}_{t-1}^*\bm{h}_{t-1}^{*,\top} \right)\right]$.
\end{lemma}

\section{Proofs of the Preliminary Lemmas}\label{AppB.2}

\begin{proof}[Proof of Lemma \ref{LemmaB.3}]
\item
\noindent (1). First, for any fixed $\tau\in(0,1)$, let $\bm{W}_{T,t}(\tau)=\frac{1}{T}\bm{W}(\tau_t)\left(\frac{\tau_t-\tau}{h}\right)^kK_h\left(\tau_t-\tau\right)$. It is straightforward to show that $\sup_{\tau\in[0,1]}\sum_{t=1}^{T}\|\bm{W}_{T,t}(\tau)\|=O(1)$,  $\sup_{\tau\in[0,1],t\geq1}\|\bm{W}_{T,t}(\tau)\|=O(1/(Th))$ and
$$
\sup_{\tau \in[0,1]}\sum_{t=1}^{T-1}\|\bm{W}_{T,t+1}(\tau)-\bm{W}_{T,t}(\tau) \|=O(1/(Th)).
$$

Second, by triangle inequality, Cauchy-Schwarz inequality and Proposition \ref{Proposition2.1},
\begin{eqnarray*}
	&&\sum_{t=1}^{T}\left\|\bm{x}_t\bm{x}_t^\top-\widetilde{\bm{x}}_t\widetilde{\bm{x}}_t^\top\right\| \leq\sum_{t=1}^{T}\left(\left\|\bm{x}_t\right\|+\left\|\widetilde{\bm{x}}_t\right\|\right)\left\|\bm{x}_t-\widetilde{\bm{x}}_t\right\| \\ &\leq&\left(\sum_{t=1}^{T}\left(\left\|\bm{x}_t\right\|+\left\|\widetilde{\bm{x}}_t\right\|\right)^2\right)^{1/2}\left(\sum_{t=1}^{T}\left\|\bm{x}_t-\widetilde{\bm{x}}_t\right\| ^2\right)^{1/2}=O_P(\sqrt{T})\cdot O_P(1/\sqrt{T})=O_P(1).
\end{eqnarray*}

In addition, similar to the proof of Proposition 2.1, we have
\begin{equation*}
	\sum_{j=0}^{\infty}j\left\|\bm{B}_j(\tau)\right\|=\sum_{j=0}^{\infty}j\left\|\bm{\Psi}_j(\tau) \bm{\omega}(\tau)\right\|\leq M \sum_{j=0}^{\infty}j \rho_A^j < \infty,
\end{equation*}
and
\begin{eqnarray*}
&&\sum_{t=1}^{T-1}\sum_{j=0}^{\infty}j\left\|\bm{B}_j(\tau_{t+1})-\bm{B}_j(\tau_{t})\right\|\\
&\leq& \sup_{\tau\in[0,1]} \sum_{j=1}^{\infty}j\|\bm{\Psi}_j(\tau)\| \sum_{t=1}^{T}\left\|\bm{\omega}(\tau_{t+1})-\bm{\omega}(\tau_{t})\right\| + M \sum_{t=1}^{T-1}\sum_{j=0}^{\infty}j\|\bm{\Psi}_j(\tau_{t+1})-\bm{\Psi}_j(\tau_{t})\|\\
&\leq&M\sum_{j=0}^{\infty}(j^2\rho_A^{j-1} + j\rho_A^{j})< \infty.
\end{eqnarray*}
Therefore, Lemma \ref{LemmaB.7} are still valid for the TV-VAR($p$) process. The proof of part (1) is completed.

\noindent (2)-(4). The proofs of part (2)-(4) can be done in a similar way to that of part (1).
\end{proof}

\begin{proof}[Proof of Lemma \ref{LemmaB.4}]

\item

\noindent (1). Let $\{S_l\}$ be a finite number of sub-intervals covering the interval $[0,1]$, which are centered at $s_l$ with the length $\delta_T=o(h^2)$. Denote the number of these intervals by $N_T$ then $N_T=O(\delta_T^{-1})$. Hence,
\begin{eqnarray*}
\sup _{\tau \in [0,1]} \left\| \frac{1}{T} \sum_{t=1}^{T}\bm{Z}_{t-1}\bm{\eta}_{t} K_h(\tau_t-\tau) \right\| 
&\leq& \max _{1 \leq l \leq N_{T}} \left\| \frac{1}{T} \sum_{t=1}^{T}\bm{Z}_{t-1}\bm{\eta}_{t} K_h(\tau_t-s_l)\right\|\\ &&+ \max_{1\leq l \leq N_T} \sup_{\tau \in S_l} \left\| \frac{1}{T} \sum_{t=1}^{T}\bm{Z}_{t-1}\bm{\eta}_{t} \left( K_h(\tau_t-\tau)-K_h(\tau_t-s_l)\right) \right\|\\
&:=&I_{T,1}+I_{T,2}.
\end{eqnarray*}

By the continuity of kernel function $K(\cdot)$ and taking $\delta_T=O(\gamma_T h^2)$  with $\gamma_T=\left(\frac{\log T}{Th} \right)^{\frac{1}{2}}$, then we have
\begin{equation*}
E|I_{T,2}|\leq M\frac{\delta_T}{h^2}E\|\bm{Z}_{t-1}\bm{\eta}_{t}\|=O(\gamma_T).
\end{equation*}

We then apply the truncation method again. Define $\bm{u}_t=\bm{Z}_{t-1}\bm{\eta}_{t}$, $\bm{u}_t^\prime = \bm{u}_t I(\|\bm{u}_t\|\leq T^{\frac{2}{\delta}} )$ and $\bm{u}_t^{\prime\prime} = \bm{u}_t-\bm{u}_t^\prime$. Then we have
\begin{eqnarray*}
I_{T,1} & = &\max_{1\leq l \leq N_T} \left\|\frac{1}{T}\sum_{t=1}^{T}\bm{u}_t^\prime+\bm{u}_t^{\prime\prime}-E(\bm{u}_t^\prime+\bm{u}_t^{\prime\prime}|\mathcal{F}_{t-1}) K_h(\tau_t-s_l)  \right\| \\
&\leq &\max_{1\leq l \leq N_T}\left\|\frac{1}{T}\sum_{t=1}^{T} \left(\bm{u}_t^\prime-E(\bm{u}_t^\prime|\mathcal{F}_{t-1})\right)K_h(\tau_t-s_l) \right\|+ \max_{1\leq l \leq N_T}\left\|\frac{1}{T}\sum_{t=1}^{T} \bm{u}_t^{\prime\prime} K_h(\tau_t-s_l) \right\|\\
&&+\max_{1\leq l \leq N_T}\left\|\frac{1}{T}\sum_{t=1}^{T} E(\bm{u}_t^{\prime\prime}|\mathcal{F}_{t-1})K_h(\tau_t-s_l) \right\|\\
&:=&I_{T,11}+I_{T,12}+I_{T,13}.
\end{eqnarray*}

Now consider $I_{T,12}$. Let $d_T=\max_{1\leq t\leq T, 1\leq l\leq N_T} K_h(\tau_t-s_l)/T$. By Holder's inequality and Chebyshev inequality, we have
\begin{eqnarray*}
E|I_{T,12}| &\leq& d_T \sum_{t=1}^{T}E\left\| \bm{\eta}_t^{\prime\prime}\right\| \leq d_T \sum_{t=1}^{T} E\left(\left\|\bm{Z}_{t-1}\bm{\eta}_t\right\|^{\delta/2}\right)^\frac{2}{\delta} \left( \frac{E\left(\left\|\bm{Z}_{t-1}\bm{\eta}_t\right\|^{\delta/2}\right)}{T}\right)^{\frac{\delta-2}{\delta}}\\
&=&O(T^\frac{2}{\delta}d_T)=o\left( \sqrt{\frac{\log T}{Th}}\right).
\end{eqnarray*}

Similarly, $ I_{T,13} =O_P(T^\frac{2}{\delta}d_T)$.

For any fixed $1\leq l \leq N_T$, let $\bm{Y}_t:=\frac{1}{T}(\bm{u}_{t}^\prime-E(\bm{u}_{t}^\prime|\mathcal{F}_{t-1}))K_{h}(\tau_t-s_l)$, then we have $E\left(\bm{Y}_t|\mathcal{F}_{t-1}\right)=0$ and $\left\|\bm{Y}_t\right\|\leq 2 T^{2/\delta}d_T$ with $d_T=\max_{1\leq t \leq T} K_{h}(\tau_t-s_l)/T$. Also, by Lemma \ref{LemmaB.3}, we have
$$
\sup_{0\leq \tau \leq 1}\frac{1}{T}\sum_{t=1}^{T}E\left(\left\|\bm{Z}_{t-1}\bm{\eta}_t\right\|^2 |\mathcal{F}_{t-1} \right)K_{h}(\tau_t-\tau)=O_P(1)
$$
which follows that
\begin{equation*}
	\max_{1\leq l \leq N_T}\left\|\sum_{t=1}^{T}E(\bm{Y}_t \bm{Y}_t^\top|\mathcal{F}_{t-1}) \right\| \leq4 \max_{1\leq l \leq N_T}\sum_{t=1}^{T}E\left( \left\|\bm{u}_t\right\|^2|\mathcal{F}_{t-1}\right)\frac{K_h^2(\tau_t-s_l)}{T^2}=O_P\left(d_T \right).
\end{equation*}
Therefore, we have $\max_{1\leq l \leq N_T}\left\|\sum_{t=1}^{T}E (\bm{Y}_t \bm{Y}_t^\top |\mathcal{F}_{t-1}) \right\|\leq \frac{M}{Th }$ in probability. By Lemma \ref{LemmaB.2}, we have
\begin{equation*}
	\begin{split}
		\Pr\left(I_{T,11} > \sqrt{8M}\gamma_T \right) & \leq \Pr\left(I_{T,11} > \sqrt{8M}\gamma_T  , \max_{1\leq l \leq N_T}\left\|\sum_{t=1}^{T}E (\bm{Y}_t \bm{Y}_t^\top |\mathcal{F}_{t-1}) \right\|\leq \frac{M}{Th }\right)  \\
		&+\Pr\left(\max_{1\leq l \leq N_T}\left\|\sum_{t=1}^{T}E (\bm{Y}_t \bm{Y}_t^\top |\mathcal{F}_{t-1}) \right\| > \frac{M}{Th }\right)\\
		& \leq N_T \exp\left(-\frac{8M\gamma_T^2}{2(\frac{M}{Th }+\gamma_T 2 T^{\frac{2}{\delta}}d_T  )}\right) + o(1)\\
		&\leq N_T\exp\left(-4 \log(T)\right)+ o(1)=o(1),
	\end{split}
\end{equation*}
if $\frac{T^{1-\frac{4}{\delta}}h}{\log T} \rightarrow \infty$, which completes the proof of part (1).

\noindent (2). Note that
\begin{eqnarray*}
&& \frac{1}{\sqrt{Th}}\sum_{t=1}^{T}\bm{\eta}_t\left(\bm{\eta}_t-\bm{\widehat{\eta}}_t\right)^\top K\left(\frac{\tau_t-\tau}{h}\right)  = \frac{1}{\sqrt{Th}}\sum_{t=1}^{T}\bm{\eta}_t\bm{z}_{t-1}^\top\left(\bm{\widehat{A}}(\tau_t)-\bm{A}(\tau_t)\right)^\top K\left(\frac{\tau_t-\tau}{h}\right) \\
&=& \frac{1}{\sqrt{Th}}\sum_{t=1}^{T}\bm{\eta}_t[\bm{z}_{t-1}^\top,\bm{0}_{(d^2p+d)\times 1}^\top] \frac{1}{2} h^2 \bm{S}_{T}^{-1}(\tau_t) \left(\begin{matrix}
                                    \bm{S}_{T,2}(\tau_t) \\
                                    \bm{S}_{T,3}(\tau_t)
                                  \end{matrix}\right) \bm{A}^{(2),\top}(\tau_t)K\left(\frac{\tau_t-\tau}{h}\right)\\
&&+\frac{1}{\sqrt{Th}}\sum_{t=1}^{T}\bm{\eta}_t[\bm{z}_{t-1}^\top,\bm{0}_{(d^2p+d)\times 1}^\top] \bm{S}_{T}^{-1}(\tau_t) \left(\frac{1}{Th}\sum_{s=1}^{T}\left(\bm{z}_{s-1}^*\bm{z}_{s-1}^\top\otimes\bm{I}_d\right) \bm{M}^\top(\tau_s)K\left(\frac{\tau_s-\tau_t}{h}\right) \right)K\left(\frac{\tau_t-\tau}{h}\right)\\
&&+\frac{1}{\sqrt{Th}}\sum_{t=1}^{T}\bm{\eta}_t[\bm{z}_{t-1}^\top,\bm{0}_{(d^2p+d)\times 1}^\top]\bm{S}_{T}^{-1}(\tau_t) \left(\frac{1}{Th}\sum_{s=1}^{T}\bm{Z}_{s-1}\bm{\eta}_s^\top K\left(\frac{\tau_s-\tau_t}{h}\right)\right)K\left(\frac{\tau_t-\tau}{h}\right)\\
&:=&J_{T,1}+J_{T,2}+J_{T,3}.
\end{eqnarray*}

For $J_{T,1}$ to $J_{T,2}$, using Lemmas \ref{LemmaB.3}.2-3, we can replace the sample covariance matrix with its converged and deterministic value with rate $O_P\left(\sqrt{\frac{\log T}{Th}}\right)$ and hence it's easy to show that $J_{T,1}$ to $J_{T,2}$ are $o_P(1)$.

For $J_{T,3}$, for notational simplicity, we ignore $\bm{S}_{T}^{-1}(\tau_t)$ and hence,
\begin{equation*}
  \begin{split}
      J_{T,3}=\ &\frac{1}{(Th)^{3/2}}\sum_{t=1}^{T}\bm{\eta}_t \bm{z}_{t-1}^\top \bm{z}_{t-1}\bm{\eta}_t^\top K(0)K\left(\frac{\tau_t-\tau}{h}\right) \\
                &+\frac{1}{(Th)^{3/2}}\sum_{i=1}^{T-1} \sum_{t=1}^{T-i}\bm{\eta}_t \bm{z}_{t-1}^\top \bm{z}_{t+i-1}\bm{\eta}_{t+i}^\top K\left(\frac{i}{Th}\right)K\left(\frac{\tau_t-\tau}{h}\right) \\
                &+\frac{1}{(Th)^{3/2}}\sum_{i=1}^{T-1} \sum_{t=1}^{T-i}\bm{\eta}_{t+i} \bm{z}_{t+i-1}^\top \bm{z}_{t-1}\bm{\eta}_{t}^\top K\left(\frac{i}{Th}\right)K\left(\frac{\tau_t-\tau}{h}\right) \\
      :=\ &J_{T,31}+J_{T,32}+J_{T,33}.\\
  \end{split}
\end{equation*}
It's easy to see $J_{T,31}=O_P\left((Th)^{-1/2}\right)$. For $J_{T,32}$,
\begin{equation*}
  \begin{split}
      J_{T,32}=\ &\frac{1}{(Th)^{3/2}}\sum_{i=1}^{T-1} \sum_{t=1}^{T-i}\bm{\eta}_t E\left(\bm{z}_{t-1}^\top \bm{z}_{t+i-1}\right)\bm{\eta}_{t+i}^\top K\left(\frac{i}{Th}\right)K\left(\frac{\tau_t-\tau}{h}\right) \\
      &+\frac{1}{(Th)^{3/2}}\sum_{i=1}^{T-1} \sum_{t=1}^{T-i}\bm{\eta}_t\left(\bm{z}_{t-1}^\top \bm{z}_{t+i-1}-E\left(\bm{z}_{t-1}^\top \bm{z}_{t+i-1}\right)\right)\bm{\eta}_{t+i}^\top K\left(\frac{i}{Th}\right)K\left(\frac{\tau_t-\tau}{h}\right) \\
  :=\ &J_{T,321}+J_{T,322}.
  \end{split}
\end{equation*}

For $J_{T,321}$,
\begin{equation*}
  \begin{split}
     E\left\|J_{T,321}\right\|^2\leq\ &\frac{1}{(Th)^3}\sum_{i=1}^{T-1}\sum_{t=1}^{T-i}\left\{E\left(\bm{z}_{t-1}^\top \bm{z}_{t+i-1}\right)E\left(\bm{\eta}_t^\top\bm{\eta}_t\right)E\left(\bm{\eta}_{t+i}^\top\bm{\eta}_{t+i}\right)\right\}^2K^2\left(\frac{i}{Th}\right)K^2\left(\frac{\tau_t-\tau}{h}\right)  \\
      =\ & O\left(\frac{1}{Th} \right),
  \end{split}
\end{equation*}
which then yields that $J_{T,321}=O_P\left((Th)^{-1/2}\right)$. For $J_{T,322}$,
\begin{equation*}
  J_{T,322}=\frac{1}{(Th)^{3/2}}\sum_{i=1}^{T-1} \sum_{t=1}^{T-i}\bm{\eta}_t\sum_{m=1}^{p}\left(\bm{x}_{t-m}^\top \bm{x}_{t+i-m}-E\left(\bm{x}_{t-m}^\top \bm{x}_{t+i-m}\right)\right)\bm{\eta}_{t+i}^\top K\left(\frac{i}{Th}\right)K\left(\frac{\tau_t-\tau}{h}\right).
\end{equation*}

For notational simplicity, let $p=1$ and thus
\begin{eqnarray*}
J_{T,322}&=&\frac{1}{(Th)^{3/2}}\sum_{i=1}^{T-1} \sum_{t=1}^{T-i}\bm{\eta}_t\left(\bm{\mu}_{t-1}^{\top}\sum_{j=0}^{\infty} \bm{\Psi}_{j,t+i-1}\bm{\eta}_{t+i-1-j} \right)\bm{\eta}_{t+i}^\top K\left(\frac{i}{Th}\right)K\left(\frac{\tau_t-\tau}{h}\right)\\
       && +\frac{1}{(Th)^{3/2}}\sum_{i=1}^{T-1} \sum_{t=1}^{T-i}\bm{\eta}_t\left(\sum_{j=0}^{\infty}\bm{\eta}_{t-1-j}^\top \bm{\Psi}_{j,t-1}^{\top} \bm{\mu}_{t+i-1}\right)\bm{\eta}_{t+i}^\top K\left(\frac{i}{Th}\right)K\left(\frac{\tau_t-\tau}{h}\right)\\
       && +\frac{1}{(Th)^{3/2}}\sum_{i=1}^{T-1} \sum_{t=1}^{T-i}\bm{\eta}_t\left(\sum_{j=0}^{\infty}\left(\bm{\eta}_{t-1-j}^\top\otimes \bm{\eta}_{t-1-j}^\top-E\left(\bm{\eta}_{t-1-j}^\top\otimes \bm{\eta}_{t-1-j}^\top\right) \right) \right.\\
       &&\cdot\left.\mathrm{vec}\left(\bm{\Psi}_{j,t-1}^{\top} \bm{\Psi}_{j+i,t+i-1}\right) \right)\bm{\eta}_{t+i}^\top K\left(\frac{i}{Th}\right)K\left(\frac{\tau_t-\tau}{h}\right)\\
       &&+\frac{1}{(Th)^{3/2}}\sum_{i=1}^{T-1} \sum_{t=1}^{T-i}\bm{\eta}_t\left(\sum_{j=0}^{\infty}\sum_{m=0,\neq j+i}^{\infty}\left(\bm{\eta}_{t+i-1-m}^\top\otimes \bm{\eta}_{t-1-j}^\top\right)
       \mathrm{vec}\left(\bm{\Psi}_{j,t-1}^{\top} \bm{\Psi}_{m,t+i-1}\right) \right)\\
       &&\cdot\bm{\eta}_{t+i}^\top K\left(\frac{i}{Th}\right)K\left(\frac{\tau_t-\tau}{h}\right)\\
  &:=&J_{T,3221}+J_{T,3222}+J_{T,3223}+J_{T,3224}.
\end{eqnarray*}

For $J_{T,3221}$,
\begin{eqnarray*}
&&E\left\|J_{T,3221}\right\|^2 \\ &\leq&\frac{1}{(Th)^3}\sum_{i=1}^{T-1}\sum_{t=1}^{T-i}E\left\|\bm{\eta}_t\left(\bm{\mu}_{t-1}^{\top}\sum_{j=0}^{\infty} \bm{\Psi}_{j,t+i-1}\bm{\eta}_{t+i-1-j} \right)\right\|^2 E\left\|\bm{\eta}_{t+i}\right\|^2 K^2\left(\frac{i}{Th}\right)K^2\left(\frac{\tau_t-\tau}{h}\right)\\
&=&O\left(\frac{1}{Th} \right).
\end{eqnarray*}

Similarly, $J_{T,3222}$ and $J_{T,3223}$ are $O_P\left((Th)^{-1/2}\right)$. For $J_{T,3224}$,
\begin{eqnarray*}
J_{T,3224}&=&\frac{1}{(Th)^{3/2}}\sum_{i=1}^{T-1} \sum_{t=1}^{T-i}\bm{\eta}_t\left(\sum_{j=0}^{\infty}\left(\bm{\eta}_{t}^\top\otimes \bm{\eta}_{t-1-j}^\top\right)
       \mathrm{vec}\left(\bm{\Psi}_{j,t-1}^{\top} \bm{\Psi}_{i-1,t+i-1}\right) \right)\bm{\eta}_{t+i}^\top K\left(\frac{i}{Th}\right)K\left(\frac{\tau_t-\tau}{h}\right)\\
       &&+\frac{1}{(Th)^{3/2}}\sum_{i=1}^{T-1} \sum_{t=1}^{T-i}\bm{\eta}_t\left(\sum_{j=0}^{\infty}\sum_{m=0,\neq j+i,\neq i-1}^{\infty}\left(\bm{\eta}_{t+i-1-m}^\top\otimes \bm{\eta}_{t-1-j}^\top\right) \right.\\
       &&\cdot\left.\mathrm{vec}\left(\bm{\Psi}_{j,t-1}^{\top} \bm{\Psi}_{m,t+i-1}\right) \right)\bm{\eta}_{t+i}^\top K\left(\frac{i}{Th}\right)K\left(\frac{\tau_t-\tau}{h}\right) := J_{T,32241}+J_{T,32242}.
\end{eqnarray*}

Similar to the proof of $J_{T,3221}$, we can show that $J_{T,32242}=O_P\left((Th)^{-1/2}\right)$.

Let $\bm{w}_{t,i}=\sum_{j=0}^{\infty}\left(\bm{\eta}_t\bm{\eta}_t^\top \otimes\bm{\eta}_{t-1-j}^\top \right)\mathrm{vec}\left(\bm{\Psi}_{j,t-1}^{\top} \bm{\Psi}_{i-1,t+i-1}^{\top}\right)$. For $J_{T,32241}$,
\bea
&& E\left\|J_{T,32241}\right\|^2 = \frac{1}{(Th)^3}\sum_{i_1=1}^{T-1}\sum_{i_2=1}^{T-1}\sum_{t_1=1}^{T-i_1}E\left\|\bm{\eta}_{t_1+i_1}^\top\bm{\eta}_{t_1+i_1}\right\|E\left\|\bm{w}_{t_1+i_1-i_2,i_2}^\top \bm{w}_{t_1,i_1}\right\|
\nonumber\\
&&\cdot K\left(\frac{i_1}{Th} \right)K\left(\frac{i_2}{Th} \right)K\left(\frac{\tau_{t_1}-\tau}{h} \right)K\left( \frac{\tau_{t_1+i_1-i_2}-\tau}{h}\right)
\nonumber\\
&& \leq \frac{M}{(Th)^3}\sum_{t_1=1}^{T}\left(\max_t\sum_{i=1}^{T-1}\left\|\bm{\Psi}_{i,t}\right\|\right)^2\left(\max_t\sum_{j=0}^{\infty}\left\|\bm{\Psi}_{j,t}\right\|\right)^2K\left( \frac{\tau_{t_1}-\tau}{h}\right) = O\left((Th)^{-2}\right).
\nonumber
\eea

Hence, $J_{T,32}=O_P\left((Th)^{-1/2}\right)$. Similar to $J_{T,32}$, $J_{T,33}=O_P\left((Th)^{-1/2}\right)$. The proof is now completed.

(3). By Lemma \ref{LemmaB.3}, we have
\begin{eqnarray*}
\sup_{\tau \in [h,1-h]}\left\|\widehat{\bm{\Sigma}}(\tau)-\bm{\Sigma}(\tau)\right\| = O_P\left(h^2 + \left(\frac{\log T}{Th}\right)^{1/2}\right).
\end{eqnarray*}
Then we just need to focus on the rate associated with $\widehat{\bm{\Omega}}(\tau)$. For notational simplicity, we ignore the $\frac{1}{T}\sum_{t=1}^{T}K_h(\tau_t-\tau)$, because
\begin{eqnarray*}
	\frac{1}{T}\sum_{t=1}^{T}K_h(\tau_t-\tau) = 1+ O((Th)^{-1})
\end{eqnarray*}
uniformly over $\tau \in [h,1-h]$.

Write
\begin{eqnarray*}
	\widehat{\bm{\Omega}}(\tau)&=&\frac{1}{Th}\sum_{t=1}^{T}\bm{\widehat{\eta}}_t\bm{\widehat{\eta}}_t^\top K\left(\frac{\tau_t-\tau}{h}\right) \\
	&=&\frac{1}{Th}\sum_{t=1}^{T}\left(\bm{\eta}_t+\bm{\widehat{\eta}}_t-\bm{\eta}_t\right)\left(\bm{\eta}_t+\bm{\widehat{\eta}}_t-\bm{\eta}_t\right)^\top K\left(\frac{\tau_t-\tau}{h}\right)\\
	&=& \frac{1}{Th}\sum_{t=1}^{T}\bm{\eta}_t\bm{\eta}_t^\top K\left(\frac{\tau_t-\tau}{h}\right)+\frac{1}{Th}\sum_{t=1}^{T}(\bm{\widehat{\eta}}_t-\bm{\eta}_t)(\bm{\widehat{\eta}}_t-\bm{\eta}_t)^\top K\left(\frac{\tau_t-\tau}{h}\right) \\
	& &+\frac{1}{Th}\sum_{t=1}^{T}\bm{\eta}_t(\bm{\widehat{\eta}}_t-\bm{\eta}_t)^\top K\left(\frac{\tau_t-\tau}{h}\right)+\frac{1}{Th}\sum_{t=1}^{T}(\bm{\widehat{\eta}}_t-\bm{\eta}_t)\bm{\eta}_t^\top K\left(\frac{\tau_t-\tau}{h}\right)\\
	&:=& I_{T,1}+I_{T,2}+I_{T,3}+I_{T,4}.
\end{eqnarray*}

Consider $I_{T,1}$. Similar to the proof of part (1), we have
\begin{eqnarray*}
	\sup_{\tau \in [0,1]}\left\|\frac{1}{T}\sum_{t=1}^{T}\left[\bm{\eta}_t\bm{\eta}_t^\top - E\left(\bm{\eta}_t\bm{\eta}_t^\top\right)\right]K_h(\tau_t-\tau) \right\|=O_P\left(\sqrt{\frac{\log T}{Th}}\right).
\end{eqnarray*}
Next, consider $I_{T,2}$. By Lemma \ref{LemmaB.3}, we have
\begin{eqnarray*}
	&&\sup_{\tau \in [0,1]} \frac{1}{T}\sum_{t=1}^{T}\left\|\bm{Z}_{t-1}\right\|^2K_h(\tau_t-\tau)\\
	&\leq&\sup_{\tau \in [0,1]} \left|\frac{1}{T}\sum_{t=1}^{T}\mathrm{tr}\left[\bm{Z}_{t-1}^\top\bm{Z}_{t-1}-E(\bm{Z}_{t-1}^\top\bm{Z}_{t-1})\right]K_h(\tau_t-\tau)\right|\\
	&&+\sup_{\tau\in[0,1]}\left|\frac{1}{T}\sum_{t=1}^{T}\mathrm{tr}\left[E(\bm{Z}_{t-1}^\top\bm{Z}_{t-1})\right]K_h(\tau_t-\tau)\right|\\
	&=&o_P(1)+O(1)=O_P(1).
\end{eqnarray*}

Hence, by the first result of Theorem \ref{Thm2.1}
\begin{eqnarray*}
	\sup_{\tau \in [0,1]}\left\|I_{T,2}\right\| &\leq& \sup_{\tau \in [0,1]} \|\widehat{\bm{\beta}}(\tau) - \bm{\beta}(\tau) \|^2\cdot\sup_{\tau \in [0,1]}\frac{1}{T}\sum_{t=1}^{T} \left\|\bm{Z}_{t-1}\right\|^2 K_h(\tau_t-\tau) = o_P\left(h^2 + \sqrt{\frac{\log T}{Th}}\right).
\end{eqnarray*}

Similarly, for $I_{T,3}$ and $I_{T,4}$, we have
\begin{eqnarray*}
	\sup_{\tau \in [0,1]}\left\|I_{T,3}\right\| &\leq& \sup_{\tau \in [0,1]} \left\|\widehat{\bm{\beta}}(\tau) - \bm{\beta}(\tau)\right\|\cdot\sup_{\tau \in [0,1]}\frac{1}{T}\sum_{t=1}^{T} \left\|\bm{Z}_{t-1}\bm{\eta}_t\right\| K_h(\tau_t-\tau) \\
	&\leq& \sup_{\tau \in [0,1]} \left\|\widehat{\bm{\beta}}(\tau) - \bm{\beta}(\tau)\right\|\cdot\left\{\sup_{\tau \in [0,1]}\frac{1}{T}\sum_{t=1}^{T} \left\|\bm{Z}_{t-1}\right\|^2 K_h(\tau_t-\tau)\right\}^{1/2} \\
	&&\cdot \left\{\sup_{\tau \in [0,1]}\frac{1}{T}\sum_{t=1}^{T} \left\|\bm{\eta}_t\right\|^2 K_h(\tau_t-\tau)\right\}^{1/2}\\
	&=& O_P\left(h^2 + \sqrt{\frac{\log T}{Th}}\right).
\end{eqnarray*}

The proof is now completed.
\end{proof}

Define $\bm{\Lambda}_{\mathsf{p}}(\tau)=[\bm{a}(\tau),\bm{A}_{\mathsf{p},1}(\tau),...,\bm{A}_{\mathsf{p},\mathsf{p}}(\tau)]$, where $\bm{A}_{\mathsf{p},j}(\tau)=\bm{A}_{j}(\tau)$ for $1\leq j \leq p$ and $\bm{A}_{\mathsf{p},j}(\tau)=0$ for $j > p$. Let $\bm{z}_{\mathsf{p},t-1}=[1,\bm{x}_{t-1}^\top,...,\bm{x}_{t-\mathsf{p}}^\top]^\top$, $\bm{z}_{\mathsf{p},t-1}^* =\left[\bm{z}_{\mathsf{p},t-1}^\top, \frac{\tau_t-\tau}{h}\bm{z}_{\mathsf{p},t-1}^\top \right]^\top$, $\bm{Z}_{\mathsf{p},t}^* = \bm{z}_{\mathsf{p},t}^*\otimes \bm{I}_d$, $\bm{M}_\mathsf{p}(\tau_t)=\bm{\Lambda}_\mathsf{p}(\tau_t)-\bm{\Lambda}_\mathsf{p}(\tau)-\bm{\Lambda}_\mathsf{p}^{(1)}(\tau)(\tau_t-\tau)-\frac{1}{2}h^2\bm{\Lambda}_\mathsf{p}^{(2)}(\tau)(\tau_t-\tau)^2$,
$\bm{\Lambda}_{\overline{\mathsf{p}}}(\tau)=[\bm{A}_{\mathsf{p},\mathsf{p}+1}(\tau),...,\bm{A}_{\mathsf{p},\mathsf{P}}(\tau)]$ and $\bm{z}_{\overline{\mathsf{p}},t-1}=[\bm{x}_{t-\mathsf{p}-1}^\top,...,\bm{x}_{t-\mathsf{P}}^\top]^\top$.

\begin{proof}[Proof of Lemma \ref{LemmaB.5}]
  \item

\noindent (1). Since $\mathsf{p}\geq p$, we have $\widehat{\bm{\eta}}_{\mathsf{p},t}=\bm{\eta}_t+\left(\bm{\Lambda}_\mathsf{p}(\tau_t)-\widehat{\bm{\Lambda}}_\mathsf{p}(\tau_t)\right) \bm{z}_{\mathsf{p},t-1}$ and
\begin{eqnarray*}
  \text{RSS}(\mathsf{p})&=&\frac{1}{T}\sum_{t=1}^{T}\bm{\eta}_t^\top\bm{\eta}_t+\frac{1}{T}\sum_{t=1}^{T}\bm{z}_{\mathsf{p},t-1}^\top\left(\bm{\Lambda}_\mathsf{p}(\tau_t)-\widehat{\bm{\Lambda}}_\mathsf{p}(\tau_t)\right)^\top \left(\bm{\Lambda}_\mathsf{p}(\tau_t)-\widehat{\bm{\Lambda}}_\mathsf{p}(\tau_t)\right) \bm{z}_{\mathsf{p},t-1}\\
  &&-2\frac{1}{T}\sum_{t=1}^{T}\mathrm{tr}\left(\bm{\eta}_t\left(\bm{\eta}_t-\widehat{\bm{\eta}}_{\mathsf{p},t}\right)^\top\right) :=  \frac{1}{T}\sum_{t=1}^{T}\bm{\eta}_t^\top\bm{\eta}_t+ I_{T,1}+I_{T,2}.
\end{eqnarray*}

Since $\bm{\eta}_t^\top\bm{\eta}_t-E(\bm{\eta}_t^\top\bm{\eta}_t)$ is m.d.s., we have $\frac{1}{T}\sum_{t=1}^{T}\bm{\eta}_t^\top\bm{\eta}_t=\frac{1}{T}\sum_{t=1}^{T}E(\bm{\eta}_t^\top\bm{\eta}_t)+T^{-1/2}$. By Theorem \ref{Thm2.1}.1,
\begin{eqnarray*}
  I_{T,1} &\leq& \frac{1}{T}\sum_{t=1}^{T}\left\|\bm{z}_{\mathsf{p},t-1}\right\|^2 \cdot \left\|\widehat{\bm{\Lambda}}_\mathsf{p}(\tau_t)-\bm{\Lambda}_\mathsf{p}(\tau_t)\right\|^2\leq\sup_{0\leq \tau \leq 1}\left\|\widehat{\bm{\Lambda}}_\mathsf{p}(\tau)-\bm{\Lambda}_\mathsf{p}(\tau)\right\|^2\cdot\frac{1}{T}\sum_{t=1}^{T}\left\|\bm{z}_{\mathsf{p},t-1}\right\|^2\\
          &=&O_P\left( (h^2+(\log T/(Th))^{1/2})^2 \right).
\end{eqnarray*}

For $I_{T,2}$,
\begin{eqnarray*}
    && \frac{1}{T}\sum_{t=1}^{T} \bm{\eta}_t\left(\bm{\eta}_t-\widehat{\bm{\eta}}_{\mathsf{p},t}\right)^\top = \frac{1}{T}\sum_{t=1}^{T} \bm{\eta}_t\bm{z}_{\mathsf{p},t-1}^\top\left(\bm{\widehat{\Lambda}}_\mathsf{p}(\tau_t)-\bm{\Lambda}_\mathsf{p}(\tau_t)\right)^\top \\
      &&=\frac{1}{2} h^2 \cdot \frac{1}{T}\sum_{t=1}^{T}\bm{\eta}_t[\bm{z}_{\mathsf{p},t-1}^\top,\bm{0}_{(d^2\mathsf{p}+d)\times 1}^\top] \bm{S}_{T}^{-1}(\tau_t) \left(\begin{matrix}
                                      \bm{S}_{T,2}(\tau_t) \\
                                      \bm{S}_{T,3}(\tau_t)
                                    \end{matrix}\right)\bm{A}_\mathsf{p}^{(2),\top}(\tau_t)\\
      &&+\frac{1}{T}\sum_{t=1}^{T}\bm{\eta}_t[\bm{z}_{\mathsf{p},t-1}^\top,\bm{0}_{(d^2\mathsf{p}+d)\times 1}^\top]\bm{S}_{T}^{-1}(\tau_t) \left(\frac{1}{Th}\sum_{s=1}^{T}\left(\bm{z}_{\mathsf{p},s-1}^*\bm{z}_{\mathsf{p},s-1}^\top\otimes \bm{I}_d \right) \bm{M}_\mathsf{p}^\top(\tau_s)K\left(\frac{\tau_s-\tau_t}{h}\right) \right)\\
      &&+\frac{1}{T}\sum_{t=1}^{T}\bm{\eta}_t[\bm{z}_{\mathsf{p},t-1}^\top,\bm{0}_{(d^2\mathsf{p}+d)\times 1}^\top] \bm{S}_{T}^{-1}(\tau_t) \left(\frac{1}{Th}\sum_{s=1}^{T}\bm{Z}_{\mathsf{p},s-1}^*\bm{\eta}_s^\top K\left(\frac{\tau_s-\tau_t}{h}\right)\right) := I_{T,3}+I_{T,4}+I_{T,5}.
\end{eqnarray*}

By the uniform convergence results stated in Lemmas \ref{LemmaB.3}.2-3, we replace the weighed sample covariance with its limit plus the rate $O_P\left((\log T/(Th))^{1/2}\right)$, and hence
\begin{eqnarray*}
\|I_{T,3}\|+\|I_{T,4}\|=O_P\left(T^{-\frac{1}{2}}h^2+h^2(\log T/(Th))^{1/2}\right).
\end{eqnarray*}

For $I_{T,5}$, let $\bm{\Sigma}(\tau)=\plim_{T\to \infty}\bm{S}_{T,0}(\tau)$, we have
$$
I_{T,6}=\frac{1}{T}\sum_{t=1}^{T} \bm{\eta}_t\bm{z}_{\mathsf{p},t-1}^\top \bm{\Sigma}^{-1}(\tau_t) \left(\frac{1}{Th}\sum_{s=1}^{T}\bm{z}_{\mathsf{p},s-1}\bm{\eta}_s^\top K\left(\frac{\tau_s-\tau_t}{h}\right)\right)+O_P\left((Th)^{-1/2}\cdot(h^2+\sqrt{\log T/(Th)})\right).
$$
Similar to the proof of $J_{T,4}$ in Lemma \ref{LemmaB.4}, we can show
$$
\frac{1}{T}\sum_{t=1}^{T} \bm{\eta}_t\bm{z}_{\mathsf{p},t-1}^\top \bm{\Sigma}^{-1}(\tau_t) \left(\frac{1}{Th}\sum_{s=1}^{T}\bm{z}_{\mathsf{p},s-1}\bm{\eta}_s^\top K\left(\frac{\tau_s-\tau_t}{h}\right)\right)=O_P((Th)^{-1}).
$$
Since $(Th)^{-1}+T^{-\frac{1}{2}}h^{\frac{3}{2}}=o\left(\rho_T^2\right)$, result (1) follows.

\noindent (2). For $\mathsf{p} < p$, we have $\widehat{\bm{\Lambda}}_\mathsf{p}(\tau)-\bm{\Lambda}_{\mathsf{p}}(\tau)=\bm{B}_\mathsf{p}(\tau)+o_P(1)$ uniformly over $\tau \in [0,1]$, where $\bm{B}_\mathsf{p}(\tau)$ is a nonrandom bias term. Since $\widehat{\bm{\eta}}_{\mathsf{p},t}=\bm{\eta}_{t}+\left(\bm{\Lambda}_{\mathsf{p}}(\tau_t)-\widehat{\bm{\Lambda}}_\mathsf{p}(\tau_t)\right)\bm{z}_{\mathsf{p},t-1}+\bm{\Lambda}_{\overline{\mathsf{p}}}(\tau_t)\bm{z}_{\overline{\mathsf{p}},t-1}$, by Lemma \ref{LemmaB.3}.4, we have
$$
  \text{RSS}(\mathsf{p})= \frac{1}{T}\sum_{t=1}^{T}E\left(\bm{\eta}_{t}^\top\bm{\eta}_{t}\right)+\frac{1}{T}\sum_{t=1}^{T}\mathrm{tr}\left( [\bm{B}_\mathsf{p}(\tau_t),\bm{\Lambda}_{\overline{\mathsf{p}}}(\tau_t)]E\left(\bm{z}_{\mathsf{P},t-1}\bm{z}_{\mathsf{P},t-1}^\top\right)[\bm{B}_\mathsf{p}(\tau_t),\bm{\Lambda}_{\overline{\mathsf{p}}}(\tau_t)]^\top \right)+o_P(1).
$$
Since $[\bm{B}_\mathsf{p}(\tau_t),\bm{\Lambda}_{\overline{\mathsf{p}}}(\tau_t)]\neq 0$ and $E\left(\bm{z}_{\mathsf{P},t-1}\bm{z}_{\mathsf{P},t-1}^\top\right)$ is a positive definite matrix, the result follows.
\end{proof}

\begin{proof}[Proof of Lemma \ref{L1}]
	\item
	
	\noindent (1). Write
	\begin{eqnarray*}
		\sigma^2 &=& \frac{1}{T^2h}\sum_{t=2}^{T}\sum_{s=1}^{t-1}\left[\int_{-1}^{1}K\left(u\right)K\left(u+\frac{t-s}{Th}\right)\mathrm{d}u\right]^2 \\
		&=& \frac{1}{T^2h}\sum_{t=2}^{T}\sum_{j=1}^{t-1}\left[\int_{-1}^{1}K\left(u\right)K\left(u+\frac{j}{Th}\right)\mathrm{d}u\right]^2 \\
		&=&\frac{1}{Th}\sum_{j=1}^{T-1}(1-j/T)\left[\int_{-1}^{1}K\left(u\right)K\left(u+\frac{j}{Th}\right)\mathrm{d}u\right]^2   \\
		&=& \int_{0}^{\infty}(1-vh)\left[\int_{-1}^{1}K\left(u\right)K\left(u+v\right)\mathrm{d}u\right]^2\mathrm{d}v+O(1/(Th))\\
		&\to&\int_{0}^{2}\left[\int_{-1}^{1-v}K\left(u\right)K\left(u+v\right)\mathrm{d}u\right]^2\mathrm{d}v.
	\end{eqnarray*}
	
	\noindent (2). Write
	\begin{eqnarray*}
		\max_t|a_t| &=& \max_t\left|\sum_{i=1}^{t-1}\frac{1}{T^2h}\left[\int_{-1}^{1}K\left(u\right)K\left(u+\frac{i}{Th}\right)\mathrm{d}u\right]^2\right| \\
		&\leq& \frac{1}{T^2h}\sum_{i=1}^{T}\left[\int_{-1}^{1}K\left(u\right)K\left(u+\frac{i}{Th}\right)\mathrm{d}u\right]^2 \\
		&=& \frac{1}{T}\int_{0}^{\infty}\left[\int_{-1}^{1}K\left(u\right)K\left(u+v\right)\mathrm{d}u\right]^2 \mathrm{d}v(1+o(1)) \\
		&=& O(1/T).
	\end{eqnarray*}
	
	\noindent (3). Write
	\begin{eqnarray*}
		\sum_{s=1}^{T-J}w_{s,s+J}^2 &=& \sum_{s=1}^{T-J}\frac{1}{T^2h}\left[\int_{-1}^{1}K\left(u\right)K\left(u+\frac{J}{Th}\right)\mathrm{d}u\right]^2 \\
		&=& \frac{T-J}{T^2h}\left[\int_{-1}^{1}K\left(u\right)K\left(u+\frac{J}{Th}\right)\mathrm{d}u\right]^2 = O(1/(Th)).
	\end{eqnarray*}
	
	\noindent (4). Write
	\begin{eqnarray*}
		&&T\sum_{s=1}^{T-1}b_s^2\\
		&=& \frac{1}{T^3h^2}\sum_{j=1}^{T-1}\sum_{t=1+j}^{T}\sum_{s=1+j}^{T}\left[\int_{-1}^{1}K\left(u\right)K\left(u+\frac{t-j}{Th}\right)\mathrm{d}u\right]^2 \left[\int_{-1}^{1}K\left(u\right)K\left(u+\frac{s-j}{Th}\right)\mathrm{d}u\right]^2 \\
		&=& \frac{1}{T^3h^2}\sum_{j=1}^{T-1}\sum_{i=1}^{T-j}\sum_{k=1}^{T-j}\left[\int_{-1}^{1}K\left(u\right)K\left(u+\frac{i}{Th}\right)\mathrm{d}u\right]^2 \left[\int_{-1}^{1}K\left(u\right)K\left(u+\frac{k}{Th}\right)\mathrm{d}u\right]^2 \\
		&\leq& \frac{1}{T^3h^2}\sum_{j=1}^{T}\sum_{i=1}^{T}\sum_{k=1}^{T}\left[\int_{-1}^{1}K\left(u\right)K\left(u+\frac{i}{Th}\right)\mathrm{d}u\right]^2 \left[\int_{-1}^{1}K\left(u\right)K\left(u+\frac{k}{Th}\right)\mathrm{d}u\right]^2 \\
		&\simeq& \left(\frac{1}{Th}\sum_{i=1}^{T}\left[\int_{-1}^{1}K\left(u\right)K\left(u+\frac{i}{Th}\right)\mathrm{d}u\right]^2\right)^2=O(1).
	\end{eqnarray*}
	
	\noindent (5). By Cauchy-Schwarz inequality,
	\begin{eqnarray*}
		&& \sum_{k=1}^{T-1}\sum_{t=1}^{k-1}\left[\sum_{j=k+1}^{T}w_{k,j}w_{t,j}\right]^2 \\
		&\leq& \frac{1}{T^4h^2}\sum_{k=1}^{T-1}\sum_{t=1}^{k-1}\left(\sum_{j=1+k}^{T}\left[\int_{-1}^{1}K\left(u\right)K\left(u+\frac{j-k}{Th}\right)\mathrm{d}u\right]^2\right) \left(\sum_{j=1+k}^{T}\left[\int_{-1}^{1}K\left(u\right)K\left(u+\frac{j-t}{Th}\right)\mathrm{d}u\right]^2\right)\\
		&\leq&\frac{M}{T^3h}\sum_{k=1}^{T-1}\sum_{t=1}^{k-1}\sum_{j=1+k}^{T}\left[\int_{-1}^{1}K\left(u\right)K\left(u+\frac{j-t}{Th}\right)\mathrm{d}u\right]^2  \\
		&\leq& \frac{M}{T^3h}\sum_{k=1}^{T-1}\sum_{j=2}^{T}\sum_{t=1}^{j-1}\left[\int_{-1}^{1}K\left(u\right)K\left(u+\frac{j-t}{Th}\right)\mathrm{d}u\right]^2 =O(1/T).
	\end{eqnarray*}
	
\end{proof}

\begin{proof}[Proof of Lemma \ref{L2}]
	\item
	
	Write $\bm{\xi}_t = [\xi_{t,1},...,\xi_{t,d}]^\top$. We first prove
	$$
	\normmm{\sum_{t=1}^{T}\xi_{t,i}}_p^{p^*} \leq M \sum_{t=1}^{T}\normmm{\xi_{t,i}}_p^{p^*}
	$$
	for $1\leq i \leq d$.
	
	By Burkholder inequality, the Minkowski inequality and the inequality that $|\sum_{i=1}^{n}a_i|^p\leq \sum_{i=1}^{n}|a_i|^p$ for $0 < p \leq 1$, we have
	\begin{eqnarray*}
		\normmm{\sum_{t=1}^{T}\xi_{t,i}}_p^{p^*} &\leq& \left\{ M E\left[ \left(\sum_{t=1}^{T}|\xi_{t,i}|^2\right)^{p/2} \right] \right\}^{p^*/p} \leq  M \left\{\sum_{t=1}^{T}\left(E\left[|\xi_{t,i}|^p\right]\right)^{2/p} \right\}^{p^*/2} \\
		&\leq&  M\sum_{t=1}^{T}\left(E\left[|\xi_{t,i}|^p\right]\right)^{p^*/p} = M\sum_{t=1}^{T}\normmm{\xi_{t,i}}_p^{p^*}.
	\end{eqnarray*}
	
	In addition, since $|\sum_{i=1}^{d}a_i|^p\leq \sum_{i=1}^{d}|a_i|^p$ for $p \in (0,1]$, $|\sum_{i=1}^{d}a_i|^p\leq d^{p-1}\sum_{i=1}^{d}|a_i|^p$ for $p > 1$ and $d$ is a fixed value, we have
	\begin{eqnarray*}
		\normmm{\sum_{t=1}^{T}\bm{\xi}_{t}}_p^{p^*}&=&\left\{E\left[ \left(\sum_{i=1}^{d} \xi_{.,i}^2\right)^{p/2} \right]\right\}^{p^*/p} \leq M \left\{\sum_{i=1}^{d} E\left| \xi_{.,i} \right|^p\right\}^{p^*/p}\\
		&\leq& M \sum_{i=1}^{d} \normmm{\sum_{t=1}^{T}\xi_{t,i}}_p^{p^*} \leq  M \sum_{t=1}^{T} \sum_{i=1}^{d}\normmm{\xi_{t,i}}_p^{p^*} = M \sum_{t=1}^{T} \sum_{i=1}^{d} \left\{E|\xi_{t,i}|^p\right\}^{p^*/p} \\
		& =  & M \sum_{t=1}^{T} \left\{\sum_{i=1}^{d} \left\{E|\xi_{t,i}|^p\right\}^{p^*/p}\right\}^{p/p^* \times p^*/p} \leq M \sum_{t=1}^{T} \left\{\sum_{i=1}^{d} E|\xi_{t,i}|^p\right\}^{p^*/p} \leq M \sum_{t=1}^{T} \normmm{\bm{\xi}_{t}}_p^{p^*},
	\end{eqnarray*}
	where $\xi_{.,i} = \sum_{t=1}^{T}\xi_{t,i}$. The proof is now completed.
\end{proof}

\begin{proof}[Proof of Lemma \ref{L3}]
	\item
	
	Without loss of generality, let $E(\bm{w}_t^*) = \bm{0}$. For any integer $I \geq 1$ introduce the truncated process $\bm{h}_{t-1,I}^* = E\left(\bm{h}_{t-1}^*|\mathcal{F}_{t-I}\right)$. Then $\bm{h}_{t-1,I}^* = 0$ if $t\leq I$ and $\bm{h}_{t-1,I}^* = \sum_{s=1}^{t-I} w_{s,t}\bm{y}_s$ for $1\leq I < t$. For $2 \leq t\leq T$, by Lemma \ref{L2},
	\begin{equation*}
		\normmm{\bm{h}_{t-1,I}^*-\bm{h}_{t-1}^*}_\delta^2 \leq M \max_t \normmm{\bm{y}_t}_\delta^2 \sum_{s=\max(1,t-I+1)}^{t-1}w_{s,t}^2 = O\left(\sum_{s=\max(1,t-I+1)}^{t-1}w_{s,t}^2\right).
	\end{equation*}
	Let $L(I) = \sum_{J=1}^{I}l(J)$ with $l(J) = \sum_{s=1}^{T-J}w_{s,s+J}^2$, $V(I)=\sum_{t=2}^{T}\mathrm{tr}\left[\bm{w}_t^* \bm{h}_{t-1,I}^*\bm{h}_{t-1,I}^{*,\top}\right]$ and $T(I) = \sum_{t=2}^{T}\mathrm{tr}\left[E(\bm{w}_t^*|\mathcal{F}_{t-I})\bm{h}_{t-1,I}^*\bm{h}_{t-1,I}^{*,\top}\right]$.
	
	By Cauchy-Schwarz inequality, Lemma \ref{L1} (iii), if $I/(Th) \to 0$, we have
	\begin{eqnarray*}
		E|V(1)-V(I)| &\leq& \sum_{t=2}^{T} E\left|\mathrm{tr}\left[\bm{w}_t^*\left(\bm{h}_{t-1}^*\bm{h}_{t-1}^{*,\top}-\bm{h}_{t-1,I}^*\bm{h}_{t-1,I}^{*,\top}\right)\right]\right| \\
		&\leq& \sum_{t=2}^{T}\normmm{\bm{w}_t^*} \normmm{\bm{h}_{t-1}^*-\bm{h}_{t-1,I}^*}_4 \normmm{\bm{h}_{t-1}^*+\bm{h}_{t-1,I}^*}_4 \\
		&\leq& M\sum_{t=2}^{T}\normmm{\bm{h}_{t-1}^*-\bm{h}_{t-1,I}^*}_4 a_t^{1/2}\\
		&\leq& M\left\{\sum_{t=2}^{T}\normmm{\bm{h}_{t-1}^*-\bm{h}_{t-1,I}^*}_4^2\right\}^{1/2} \left\{\sum_{t=2}^{T}a_t\right\}^{1/2}\\
		&=&O(1)[L(I)]^{1/2} \to 0,
	\end{eqnarray*}
	since
	\begin{eqnarray*}
		\normmm{\bm{h}_{t-1}^*+\bm{h}_{t-1,I}^*}_4 &\leq& \left\{M\sum_{s=1}^{t-1}\normmm{w_{s,t}\bm{y}_{s}}_4^2\right\}^{1/2}+\left\{M\sum_{s=1}^{t-I}\normmm{w_{s,t}\bm{y}_{s}}_4^2\right\}^{1/2}\\
		&=&O\left(\left\{\sum_{s=1}^{t-1}w_{s,t}^2\right\}^{1/2} \right)=o(a_t^{1/2}).
	\end{eqnarray*}
	
	Define the projection operator $\mathcal{P}_t\bm{\xi} = E\left(\bm{\xi}|\mathcal{F}_{t}\right)-E\left(\bm{\xi}|\mathcal{F}_{t-1}\right)$. For $0\leq j\leq I-1$, let $U(j) = \sum_{t=2}^{T}\mathrm{tr}\left[\left(\mathcal{P}_{t-j} \bm{w}_t^* \right)\bm{h}_{t-1,I}^*\bm{h}_{t-1,I}^{*,\top}\right]$, then
	$$
	V(I)-T(I)= \sum_{t=2}^{T}\mathrm{tr}\left[\left(\sum_{j=0}^{I-1}\mathcal{P}_{t-j} \bm{w}_t^*\right) \bm{h}_{t-1,I}^*\bm{h}_{t-1,I}^{*,\top}\right] = \sum_{j=0}^{I-1}U(j).
	$$
	Note that $\left\{\left(\mathcal{P}_{t-j} \bm{w}_t^*\right) \bm{h}_{t-1,I}^*\bm{h}_{t-1,I}^{*,\top}\right\}_{t=2}^{T}$ forms a martingale difference sequence since
	$$
	E\left\{\left(\mathcal{P}_{t-j} \bm{w}_t^*\right) \bm{h}_{t-1,I}^*\bm{h}_{t-1,I}^{*,\top}|\mathcal{F}_{t-j-1}\right\} =\left[ E( \bm{w}_t^*|\mathcal{F}_{t-j-1})-E( \bm{w}_t^*|\mathcal{F}_{t-j-1})\right]\bm{h}_{t-1,I}^*\bm{h}_{t-1,I}^{*,\top}=0.
	$$
	
	By Lemma \ref{L1} (ii), Lemma \ref{L2} and Cauchy-Schwarz inequality, since $\normmm{\mathcal{P}_{t-j}\bm{w}_t^*}_{\delta/2} \leq 2 \normmm{\bm{w}_t^*}_{\delta/2} < \infty$,
	\begin{eqnarray*}
		\normmm{U(j)}_{\delta/4}^{\delta/4}&\leq& M \sum_{t=2}^{T}\normmm{\left(\mathcal{P}_{t-j} \bm{w}_t^*\right) \bm{h}_{t-1,I}^*\bm{h}_{t-1,I}^{*,\top}}_{\delta/4}^{\delta/4} \leq M \sum_{t=2}^{T}\normmm{\bm{h}_{t-1,I}^*}_\delta^{\delta/2}\\
		&\leq&M \sum_{t=2}^{T}a_t^{\delta/4} \leq M \max_t a_t^{\delta/4-1} \sum_{t=2}^{T}a_t = O\left(T^{1-\delta/4}\right).
	\end{eqnarray*}
	In addition, by $E|V(1)-V(I)| \to 0$,
	\begin{eqnarray*}
		E|V(1)| &\leq& \normmm{V(I)-T(I)}_{\delta/4}+E|T(I)|+o(1) \\
		&\leq& \sum_{j=0}^{I-1}\normmm{U(j)}_{\delta/4} + \max_t\normmm{E(\bm{w}_t^*|\mathcal{F}_{t-I})} \sum_{t=2}^{T} \normmm{\bm{h}_{t-1,I}^*}_4^2 = o(1),
	\end{eqnarray*}
	since $\max_t\normmm{E(\bm{w}_t^*|\mathcal{F}_{t-I})} \to 0$  as $I \to \infty$. The proof is now completed.
\end{proof}

\begin{proof}[Proof of Lemma \ref{L4}]
	\item
	For notational simplicity, let $\bm{H}_t = \bm{I}_{d^2}$. Write
	\begin{eqnarray*}
		&& \sum_{t=2}^{T}\mathrm{tr}\left[\bm{h}_{t-1}^*\bm{h}_{t-1}^{*,\top}-E\left(\bm{h}_{t-1}^*\bm{h}_{t-1}^{*,\top}\right)\right] \\
		&=&\sum_{t=2}^{T}\sum_{s=1}^{t-1}\mathrm{tr}\left[\left(\bm{Z}_{s-1}\bm{\eta}_s\bm{\eta}_{s}^\top\bm{Z}_{s-1}^\top -E\left(\bm{Z}_{s-1}\bm{\eta}_s\bm{\eta}_{s}^\top\bm{Z}_{s-1}^\top\right)\right)w_{s,t}^2\right]\\
		&&+2\sum_{t=3}^{T}\sum_{s_1=2}^{t-1}\sum_{s_2=1}^{s_1-1}\mathrm{tr}\left[\bm{Z}_{s_1-1}\bm{\eta}_{s_1}\bm{\eta}_{s_2}\bm{Z}_{s_2-1}^\top w_{s_1,t}w_{s_2,t}\right]\\
		&=& I_{T,1} + 2I_{T,2}.
	\end{eqnarray*}
	
	Consider $I_{T,1}$. Write
	\begin{eqnarray*}
		I_{T,1} &=& \sum_{t=2}^{T}\sum_{s=1}^{t-1}\mathrm{tr}\left[\left(\bm{\eta}_s\bm{\eta}_{s}^\top-\bm{\Omega}(\tau_s)\right)\bm{Z}_{s-1}^\top\bm{Z}_{s-1}\right]w_{s,t}^2\\ &&+\sum_{t=2}^{T}\sum_{s=1}^{t-1}\mathrm{tr}\left[\bm{\Omega}(\tau_s)\left(\bm{Z}_{s-1}^\top\bm{Z}_{s-1}-E\left(\bm{Z}_{s-1}^\top\bm{Z}_{s-1}\right)\right)\right]w_{s,t}^2 \\
		&=& \frac{1}{T}\sum_{s=1}^{T-1}\mathrm{tr}\left[\left(\bm{\eta}_s\bm{\eta}_{s}^\top-\bm{\Omega}(\tau_s)\right)\bm{Z}_{s-1}^\top\bm{Z}_{s-1}\right]\left(T\sum_{t=s+1}^{T}w_{s,t}^2\right)\\ 
		&&+ \frac{1}{T}\sum_{s=1}^{T-1}\mathrm{tr}\left[\bm{\Omega}(\tau_s)\left(\bm{Z}_{s-1}^\top\bm{Z}_{s-1}-E\left(\bm{Z}_{s-1}^\top\bm{Z}_{s-1}\right)\right)\right]\left(T\sum_{t=s+1}^{T}w_{s,t}^2\right)  \\
		&=&I_{T,11}+I_{T,12}.
	\end{eqnarray*}
	
	Since $\mathrm{tr}\left[\left(\bm{\eta}_s\bm{\eta}_{s}^\top-\bm{\Omega}(\tau_s)\right)\bm{Z}_{s-1}^\top\bm{Z}_{s-1}\right]$, $s=1,2,...$ are a martingale difference sequence and $T\sum_{t=1}^{T}w_{s,t}^2=O(1)$ by Lemma \ref{L1}.2, we have $I_{T,11} = o_P(1)$. In addition, by Lemma \ref{LemmaB.3}.4, we have $I_{T,12} = o_P(1)$.
	
	Next, consider $I_{T,2}$. By Lemma \ref{L2}, Cauchy-Schwarz inequality and Lemma \ref{L1}.5,
	\begin{eqnarray*}
		\normmm{I_{T,2}}^2 &\leq&M \sum_{s_1=2}^{T-1}\normmm{\mathrm{tr}\left[\bm{Z}_{s_1-1}\bm{\eta}_{s_1}\sum_{s_2=1}^{s_1-1}\bm{\eta}_{s_2}\bm{Z}_{s_2-1}^\top \sum_{t=s_1+1}^{T}w_{s_1,t}w_{s_2,t}\right]}^2 \\
		&\leq& M \sum_{s_1=2}^{T-1}\normmm{\bm{Z}_{s_1-1}\bm{\eta}_{s_1}}_4^2\normmm{\sum_{s_2=1}^{s_1-1}\bm{Z}_{s_2-1}\bm{\eta}_{s_2} \sum_{t=s_1+1}^{T}w_{s_1,t}w_{s_2,t}}_4^2 \\
		&\leq& M \sum_{s_1=2}^{T-1}\normmm{\bm{Z}_{s_1-1}\bm{\eta}_{s_1}}_4^2\sum_{s_2=1}^{s_1-1}\normmm{\bm{Z}_{s_2-1}\bm{\eta}_{s_2} \sum_{t=s_1+1}^{T}w_{s_1,t}w_{s_2,t}}_4^2 \\
		&=&O\left(\sum_{s_1=2}^{T-1}\sum_{s_2=1}^{s_1-1}(\sum_{t=s_1+1}^{T}w_{s_1,t}w_{s_2,t})^2 \right)=O(1/T).
	\end{eqnarray*}
	
	Combine the above results, the proof is now complete.
\end{proof}

\begin{proof}[Proof of Lemma \ref{L5}]
	Note that $\bm{y}_{t}^{\top}\bm{H}_t\bm{h}_{t-1}^*$, $t \in \mathbb{Z}$, are martingale differences with respect  to the filtration $\mathcal{F}_{t}$. We apply Lemma B.1 to prove the asymptotic normality of $Q_T$. By $E\left(\left\|\bm{\epsilon}_t\right\|^\delta|\mathcal{F}_{t-1}\right)<\infty$ a.s., we have
	$$
	E\left[\left\|\bm{y}_t\right\|^\delta \right] \leq E\left[E\left(\left\|\bm{\epsilon}_t\right\|^\delta|\mathcal{F}_{t-1}\right)\left\|\bm{z}_{t-1}\otimes \bm{I}_d\right\|^\delta \right] < E\left[M\left\|\bm{z}_{t-1}\right\|^\delta \right] < \infty.
	$$
	By Cauchy-Schwarz inequality, Lemma \ref{L2} and Lemma \ref{L1} (2), the Lindeberg condition is satisfied since
	\begin{eqnarray*}
		\sum_{t=2}^{T}\normmm{\bm{y}_t\bm{H}_t\bm{h}_{t-1}^*}_{\delta/2}^{\delta/2} &\leq& \sum_{t=2}^{T}\normmm{\bm{y}_t\bm{H}_t}_{\delta}^{\delta/2}\normmm{\bm{h}_{t-1}^*}_{\delta}^{\delta/2} \\
		&\leq& M \max_t \normmm{\bm{y}_t}_{\delta}^{\delta}\sum_{t=2}^{T}a_t^{\delta/4}  \\
		&=&O(1)\cdot \max_{t}a_t^{\delta/4-1}=o(1).
	\end{eqnarray*}
	Apply Lemmas \ref{L3} and \ref{L4} with $\bm{w}_t^* = E\left(\bm{H}_t^\top\bm{y}_t\bm{y}_{t}^{\top}\bm{H}_t|\mathcal{F}_{t-1}\right)$, then we have the convergence of conditional variance
	\begin{eqnarray*}
		&&\sum_{t=2}^{T}\mathrm{tr}\left[E\left(\bm{H}_t^\top\bm{y}_t\bm{y}_{t}^{\top}\bm{H}_t|\mathcal{F}_{t-1}\right)\bm{h}_{t-1}^*\bm{h}_{t-1}^{*,\top} \right]\\
		&\to_P&\sum_{t=2}^{T}\mathrm{tr}\left[E\left(\bm{H}_t^\top\bm{y}_t\bm{y}_{t}^{\top}\bm{H}_t\right)E\left(\bm{h}_{t-1}^*\bm{h}_{t-1}^{*,\top} \right)\right].
	\end{eqnarray*}
	
	By Lemma B.1, the proof is now completed.
	
\end{proof}

{\footnotesize
\bibliography{ma}
}
}

\end{document}